\documentclass{article} 
\usepackage{iclr2027_conference,times}

\definecolor{orangetag}{rgb}{1.0, 0.55, 0.0}

\newcommand{\Cvxa}{\mathcal{A}}
\newcommand{\Cvxb}{\mathcal{B}}

\newcommand{\DD}{{\rm D}}
\newcommand{\EE}{{\rm E}}

\newcommand{\NN}{{\rm N}}
\newcommand{\MM}{{\rm M}}

\newcommand{\KK}{{\rm K}}
\newcommand{\LL}{{\rm L}}

\newcommand{\Reg}{\mathcal{R}}
\newcommand{\OT}{{\rm OT}}
\newcommand{\QM}{{\rm QM}}

\newcommand{\diff}{{\rm d}}

\newcommand{\N}{\mathbb{N}}

\newcommand{\one}{\mathbf{1}}

\newcommand{\meas}[1]{\mathcal{M}\left( #1 \right)}
\newcommand{\measplus}[1]{\mathcal{M}_+\left( #1 \right)}

\newcommand{\coupling}[1]{\Pi\left( #1 \right)}

\newcommand{\Engy}{\mathcal{E}}
\newcommand{\Loss}{\mathcal{L}}
\newcommand{\Corr}{\rm Corr}

\newcommand{\C}{\mathcal{C}}
\newcommand{\T}{\mathcal{T}}

\newcommand{\Hilb}{\mathcal{H}}

\newcommand{\img}[1]{{\rm Im}\left({#1}\right)}

\newcommand{\TV}{{\rm TV}}
\newcommand{\Div}[3]{{\rm D}_{#1}\left(#2 \Vert #3 \right)}

\newcommand{\scal}[1]{\left\langle #1 \right\rangle}
\newcommand{\norm}[1]{\left\Vert #1 \right\Vert}
\newcommand{\abs}[1]{\left\vert #1 \right\vert}

\newcommand{\HS}{{\rm HS}}

\newcommand{\X}{\mathcal{X}}
\newcommand{\Y}{\mathcal{Y}}

\usepackage{amsthm} 

\newtheorem{theorem}{Theorem}
\newtheorem{proposition}{Proposition}
\newtheorem{corollary}{Corollary}
\newtheorem{lemma}{Lemma}
\newtheorem{remark}{Remark}
\newtheorem{definition}{Definition}

\usepackage{amsmath,amsfonts,bm}

\def\eqref#1{equation~\ref{#1}}

\def\1{\bm{1}}

\DeclareMathAlphabet{\mathsfit}{\encodingdefault}{\sfdefault}{m}{sl}
\SetMathAlphabet{\mathsfit}{bold}{\encodingdefault}{\sfdefault}{bx}{n}

\newcommand{\R}{\mathbb{R}}

\newcommand{\KL}{D_{\mathrm{KL}}}

\usepackage[ruled,vlined]{algorithm2e}
\DontPrintSemicolon

\usepackage{hyperref}
\usepackage{url}

\usepackage{cleveref}

\usepackage{thmtools}
\usepackage{thm-restate}

\usepackage{graphicx}

\usepackage{booktabs}
\usepackage{tabularx}
\newcolumntype{R}{>{\raggedleft\arraybackslash}X}
\newcolumntype{C}{>{\centering\arraybackslash}X}

\title{A Unified Dual Method for Matching Problems}

\author{Guillaume Houry \\
HeKA team \\
Inria, Université Paris-Cit\'e, Inserm\\
F-75015 Paris, France \\
\texttt{guillaume.houry@inria.fr} \\
\And
Ferdinand Genans \\
LPSM \\
Sorbonne Université \\
F-75005 Paris, France \\
\texttt{genans.ferdinand@gmail.com} \\
\AND
Jean Feydy \\
HeKA team \\
Inria, Université Paris-Cit\'e, Inserm\\
F-75015 Paris, France \\
\texttt{jean.feydy@inria.fr} \\
\And
François-Xavier Vialard \\
LIGM \\
Universit\'e Gustave Eiffel \\
Marne-la-Vall\'ee, France \\
\texttt{francois-xavier.vialard@univ-eiffel.fr} \\
}

\iclrfinalcopy 
\begin{document}

\maketitle

\begin{abstract}
Matching problems are ubiquitous in data science as they enable the alignment of structured objects and distributions.
While existing solvers are often tailored to specific matching formulations, we unify a broad class of such problems within a common mathematical and optimization framework based on duality theory. 
Theoretically, we demonstrate that matching objectives decomposable as a difference of convex (DC) functions can be recast as implicit registration problems. This connection links matching to another well-studied class of objectives and yields a dual formulation amenable to natural optimization strategies.
We then apply these findings to quadratic matching (QM) problems, which admit DC decompositions and for which we provide extensive convergence guarantees.
Our framework applies to Gromov-Wasserstein (GW), as well as its unbalanced formulation and several variants, which are increasingly popular QM problems. Numerically, we implement our algorithms at scale for various data modalities such as graphs, point clouds, meshes, and word embeddings.
Finally, our modular approach allows us to explore new formulations such as fracture matching, broadening the scope of problems that can be addressed within this framework.
\end{abstract}

\section{Introduction}
Many objects studied in data science and imaging, such as point clouds, shapes, graphs, physical variables, or probability distributions, are naturally spread across space.
Comparing, aligning, and transferring information between such objects requires to establish reliable spatial correspondences between such inputs.
To this end, two approaches are available: registration, which searches for optimal transforms of one source object to a target one, and matching, which draws pointwise correspondences between their constituent points.
Although the former class of methods provides more regular output, it is also less flexible and is not applicable to all data modalities.
The latter, on the other hand, is generally harder to solve algorithmically, and the lack of efficient solvers for arbitrary problems is still an important bottleneck for many downstream applications. 

Arguably the most popular category of matching techniques, quadratic matching searches for matchings as solutions of a quadratic optimization problem, allowing the use of principled numeric routines while being sufficiently expressive to encode complex geometric objectives.
Depending on the exact problem and the scientific community, quadratic matching takes different names: quadratic assignment  \citep{lawler1963quadratic}, Gromov--Wasserstein \citep{memoli2011gromov}, graph matching \citep{livi2013graph}, or quadratic-form optimal transport \citep{wang2025quadratic}.
More custom QM instances have also been introduced for shape matching \citep{mandad2017variance,brifault2025efficient}, evidencing the importance of this class of problems. 

Yet, quadratic matching remains challenging to compute.
Its non-convexity makes it NP-hard to minimize globally. 
Local optimization via first-order methods is possible \citep{frank1956algorithm}, but it requires to manipulate large matching matrices explicitly which limits the scalability of this approach.
A promising line of work exploit duality properties of quadratic matching to recast quadratic matching problems as a sequence of linear matching instances, or optimal transport (OT).
OT is a popular and well-studied tool in data science, its success coming from the availability of efficient solvers \citep{peyre2019computational}.
Indeed, OT admits a dual formulation that implicitly encodes matching matrices in variables of lower size, removing its main computational bottleneck.
OT is also a flexible framework, since regularizations can be added to the linear optimization problems to smooth or sparsify the final matching \citep{blondel2018smooth}, to ignore outliers \citep{chapel2020partial,liero2018optimal}, or to enable fast and GPU-compatible algorithms \citep{cuturi2013sinkhorn}. 
Building upon this connexion between quadratic matching and optimal transport, faster QM methods have been proposed in specific cases \citep{gold1996graduated,peyre2016gromov,sejourne2023unbalanced}, and duality properties of some QM instances have been recently investigated \citep{rioux2024entropic,zhang2024gromov,sebbouh2024structured,houry2026gromov}.
However, these works still lack the generality required to solve any QM problem at scale.

\paragraph{Contribution.} Focusing on the current limitations of the matching literature, we propose a new framework to solve any quadratic matching locally in an efficient and rigorous way.
We first draw a conceptual link between matching and registration problems, revealing a hidden geometric structure in matching problems, and introduce a dual formula that reduce arbitrary matching problems to a sequence of OT computations.
We then specialize these results to the quadratic case, allowing us to build an efficient solver with quantitative convergence guarantees that we apply in various practical settings.
Beyond the methodological contributions, we release an optimized, backend-agnostic implementation that is modular by design and readily extensible to other regularized quadratic problems. While unbalanced Gromov-Wasserstein serves as our primary motivating application throughout the paper, the framework applies with minimal adaptation to a wider class of problems.

\paragraph{Notations.}
In this article, $\X$ and $\Y$ are two compact measurable metric spaces, $\C(\X)$ denotes the space of continuous real-valued functions on $\X$, $\meas{\X}$ is the space of measures of $\X$, and $\measplus{\X}$ is its restriction to positive measures.
We recall that in the discrete setting, when $\X = \{x_1,\dots,x_\NN\}$ and $\Y = \{y_1,\dots,y_\MM\}$ are finite point sets, any measures $\alpha \in \meas{\X}$ and $\beta \in \meas{\Y}$ correspond to weight vectors of $\R^\NN$ and $\R^\MM$ and measures $\pi$ of $\meas{\X \times \Y}$ correspond to matrices of $\R^{\NN \times \MM}$:
\begin{equation*}
    \alpha = \sum_{i=1}^{\NN} \alpha_i \delta_{x_i}, \quad  \beta = \sum_{j=1}^{\MM} \beta_j \delta_{y_j} \quad  \text{ and } \quad \pi_{ij} = \sum_{i,j} \pi_{ij} \delta_{(x_i,y_j)}.
\end{equation*}
Let $A : \X \times \Y \to \Hilb_A$ be a continuous function with values in a Hilbert space $\Hilb_A$.
It naturally acts on $\meas{\X \times \Y}$ via its associated \emph{integral operator}, defined as follows:
\begin{equation*}
    A : \pi \in \meas{\X \times \Y} \;\longmapsto\; A[\pi] := \int A(x,y)\,\diff\pi(x,y) \in \Hilb_A.
\end{equation*}
In the discrete case and when $\Hilb_A = \R^{\DD}$ is Euclidean, it becomes $[ \sum_{ij} A(x_i,y_j)_d \pi_{ij}]_d \in  \R^{\DD}$. 
Given two Hilbert vectors $x \in \Hilb_\X$ and $y \in \Hilb_\Y$, their \emph{tensor product} is the rank-$1$ linear operator:
\begin{equation*}
    x y^{\top}: z \in \Hilb_\Y \longmapsto \scal{y,z}_{\Hilb_\Y} x~~\in \Hilb_\X.
\end{equation*}
The limit of sums of such operators forms a Hilbert space, the \emph{Hilbert-Schmidt space} $\HS(\Hilb_\Y, \Hilb_\X)$.
In the Euclidean case, the tensor product of two vectors $x \in \R^\DD$ and $y \in \R^\EE$ is $x y^\top\!=\![x_d y_e]_{de}\!\in\!\R^{\DD \times \EE}$, and $\HS(\R^\DD, \R^\EE)$ identifies with $\R^{\DD \times \EE}$ equipped with the Frobenius norm.
Finally, given a convex function $ \mathcal{F} : \meas{\X \times \Y} \to \R$, its \emph{convex conjugate} $\mathcal{F}^*$ is defined as:
\begin{equation*}
    \mathcal{F}^*: c \in \C(\X \times \Y) \longmapsto \sup_{\pi \in \meas{\X \times \Y}} \left( \int c(x,y)\diff\pi(x,y) - \mathcal{F}(\pi) \right).
\end{equation*}

\section{Measure Matching as an Implicit Registration Problem}
\label{section:measure_matching}

To model both finite and continuous inputs in a common mathematical framework, our work focuses on matching from a measure-theoretical perspective.
Given a source $\alpha \in \measplus{\X}$ and a target $\beta \in \measplus{\Y}$, we represent matchings from $\alpha$ to $\beta$ as elements $\pi$ of $\measplus{\X \times \Y}$ -- which we call \emph{transport plans} by analogy with optimal transport.
Each $\pi(x,y)$ indicates the amount of $\alpha(x)$ that is sent to $\beta(y)$.
This framework is naturally compatible with matching relaxation methods where one point of $\alpha$ can be sent to several points of $\beta$, and conversely: one-to-one matching corresponds plans $\pi$ that describe the graph of a bijective function of $\X \to \Y$, or, in the finite setting, when $\pi$ is a permutation matrix.  
Despite the variety of existing methods, most matching algorithms can be represented as variational problems in this measure setting, taking the form:
\begin{equation}
\label{eq:general_matching}
  {\rm M}(\alpha, \beta) ~~=~~  \min_{\pi \in \measplus{\X \times \Y}} \;\Engy_\Reg(\pi) \qquad \text{with} \qquad \Engy_\Reg(\pi) := \Engy(\pi) + \Reg_{\alpha,\beta}(\pi).
\end{equation}
The continuous energy term $\Engy$ encodes the objective of the matching problem; an ideal matching should reach $\Engy(\pi) = 0$. 
The convex constraint term $\Reg_{\alpha,\beta}$ forces the matching to \emph{send} $\alpha$ to $\beta$ and may contain data-specific regularizations; when there is no ambiguity, we will keep the dependency on $\alpha, \beta$ implicit and simply write $\Reg := \Reg_{\alpha,\beta}$.
To clarify this formalism, we use Gromov--Wasserstein (GW) as a running example throughout our article. 
GW promotes isometric matching by penalizing distortions between matched pairs. In the finite case, these distorsions are quantified by symmetric matrices $C^\X$ and $C^\Y$ so that GW reads:
\begin{equation*}
    \min_{\pi \in \R_+^{\NN \times \MM}} \frac{1}{2} \sum_{ijkl} (C^\X_{ik}- C^\Y_{jl})^2 \pi_{ij} \pi_{kl} \quad \text{ s.t. } \quad \sum_j \pi_{ij} = \alpha_i, \quad \text{ and } \quad \sum_i \pi_{ij} = \beta_j.
\end{equation*}
After encoding the constraints in the term $\Reg$, GW takes the form of \cref{eq:general_matching} with:
\begin{equation*}
    \Engy(\pi) = \frac{1}{2} \sum_{ijkl} (C^\X_{ik}- C^\Y_{jl})^2 \pi_{ij} \pi_{kl}  
    \quad \text{ and } \quad
    \Reg(\pi) = 
    \begin{cases} 
    0 & \text{if } \sum_j \pi_{ij} = \alpha_i \text{ and }  \sum_i \pi_{ij} = \beta_j \\
    +\infty & \text{otherwise}.
    \end{cases}
\end{equation*}
In the continuous case, we replace matrix operations with their measure-based counterparts:
\begin{multline*}
    \Engy(\pi) = \frac{1}{2} \int \bigl(c_\X(x,x') - c_\Y(y,y')\bigr)^2 \diff\pi(x,y) \diff\pi(x',y')  \\
\text{ and } \quad
    \Reg(\pi) =
    \begin{cases} 
    0 & \text{if } \pi_1 := \int \diff\pi(x,\cdot) = \diff\alpha(x) \text{ and }  \pi_2 := \int \diff\pi(\cdot,y) =  \diff\beta(y) \\
    +\infty & \text{otherwise},
    \end{cases}
\end{multline*}
where $c_\X, c_\Y$ are continuous symmetric functions.
Classical GW variants add regularizations to $\Reg$ or relax the margins constraints: we give typical examples in \cref{appendix:regularizations}.
Notably, entropic GW adds a KL-divergence term to $\Reg$, and unbalanced GW relaxes the marginals using: 
\begin{equation}
\label{eq:unbalanced_reg}
\Reg(\pi) = \rho \KL(\pi_1, \Vert \alpha) + \rho \KL(\pi_2 \Vert \beta) + \varepsilon \KL(\pi \Vert \alpha \otimes \beta), \quad \rho, \varepsilon \geq0. 
\end{equation}
Matching problems are generally hard to solve numerically.
They are usually non-convex, making global solutions difficult to compute.
Besides, their objectives are defined on the set of plans $\pi \in \measplus{\X \times \Y}$.
In the finite case, these are encoded by large matrices, making even local minimization expensive to perform.
In specific cases, recent works have alleviated these issues by reformulating matching objectives into easier problems  \citep{rioux2024entropic,zhang2024gromov,sebbouh2024structured,houry2026gromov}.
We generalize these works to arbitrary matching problems, focusing on the cases where $\Engy$ is decomposed into convex and concave terms:
\begin{equation}
\label{eq:dc_matching}
    \mathcal{E}(\pi) = \Cvxa(\pi) - \Cvxb(\pi),
\end{equation}
where $\Cvxa, \Cvxb$ are proper lower semi-continuous convex functions and continuous on their domains.
This structure is known as difference of convex (DC).
Under mild assumptions, any matching problem is approximable by a DC one, and in the finite case, exact DC decompositions always exist.
Without loss of generality, we can also assume the convex conjugates $\Cvxa^*$, $\Cvxb^*$ to be continuous on $\C(\X \times \Y)$; we formalize these properties in  \cref{appendix:general_matching}.
The DC structure isolates the hard part of the problem -- the concave term -- so we can handle it separately, allowing us to prove the equivalence:

\begin{restatable}{theorem}{thmccvdualitygeneral} 
\label{theorem:ccv_duality_general}
If $\Engy$ admits the structure of \cref{eq:dc_matching}, then \cref{eq:general_matching} is equivalent to:
    \begin{equation}
    \label{eq:general_matching_dualccv}
    {\rm M}(\alpha, \beta) ~~=~~
    \inf_{c_\Cvxb \in \C(\X \times \Y)} ~ \min_{\pi \in \measplus{\X \times \Y}}\; -\int c_\Cvxb(x,y) \diff \pi(x,y) + \mathbf{R}(\pi) + \mathbf{S}(c_\Cvxb),
    \end{equation}
    where $\mathbf{R}(\pi)= \Cvxa + \Reg$ and $\mathbf{S} =  \Cvxb^*$ are convex functions.
If $\X$ and $\Y$ are compact subsets of Euclidean spaces, then there exists two Hilbert-valued continuous embeddings $\Phi : \X \to \Hilb$ and $\Psi: \Y \to \Hilb$, a convex set $U \subset \HS(\Hilb,\Hilb)$ and a convex function $\mathbf{S}'$ such that \cref{eq:general_matching} becomes:
\begin{equation}
\label{eq:general_matching_registrform}
  {\rm M}(\alpha, \beta) ~~=~~
  \inf_{\Gamma \in U} \min_{\pi \in \measplus{\X \times \Y}} - \int \scal{\Phi(x), \Gamma \Psi(y)} \diff\pi(x,y) + \mathbf{R}(\pi) + \mathbf{S}'(\Gamma).
    \end{equation}
\end{restatable}

This structure is identical to another class of methods seeking correspondences between geomeric objects called \emph{registration}.
The principle of registration is to find a deformation of the source shape $\beta$ that minimizes its spatial discrepancy relative to the target shape $\alpha$.  
Given a space of transformations $\T$ and a real-valued loss function $\Loss$, standard registration problems are phrased as:
\begin{equation}
\label{eq:registration_pb}
    \min_{\Gamma \in \T} \min_{\pi \in \meas{\X \times \Y}} \frac{1}{2} \int \Loss(x, \Gamma(y)) \diff\pi(x,y) + \mathbf{R}(\pi) + \mathbf{S}(\Gamma),
\end{equation}
where $\mathbf{R}$ is convex and $\mathbf{S}$ is an arbitrary regularizer.
Many popular algorithms follow this structure, such as iterative closest points \citep{besl1992method}, robust point matching \citep{gold1994new}, coherent point drift \citep{myronenko2010point}, correlation-based registration \citep{tsin2004correlation}, functional maps \citep{ovsjanikov2012functional}, Procrustes-Wasserstein \citep{grave2019unsupervised} and other OT-based methods \citep{alvarez2019towards}.
Registration is also non-convex, but non-convexity is isolated within the outer minimization over $\Gamma$: it provides a cleaner framework for analyzing and exploring optimization landscapes.
Registration objectives are also more robust: unlike plans $\pi$, optimal transformations $\Gamma$ depend only on the global input geometry, making them intrinsically robust to empirical noise, point resampling, or permutations, and opening the door to geometrically-driven heuristics.
Thanks to our reformulation, we can exploit these properties and adapt the numerous existing tools for registration to the matching context. 

Our reformulation also reveals a nested minimization structure whose local solutions can be obtained by alternating optimization over $\Gamma$ and $\pi$.
Since the problem is convex in $\pi$, this parameter can be updated using any off-the-shelf solver.
Yet, this approach still requires explicit optimization over the plan $\pi$.
In classical registration, the regularization $\mathbf{R}$ is typically designed to make this optimization cheap.
However, in our setting, $\mathbf{R}$ is prescribed by the original objective $\Engy$ so we need a method that handles arbitrary regularizations on $\pi$.
Taking inspiration from optimal transport, we adapt duality results for DC functions \citep{toland1978duality} to avoid manipulating $\pi$ explicitly:
\begin{restatable}{theorem}{tolanddualitygeneral}
\label{theorem:tolanddualitygeneral}
Let us assume that $\Reg$ is a proper lower semi-continuous convex function.
Let $\Phi$ be its convex set of dual variables and $\mathcal{D}_\Reg$ its dual objective that satisfies, for all $c \in \C(\X \times \Y)$:
\begin{equation}
\label{eq:dualizable_reg}
    \min_{\pi \in \measplus{\X \times \Y}} \int c\diff\pi + \Reg(\pi) = \sup_{\phi \in \Phi}\; \mathcal{D}_\Reg(\phi;\, c, \alpha, \beta).
\end{equation}
The matching problem of \cref{eq:dc_matching} is then equivalent to the following equations:
\begin{align}
\label{eq:toland_dual}
    \inf_{c_\Cvxb \in \C(\X \times \Y)}~~ \sup_{c_\Cvxa \in \C(\X \times \Y)} ~~ \sup_{\phi \in \Phi} ~~ 
            \mathcal{D}_\Reg\bigl(\phi;\, c_\Cvxa - c_\Cvxb,\, \alpha, \beta\bigr) + \Cvxb^*(c_\Cvxb) - \Cvxa^*(c_\Cvxa), \\
\label{eq:toland_primal}
       \inf_{c_\Cvxb \in \C(\X \times \Y)}~~ \sup_{c_\Cvxa \in \C(\X \times \Y)} ~~ \min_{\pi \in \measplus{\X \times \Y}} ~~ 
            \int (c_\Cvxa - c_\Cvxb) \diff\pi  + \Reg(\pi) + \Cvxb^*(c_\Cvxb) - \Cvxa^*(c_\Cvxa). 
\end{align}
If $\X$ and $\Y$ are compact subsets of Euclidean spaces, then there exists a sequence of continuous Hilbert embeddings $\Phi_\Cvxa, \Phi_\Cvxb : \X \to \Hilb$ and $\Psi_\Cvxa, \Psi_\Cvxb: \Y \to \Hilb$ such that these equations become:
\begin{equation}
\label{eq:general_matching_registrform_toland}
\begin{gathered}
{\rm M}(\alpha, \beta)  ~~ = ~~ \inf_{\Gamma \in U} ~~  \sup_{\Xi \in U} ~~ \sup_{\phi \in \Phi} ~~  \mathcal{G}(\Gamma, \Xi, \phi)
~~ = ~~
     \inf_{\Gamma \in U}~~ \sup_{\Xi \in U} ~~ \min_{\pi \in \measplus{\X \times \Y}} ~~ 
     \mathcal{F}(\Gamma, \Xi, \phi), \\
 \text{with } \quad \mathcal{G}(\Gamma, \Xi, \phi) :=
            \mathcal{D}_\Reg\bigl(\phi;\, c_{\Gamma, \Xi},\, \alpha, \beta\bigr) + \mathbf{S}'(\Gamma) - \mathbf{S}''(\Xi)
\\
\text{and } \quad \mathcal{F}(\Gamma, \Xi, \phi) := \int c_{\Gamma, \Xi} \diff\pi + \Reg(\pi) + \mathbf{S}'(\Gamma) - \mathbf{S}''(\Xi),
\end{gathered}
\end{equation}
where $U \subset \HS(\Hilb,\Hilb)$ is a convex set, $\mathbf{S}'$ and $\mathbf{S}''$ are convex functions and $c_{\Gamma, \Xi}$ is the following effective cost (where the matrix notation indicates the direct sum, or concatenation, of its rows):
\begin{equation}
\label{eq:effective_cost_decomp}
    c_{\Gamma,\Xi}(x,y) = - \scal{\begin{pmatrix}
        \Phi_\Cvxa(x) \\ \Phi_\Cvxb(x) \end{pmatrix},\begin{pmatrix}
        -\Xi \Psi_\Cvxa(y) \\ \Gamma \Psi_\Cvxb(y) \end{pmatrix}}.
\end{equation}
\end{restatable}

For classical choices of $\Reg$, either $\Reg$-regularized linear problems ($\min_\pi \int c\diff\pi + \Reg(\pi)$) or their duals ($\sup_\phi \mathcal{D}_\Reg (\phi;\, c,\, \alpha, \beta)$) admit efficient solvers for any $c \in \C(\X \times \Y)$.
We refer to these problems as $\Reg\OT$, as they have the structure of a $\Reg$-regularized optimal transport.
If efficient $\Reg\OT$ solvers are available, \cref{theorem:tolanddualitygeneral} provides a strategy to solve the original matching efficiently. 
As a counterpart, the embedding space on which matching is performed is now infinite-dimensional and structurally complex: our reformulation transfers the difficulty of matching problems into the embeddings $\Phi_\Cvxa$, $\Phi_\Cvxb$, $\Psi_\Cvxa$ and $\Psi_\Cvxb$, the problem structure being now agnostic to the matching objective $\Engy$.

\section{An Explicit Dual for Regularized Quadratic Matching}
\label{section:quadratic_matching_case}

To convert the results of the previous section into practical numerical algorithms, we need to explicitly have access to the embeddings $\Phi_\Cvxa$, $\Phi_\Cvxb$, $\Psi_\Cvxa$ and $\Psi_\Cvxb$ which are generally unknown and high-dimensional, limiting the applicability of our method to arbitrary matching problems.
However, in the quadratic matching (QM), the DC decomposition is often simple and explicit, bypassing these limitations.
In QM, objectives are parametrized by a symmetric kernel $k \in \C((\X\times \Y)^2)$: 
\begin{equation}
    \label{eq:rgw_def}
    \Engy(\pi) = \frac{1}{2} \int k(x,y,x',y')  \diff\pi(x,y) \diff \pi(x',y').
\end{equation}
When $\Engy$ follows this structure, we denote by $\Reg\QM(\alpha,\beta)$ the $\Reg$-regularized quadratic matching of \cref{eq:general_matching}.
For instance, Gromov-Wasserstein is a QM with $k(x,y,x',y') =  (c_\X(x,x') - c_\Y(y,y'))^2$.

In the finite setting and several continuous cases that we detail in \cref{appendix:general_matching}, $k$ admits an explicit decomposition as a difference of quadratic functions.
We refer to such $k$ as factorizable kernels:
\begin{definition}[Factorizable kernel]
We say that $k$ is factorizable if its objective is decomposable as:
\begin{equation}
\label{eq:loss_factorized}
    \Engy(\pi) = \frac{1}{2} \norm{A[\pi]}^2 - \frac{1}{2} \norm{B[\pi]}^2,
\end{equation}
where $A : \meas{\X \times \Y} \to \Hilb_A$ and $B : \meas{\X \times \Y} \to \Hilb_B$ are integral operators.
\end{definition}

Under this hypothesis, we directly adapt the formulas of the previous section to the $\Reg\QM$ case:

\begin{restatable}{theorem}{outermin}
    \label{prop:outer_gamma}
    Let assume that $k$ is factorizable as in \cref{eq:dc_matching}. Then, the following identity holds:
    \begin{equation}
        \label{eq:outer_gamma}
        \Reg\QM(\alpha,\beta) = \min_{\Gamma \in \img{B}} \min_{\pi \in \measplus{\X \times \Y}}\; -\int \scal{\Gamma, B(x,y)} \diff \pi(x,y) + \mathbf{R}(\pi) + \mathbf{S}(\Gamma),
    \end{equation}
    with $\mathbf{R}(\pi) = \tfrac{1}{2}\norm{A [\pi]}^2 + \Reg(\pi)$ and $\mathbf{S}(\Gamma) = \tfrac{1}{2} \norm{\Gamma}^2$.
Moreover:
    \begin{equation}
            \label{eq:quad_minmax}
    \Reg\QM(\alpha,\beta) = \min_{\Gamma \in \img{B}}\;
            \max_{\Xi \in \img{A}}\;
            \sup_{\phi \in \Phi}\;
            \mathcal{D}_\Reg\bigl(\phi;\, c_{\Gamma,\Xi},\, \alpha, \beta\bigr) - \tfrac{1}{2}\norm{\Xi}^2 + \tfrac{1}{2} \norm{\Gamma}^2,
    \end{equation}
    with the effective cost $c_{\Gamma,\Xi}(x,y) := \scal{\Xi, A(x,y)} - \scal{\Gamma, B(x,y)}$.
\end{restatable}
To recover the registration-like structure of \cref{eq:general_matching_registrform}, we must also assume that $B$ is block-separable:
\begin{definition}[Separable function]
\label{definition:separable}
A function $B : \X \times \Y \to \Hilb_B$ is separable if there exist two embedding functions $\Phi$ and $\Psi$ such that:
\begin{equation*}
    \text{For all } x \in \X \text{ and } y \in \Y, \quad B(x,y) = \Phi(x) \Psi(y)^\top.
\end{equation*}
We say that $B$ is block-separable if there is a collection of separable functions $B_1,\dots,B_\LL$ such that $\quad \norm{B[\pi]}^2 = \norm{B_1[\pi]}^2 + \dots + \norm{B_\LL[\pi]}^2$ for all $\pi \in \meas{\X \times \Y}$.
\end{definition}
Then, $\Reg\QM$ admits the following form (we give a detailed version in the appendix, in \cref{thm:qot_as_registration_blockseparable}):
\begin{restatable}{theorem}{qotasregistration}
\label{thm:qot_as_registration}
If $B$ is block-separable, then let $\mathbf{R}$ and $\mathbf{S}$ be as in \cref{prop:outer_gamma}.
There exists embeddings $\Phi_\Cvxb: \X \to \Hilb_\Phi$, $\Psi_\Cvxb: \Y \to \Hilb_\Psi$ and a compact subspace $U \subset \HS(\Hilb_\Phi,\Hilb_\Psi)$ such that:
\begin{equation}
\label{eq:quadmatching_as_registration}
       \Reg\QM(\alpha,\beta) = \min_{\Gamma \in U} \min_{\pi \in \measplus{\X \times \Y}}\; - \int \scal{\Phi_\Cvxb(x), \Gamma \Psi_\Cvxb(y)} \diff \pi(x,y) + \mathbf{R}(\pi) + \mathbf{S}(\Gamma).
\end{equation}
If both $A$ and $B$ are block-separable, there exists embeddings $\Phi_\Cvxa$,$\Phi_\Cvxb$,$\Psi_\Cvxa$ and $\Phi_\Cvxb$ such that the effective cost $c_{\Gamma,\Xi}$ in \cref{eq:general_matching_registrform_toland} takes the scalar product form of \cref{eq:effective_cost_decomp} and $\mathbf{S}',\mathbf{S}'' = \tfrac{1}{2} \norm{\cdot}^2$.
\end{restatable}

If we have access to a block-sparse decomposition of the kernel $k$, then the embeddings are explicitly computable and \cref{prop:outer_gamma,thm:qot_as_registration} provide a computation-compatible version of \cref{theorem:tolanddualitygeneral}.
In the finite case, this can be done by eigendecomposition of $k$.
For GW kernels, explicit computations that we detail in \cref{appendix:gw_decomposition} provide block embeddings directly without the need to diagonalize the whole kernel $k$, making their computation tractable on large inputs.
Moreover, these embeddings are now related to the images of $A$ and $B$: when these operators are low-rank, which is generally the case in practice, the optimization problem becomes low-dimensional and easier to solve.

\section{Numerical Optimization of the $\Reg\QM$ Dual}
\label{section:numerical_optim}
We now exploit these theoretical results to develop a unified algorithm for $\Reg\QM$ with quantitative convergence guarantees.
Up to a factorization of the kernel $k$, always possible in the finite case, our method applies to any QM instance and most standard choices of $\Reg$, making the proposed framework applicable in a broad range of settings.
The algorithm proceeds by alternate minimization over the block of variables of \cref{eq:general_matching_registrform_toland}: $\Gamma$, $\Xi$ and either $\Phi$ or $\pi$.
Since our method is compatible with both primal and dual $\Reg\OT$ solvers, we can choose whether to optimize over $\Phi$ or $\pi$ depending on which formulation is computationally more efficient.
Each block is updated as follows:

\paragraph{$\phi$/$\pi$-block: $\Reg\OT$ solve.} For fixed $(\Gamma, \Xi)$, maximizing $\mathcal{F}$ over $\phi$ or $\mathcal{G}$ over $\pi$ amounts to solving a primal or dual $\Reg\OT$ problem with cost $c_{\Gamma,\Xi}$.
For classical $\Reg$, any compatible solver can be applied.

\paragraph{$\Xi$-block: gradient ascent.} When $\Gamma$ is fixed, the function $h_\Gamma: \Xi \mapsto \min_{\pi} \mathcal{G}(\Gamma, \Xi, \pi)$ is $1$-strongly concave.
If the optimal $\pi_{\Gamma,\Xi}$ is unique, then $h_\Gamma$ is differentiable at $\Xi$ with $\nabla h_\Gamma(\Xi) = A[\pi_{\Gamma,\Xi}] - \Xi$.
The block $\Xi$ can then be optimized by alternating $\Reg\OT$ solves and gradient ascent steps.

\paragraph{$\Gamma$-block: closed form.} Finally, if both $\Xi$ and $\pi$ are fixed, then \cref{eq:outer_gamma} shows that $\Gamma \mapsto \mathcal{G}(\Gamma, \Xi, \pi)$ is minimized at $\Gamma = B[\pi]$: the $\Gamma$-block is updated using this closed-form minimum.

By design, for well-chosen gradient steps in the $\Xi$-block, each update improves the current block objective value: this approach is guaranteed to converge.
In practice, however, many $\Reg\OT$ solvers only provide an approximation of the true solution $(\phi, \pi)$, and gradient ascent on $\Xi$ never reaches its true maximum in finite time: we must set tolerance thresholds to indicate when these computations should be stopped, yielding inexact block updates. 
Our alternating minimization scheme is summarized in \cref{alg:ugw}. Despite the inexact updates, the stopping criterion imposed on $\Xi$ still guarantees finite-time convergence under appropriate assumptions.
We state this result informally here, providing full statements and detailed estimates in \cref{appendix:end-to-end-cv}.

\begin{theorem}[Convergence of \cref{alg:ugw}, informal] \label{thm:main_convergence}
Let $0 < \delta \leq 1$.
\textbf{(i)} If the inner loops converge, then \cref{alg:ugw} terminates after $O(\delta^{-2})$ outer iterations.
\textbf{(ii)} If the function $h_\Gamma$ is $L$-smooth for all $\Gamma$, then the choice $\tau = 1/L$ and $\delta_{\Reg\OT} = \Theta(\delta^2/L^2)$ ensures that each inner loop converges in $\widetilde{O}(L)$ iterations, for a total of  $\widetilde{O}(L \delta^{-2})$ cumulated $\Reg\OT$ solvings steps.
\end{theorem}

\paragraph{Characterization of output plans.} 
Theorem~\ref{thm:main_convergence} ensures that our algorithm terminates but does not qualify the properties of the final plan $\widehat{\pi}$.
By non-convexity of $\Reg\QM$, we cannot expect $\widehat{\pi}$ to minimize the problem globally: as usual in DC programming \citep{le2018dc}, we will rather show that $\widehat{\pi}$ is almost critical.
A plan $\pi^*$ is \emph{DC-critical} if it is optimal after linearization of the concave term $\pi \mapsto-\tfrac{1}{2}\norm{B[\pi]}^2$ in the $\Reg\QM$ objective: that is, if for all $\pi \in \measplus{\X \times \Y}$,
\begin{gather*}
\tfrac{1}{2} \norm{A[\pi^*]}^2 - \scal{B[\pi^*], B[\pi^*]} + \Reg(\pi^*) \leq \tfrac{1}{2} \norm{A[\pi]}^2 - \scal{B[\pi^*], B[\pi]} + \Reg(\pi), \quad \text{or equivalently}: \\
\tfrac{1}{2} \norm{A[\pi^*]}^2 - \tfrac{1}{2} \norm{B[\pi^*]}^2 + \Reg(\pi^*) \leq \tfrac{1}{2} \norm{A[\pi]}^2 - \tfrac{1}{2} \norm{B[\pi]}^2 + \Reg(\pi) + \tfrac{1}{2} \norm{B[\pi - \pi^*]}^2.
\end{gather*}
The last inequality is approximately satisfied by $\widehat\pi$, its discrepancy controlled by the tolerance $\delta$:

\begin{theorem}
\label{corr:main_convergence_criticality}
If $\delta_{\Reg\OT} \le \delta^2/4$, and if the inner loops converge, the output of \cref{alg:ugw} is a $\delta$-critical plan in the following sense (recalling that $\Engy_\Reg = \Engy + \Reg$ denotes the matching objective):
\begin{equation*}
    \text{For all } \pi \in \measplus{\X \times \Y}, \quad 
    \Engy_\Reg(\widehat\pi) \;\le\; \Engy_\Reg(\pi) + \tfrac{1}{2}\norm{B[\pi - \widehat\pi]}^2 + \delta \norm{B[\pi - \widehat\pi]} + \delta^2/4.
\end{equation*}
\end{theorem}

\paragraph{Practical example: entropic unbalanced GW.} Theorem~\ref{thm:main_convergence} requires $h_\Gamma$ to be $L$-smooth in order to ensure end-to-end convergence.
This hypothesis does not hold for all choice of $\Reg$, but is satisfied whenever $\Reg$ contains an entropic term $\varepsilon \KL(\pi \Vert \alpha\otimes\beta)$, such as in the unbalanced case of \cref{eq:unbalanced_reg}: then, $L = 1 + O(\varepsilon^{-1})$.
Corresponding $\Reg\OT$ problems can be solved with the Sinkhorn algorithm, an iterative method that improves the dual variables $\phi$ at each step \citep{cuturi2013sinkhorn, chizat2018scaling}.
Relying on the Sinkhorn algorithm has several practical advantages.
First, it is a dual solver that does not manipulate plans $\pi$ explicitly, improving memory usage.
Besides, Sinkhorn converges linearly: exploiting this property, we further prove in \cref{subsec:unified_ugw} that \cref{alg:ugw} terminates after $\widetilde{O}((\rho+\varepsilon)/(\delta^2\varepsilon^2))$ single Sinkhorn iterations.
The linear convergence of Sinkhorn also permits to replace the $\Xi$-updates by Nesterov's accelerated scheme with an inexact, but arbitrarily close, oracle \citep{devolder2013strongly}. This acceleration scheme can save a factor $\varepsilon^{-1/2}$ in the end-to-end complexity; see \cref{cor:ugw_nesterov_acceleration} for details and conditions of use.
Finally, since we have an estimate of the smoothness in this case, we can calibrate the step size $\tau$ automatically without having to tune this hyperparameter manually.
This evidences the benefits of entropic regularization for QM, and allows us to build optimized end-to-end pipelines for unbalanced GW based on the Sinkhorn algorithm.

\paragraph{Warm starts and early stopping.} The stopping criteria of \cref{alg:ugw} are theoretically-grounded but rather conservative: since $(\Gamma, \Xi)$ varies slowly between consecutive $\Reg\OT$ solves, warm-starting the solver at its previous dual potentials makes a few solver iterations per $\Xi$- and $\phi$-update -- or even a fixed number of them -- sufficient in practice.
Empirically, in the entropic case, our method still remains stable when only one Sinkhorn update is performed at each $\phi$-block.
This represents a large computational acceleration, since we do not solve a full $\Reg\OT$ problem before each $\Xi$ update.  

\paragraph{Efficient implementation in the block-separable case.} To fully benefit from memory efficiency of dual $\Reg\OT$ solvers like Sinkhorn, we must replace the computations of $A[\widehat\pi]$ and $B[\widehat\pi]$ with cheaper ones.
When $A$ and $B$ are block-separable, they are parametrized by their embeddings so that $A[\pi]$ and $B[\pi]$ take the form of a cross-correlation:
\begin{equation}
\label{eq:primal_corr}
A[\pi] = {\rm Corr}_A[\pi] := \int \Phi_A(x)\Psi_A(y)^{\top} \diff\pi(x,y).
\end{equation}
These standard operations are efficiently implementable with Sinkhorn-like solvers.
Moreover, the cost $c_{\Gamma,\Xi}$ of \cref{eq:effective_cost_decomp} is an inner product, a cost function supported by most OT libraries: this makes our algorithm compatible with any choice of $\Reg$ as long as a $\Reg\OT$ solver is available.
The block-wise structure makes it fully modular (\cref{fig:schema_algo}), and benefits from modern Sinkhorn backends \citep{charlier2021kernel,ye2026flashsinkhorn}: our method is GPU-friendly, memory-efficient, and accelerable with multiscaling.
Since matching problems are parametrized by embeddings, dimensionality reduction can be naturally applied to offer a substantial speed-up over exact resolution.
We refer to \cref{appendix:detailed_implementation} for implementation details, acceleration techniques, and a discussion of initialization strategies.

\begin{algorithm}[t]
    \caption{$\Reg\QM$ solver}
    \label{alg:ugw}
    \KwIn{measures $\alpha, \beta$; operators $A, B$; $\Reg\OT$ solver; tolerances $\delta$, $\delta_{\Reg\OT}$; step size $\tau$.}
    \KwOut{a $(\delta, \delta^2/4)$-critical plan $\widehat\pi$ (\cref{thm:main_convergence}), with its variables $(\Gamma, \Xi, \phi)$.}
    Initialize $\Gamma_0$, $\Xi$ and $\phi$\;
    \For{$t = 0, 1, 2, \dots$}{
        \For{$k = 0, 1, 2, \dots$}{
            $(\phi, \widehat\pi) \gets$ $\Reg\OT$ solver (cost $c_{\Gamma_t, \Xi}$, warm start at $\phi$, accuracy $\delta_{\Reg\OT}$) \tcp*[r]{$\phi$-block}
            \lIf{$\tfrac{1}{2}\norm{A[\widehat\pi] - \Xi}^2 + \delta_{\Reg\OT} \le \delta^2/4$}{\textbf{break}}
            $\Xi \gets \Xi + \tau\,(A[\widehat\pi] - \Xi)$ \tcp*[r]{$\Xi$-block}
        }
        $\Gamma_{t+1} \gets B[\widehat\pi]$ \tcp*[r]{$\Gamma$-block}
        \lIf{$\norm{\Gamma_{t+1} - \Gamma_t} \le \delta$}{\KwRet{$\widehat\pi$}}
    }
\end{algorithm}

\begin{figure}[t]
\centering
\includegraphics[width=\linewidth]{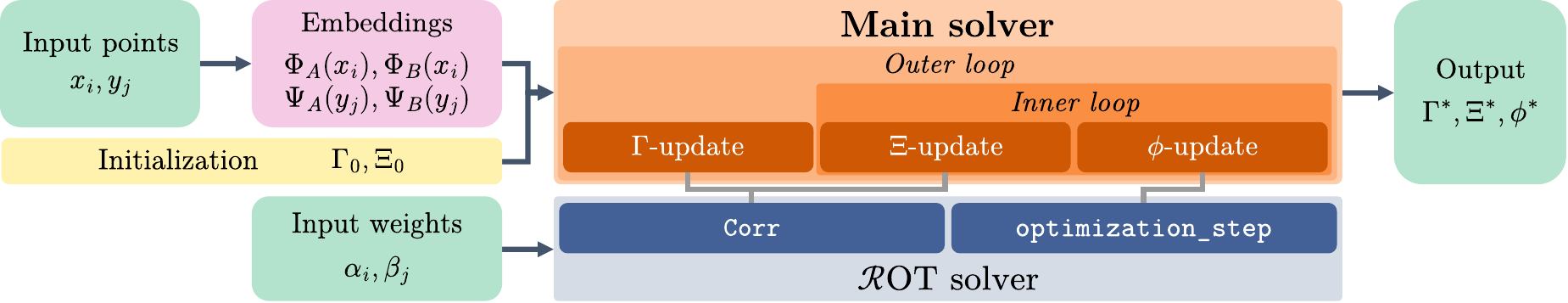} 
\caption{High-level description of our QM method for block-separable kernels. 
To apply our main solver to a specific $\Reg\QM$ problem, we only need to specifiy the block-separable embeddings of the objective (using eigendecomposition of $k$ or explicit equations) and a $\Reg\OT$ solver for the chosen $\Reg$.
This makes our approach fully modular and compatible with any state-of-the-art OT backend.
}
\label{fig:schema_algo}
\end{figure}

\section{Applications}

We now highlight the practical utility of our new framework by presenting several applications on a wide variety of data types, all solved using our numerical routine presented in \cref{section:numerical_optim}. 
We detail our experimental setup in \cref{appendix:expe_setup}.
These applications aim to demonstrate the versatility of our solver: we used the same algorithm to compute all these experiments, the only changes lying in the feature embeddings and the selected term $\Reg$.

\paragraph{Graph matching.} Fig~\ref{fig:graph_benchmark} evaluates our method on random graph matching using GW with shortest path as base costs.
We compare our exact and dimension-reduced methods against the standard EGW solver of \cite{peyre2016gromov} and the low-rank (LR) solver of \cite{scetbon2022linear}. 
All methods perform similarly, except for LR which trades matching quality for computational speed.
In sharp contrast, our method reaches an excellent accuracy with as few as $D=10$ dimensions.
Although EGW has a larger asymptotic complexity than our low-dimensional methods, its minimal overhead makes it competitive in this setting.
However, EGW lacks theoretical convergence guarantees: in \cref{appendix:additional_expes}, we highlight simple cases where EGW fails to converge entirely.
To our knowledge, our method is the first to robustly compute GW matchings on graphs with tens of thousands of nodes in a few seconds without suffering from such convergence issues.

\begin{figure}[h!]
\centering
\includegraphics[width=\linewidth]{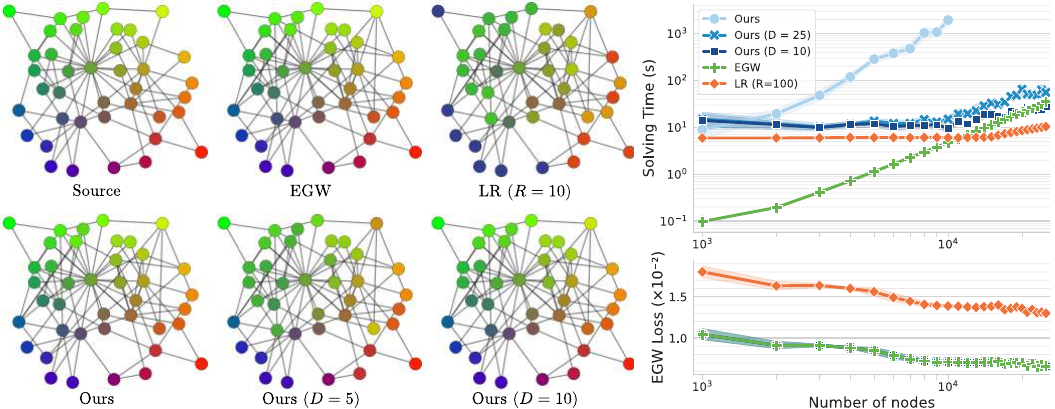} 
\caption{Gromov-Wasserstein methods for graph matching.
\textit{Left}: qualitative analysis; colors indicate which nodes are matched together.
\textit{Right}: quantitative results of speed and accuracy.
}
\label{fig:graph_benchmark}
\end{figure}

\paragraph{Word embedding alignment.} To showcase our method in high dimension, we reproduce the experiment of \cite[Table 1]{alvarez2018gromov} by aligning vocabularies of different languages.
The MUSE dataset \citep{conneau2017word} provides $300$-dimensional embeddings of different language vocabularies (English, Spanish, French, Italian and German).  
We apply GW between the $N = 20k$ most frequent words of these languages, using a multiscale scheme to reduce solving time.
We infer a translation for each source word in the target language and report total precision in \cref{table:nlp_results}.
We see that our solver obtains comparable results to the baseline while being almost $20\times$ faster. 
Using Nystr\"om methods \citep{nakatsukasa2020fast} further accelerates our algorithm by reducing the embedding dimensions, providing a flexible trade-off between precision and running time.

\begin{table*}[t!]
  \caption{Solving time and precision of word translation using balanced GW, following the method of \cite{alvarez2018gromov}.
  Baseline results are directly reported from the original study.}
  \label{table:nlp_results}
  \begin{center}
    \begin{small}
    \begin{tabularx}{\textwidth}{lcrCCCCCCCC}
          \toprule
           \sc Solver & \sc Dim. & \sc Time & \multicolumn{2}{c}{\sc \quad En-Es} & \multicolumn{2}{c}{\sc \quad En-Fr} & \multicolumn{2}{c}{\sc \quad En-It} & \multicolumn{2}{c}{\sc \quad En-De} \\
           \cmidrule(lr){4-5}\cmidrule(lr){6-7} \cmidrule(lr){8-9} \cmidrule(lr){10-11} 
           & & &  $\rightarrow$ & $\leftarrow$ &  $\rightarrow$ & $\leftarrow$ &  $\rightarrow$ & $\leftarrow$ & $\rightarrow$ & $\leftarrow$ \\
          \midrule
          \textsc{Ours} & $100$ & $\mathbf{36s}$ & $75.4$ & $75.8$ & $79.1$ & $79.3$ & $71.7$ & $70.6$ & $58.4$ & $59.6$ \\
          \textsc{Ours} & $200$ & $70s$ & $79.6$ & $80.5$ & $81.4$ & $81.8$ & $78.8$ & $78.0$ & $71.8$ & $74.2$ \\
          \textsc{Ours} & $300$ & $90s$ & $80.7$ & $\mathbf{81.6}$ & $\mathbf{81.7}$ & $\mathbf{82.2}$ & $\mathbf{79.6}$ & $\mathbf{79.0}$ & $\mathbf{72.7}$ & $\mathbf{75.2}$ \\
          \textsc{Base} & $300$ & $2220s$ & $\mathbf{81.7}$ & $80.4$ & $81.3$ & $78.9$ & $78.9$ & $75.2$ & $71.9$ & $72.8$\\
          \bottomrule
    \end{tabularx}
    \end{small}
  \end{center}
  \vskip -0.1in
\end{table*}

\paragraph{Semi-balanced shape matching.} In \cref{fig:shape_matching}, we match a source mesh to various partial targets from the SHREC'16 dataset \citep{cosmo2016shrec} using semi-balanced GW.
Semi-balancedness ensures that the whole partial shapes are matched to the source, while ignoring vertices of the full mesh that correspond to holes of the target.
Our method converges in a few minutes and accurately preserves the aspect of the textures.
We provide qualitative comparisons with other GW baselines in \cref{appendix:additional_expes} which highlight the limitation of concurrent unbalanced GW methods. 

\begin{figure}[t!]
\centering
\includegraphics[width=\linewidth]{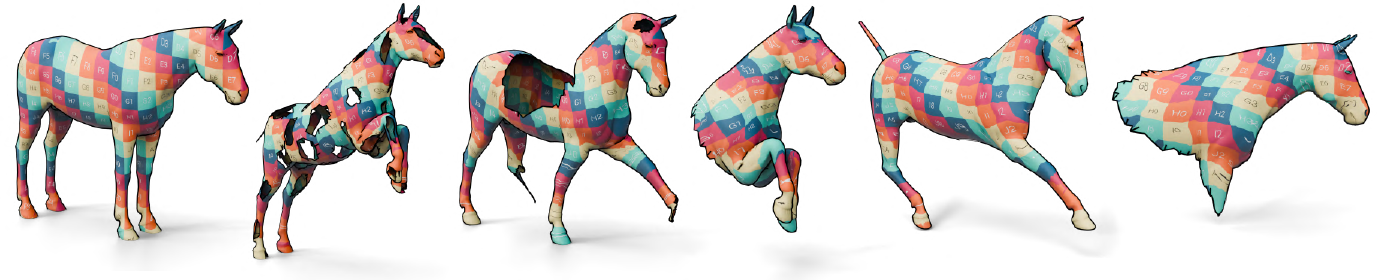} 
\caption{Texture transfer between source shape (\emph{left}) and different targets using our method.
The source contain $10k$ vertices; the targets, from $3k$ to $10k$.}
\label{fig:shape_matching}
\end{figure}

\paragraph{Fused matching of 3D point clouds.} To demonstrate the large-scale applicability of our method, we apply our solver to the scene reconstruction dataset of \cite{choi2015robust} which contains various scans of the same room encoded as 3D clouds with $100k$ to $200k$ points.
We match these scans with unbalanced GW, using RGB colors as additional linear features, following the \emph{fused} idea of \cite{vayer2020fused}.
The results are computed in about $30$ min and are qualitatively satisfying, as depicted in \cref{fig:livingroom_matching}.
    
\paragraph{Fracture reconstruction.} In \cref{fig:fracture_matching}, we highlight our framework polyvalence by matching fractured objects with their original templates from the synthetic dataset of \cite{li2025garf}.
We use a cluster-aware variant of GW that makes it invariant to isometries of the individual fragments; due to the high number of local minima, we selected the output as the best result over dozens of initializations.
We combine this technique with multiscaling to reduce the computational burden of these computations.
Since our input shapes contain $100k$ to $200k$ vertices each, an optimization still takes several hours.
Yet, this preliminary work shows how our framework can help to tackle new matching problems with custom geometries without needing to redesign solvers from scratch.

\begin{figure}[t!]
\centering
\begin{minipage}{0.49\textwidth}
  \centering
\includegraphics[width=\linewidth]{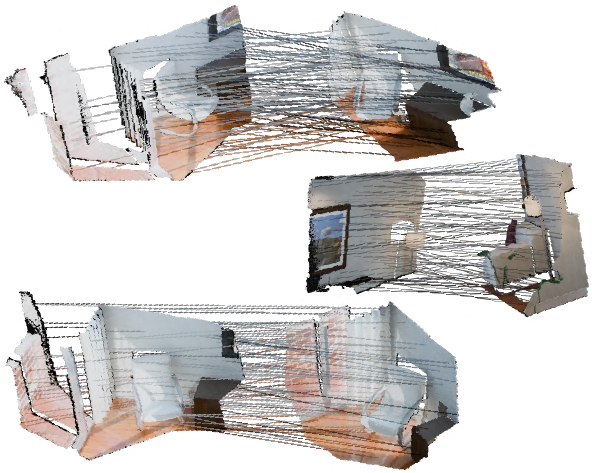} 
\caption{Matching between 3D point clouds represented by $100$ edges sampled from $\pi^*$. 
}
\label{fig:livingroom_matching}
\end{minipage}
\hfill
\begin{minipage}{0.49\textwidth}
  \centering
\includegraphics[width=\linewidth]{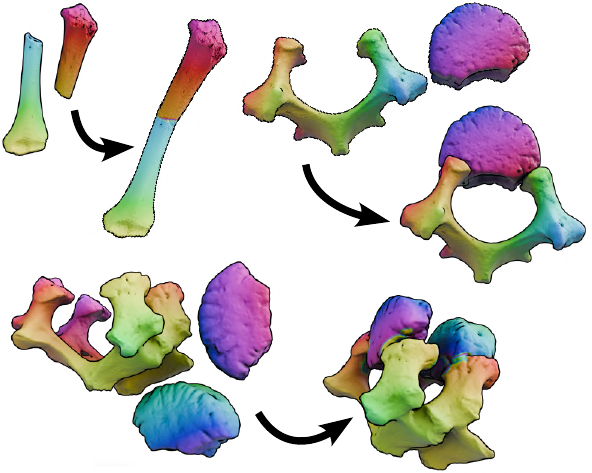} 
\caption{Color transfer between fractured and original objects using our method.
}
\label{fig:fracture_matching}
\end{minipage}
\end{figure}

\section{Conclusion}

Our work provides a common foundation for studying and solving a wide range of matching problems in a principled manner.
Rather than requiring a dedicated solver for each formulation, it offers a flexible framework that can be readily adapted to new models.
Built on duality theory, our framework encompasses Gromov--Wasserstein matching and its fused and unbalanced entropic variants within a unified algorithmic approach.
For a broad class of regularized quadratic problems, we establish convergence rates, thereby bridging a gap in the existing literature.
As illustrated by our fracture-matching model, the framework also enables new formulations to be explored without designing a solver from scratch.
We therefore believe that generic solvers such as ours can facilitate rapid prototyping and testing, and help the community to investigate approaches beyond classical matching algorithms.
More broadly, our theoretical results provide insight into which classes of matching problems are easier to solve and reveal connections between the many methods proposed in the literature.

\subsection*{Acknowledgements}

This work was supported by the French “Agence Nationale
de la Recherche” via the “PR[AI]RIE-PSAI” project (ANR-
23-IACL-0008). 

\subsection*{AI use statement}

In this work, we used generative AI tools to assist with and polish writing, to retrieve and discover related work, and to assist with the development and verification of some mathematical proofs. We did not use generative AI tools for research ideation, the design or execution of the research, or the generation of synthetic datasets.

All AI-assisted content was carefully reviewed and independently verified by the authors. In particular, we checked the validity of all mathematical arguments and proofs, rewrote any AI-generated text ourselves, and verified the existence and relevance of all references retrieved with the assistance of LLMs. We take full responsibility for the final content of this work, including any text, claims, proofs, or other artifacts produced with the aid of generative AI.

\subsection*{Reproducibility statement}

To ensure reproducibility of our results, we provide extensive details of both theoretical and implementation aspects in the appendix.  
In \cref{appendix:expe_setup}, we describe the setup of our experiments and the preprocessing of the input data.
In \cref{appendix:additional_results,appendix:detailed_implementation}, we provide a step-by-step description of all the computations and implementation tools necessary to reproduce our algorithm; we also attach our source code with simple examples of experiments as a file to the submission.
Finally, in \cref{appendix:proofs,appendix:end-to-end-cv}, we rigorously prove all the theoretical results that we state in our article and discuss thoroughly convergence guarantees.

\bibliography{bibliography}
\bibliographystyle{iclr2027_conference}

\appendix

\section{Additional Experiments}
\label{appendix:additional_expes}
\subsection{Other Baseline Comparisons}

\paragraph{Failure cases of naive solver.} We illustrate the relevance of our approach by identifying examples where our algorithm reaches a true EGW minimum whereas concurrent solvers \citep{peyre2016gromov} fail to converge.
As explained in \cite{houry2026gromov}, this can happen when the costs are not CNT.
We exhibit two typical examples of non-CNT configurations: shortest path distances on arbitrary graphs (\cref{fig:graph_pathological}), and 2D point clouds for base costs equal to $\norm{\cdot}^p$, with $p > 2$ (\cref{fig:points_pathological}).
In both cases, our method manage to converge to the correct minimum whereas the standard baseline oscillates between two transport plans without reaching any critical point.
Note that these examples represent typical inputs for GW applications and are not purely abstract pathological examples.
Therefore, this type of behavior is to be expected in practice, and using our solver ensures that convergence is indeed reached in every configuration.
Also note that, despite the symmetries of the inputs in \cref{fig:graph_pathological} and the symmetry of the initialization, our solver breaks the symmetry and reaches one of the global minima of the problem.
Therefore, numerical noise is sufficient to prevent our algorithm from getting stuck in a critical point -- a well-known property of gradient descent \citep{lee2016gradient}.

\paragraph{Partial shape matching using other baselines.} We complement the results of \cref{fig:shape_matching} with matching obtained using other Gromov-Wasserstein methods.
Except from our method, only two other baselines are able to solve unbalanced GW problems: the UGW solver of \cite{sejourne2021unbalanced}, and the low-rank relaxation of \cite{scetbon2023unbalanced}.
The first algorithm cannot scale properly to the inputs used in these experiments.
Although faster, the low-rank solver rely on a relaxation that creates undesirable quantization artifacts, as we illustrate at the top of \cref{fig:shape_matching_appendix}.
Therefore, to the best of our knowledge, our method is the only one which is able to compute accurate unbalanced Gromov-Wasserstein solutions on this input size in a reasonable time.
As a reference, we also show the matching obtained using a balanced Gromov-Wasserstein solver on these inputs at the bottom \cref{fig:shape_matching_appendix}.
We observe large distorsion in the results, evidencing the need to rely on unbalanced versions of Gromov-Wasserstein to tackle this type of problems.

\begin{figure}[h]
\centering
\includegraphics[width=0.9\linewidth]{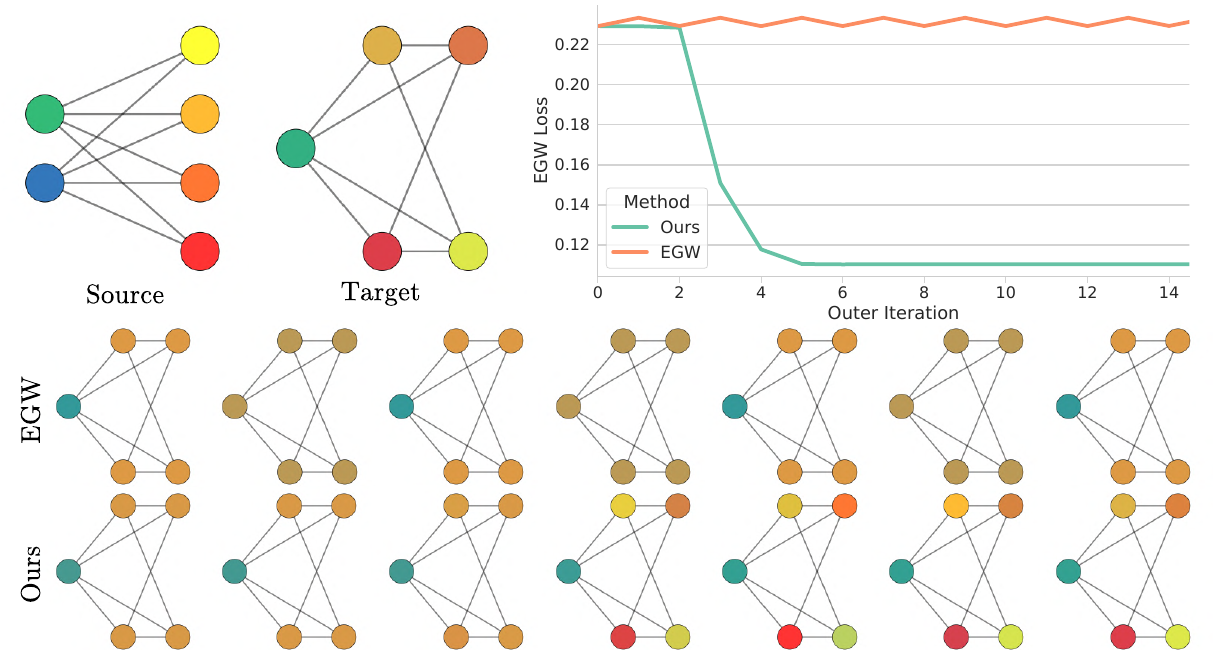} 
\caption{Graph matching obtained through the iterations of the EGW solver (\cite{peyre2016gromov}) and ours. EGW gets stuck between two configurations, whereas ours converge to the true optimum.
}
\label{fig:graph_pathological}
\end{figure}

\begin{figure}[h]
\centering
\includegraphics[width=0.95\linewidth]{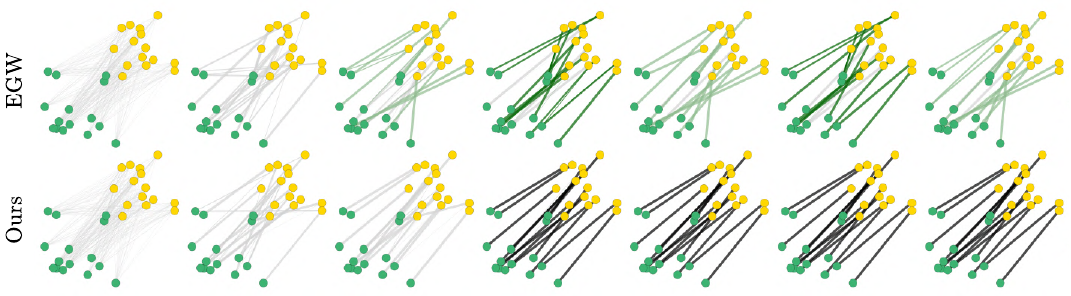} 
\caption{Matching obtained through iterations of the EGW solver and ours with $c_\X, c_\Y = \norm{\cdot}^6$.}
\label{fig:points_pathological}
\end{figure}

\begin{figure}[h]
\centering
\includegraphics[width=0.9\linewidth]{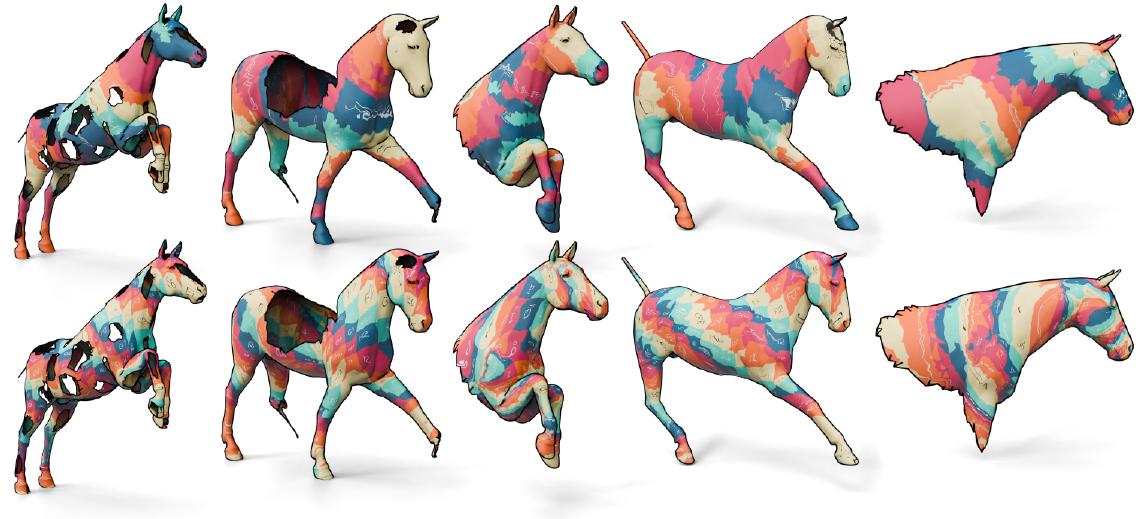} 
\caption{Texture transfer between the source of \cref{fig:shape_matching} and different targets using the LR unbalanced solver of \cite{scetbon2023unbalanced} (\emph{top}) and  the balanced EGW solver of \cite{peyre2016gromov} (\emph{bottom}).
}
\label{fig:shape_matching_appendix}
\end{figure}

\clearpage

\subsection{Experimental Setup}
\label{appendix:expe_setup}

All our algorithms were implemented in Pytorch \cite{imambi2021pytorch} and executed on a RTX6000 GPU with 24GB of VRAM.
To ensure a fair comparison with baselines, we reimplemented concurrent algorithms in Pytorch as well.
In dimension reduction settings, we performed OT related operations using either the Keops \citep{charlier2021kernel} or FlashSinkhorn \citep{ye2026flashsinkhorn} backends. 
Below, we provide additional details on the experiments we performed in this article.

\paragraph{Graph matching.} We compute graph matching using the balanced entropic GW objective:
\begin{equation*}
     \min_{\pi \in \measplus{\X \times \Y}} \int (d_\X(x,x') - d_\Y(y,y'))^2\diff\pi(x,y)\diff\pi(x',y') + \iota_{\Pi(\alpha,\beta)}(\pi) + \varepsilon \KL(\pi \Vert \alpha \otimes \beta),
\end{equation*}
where $d_\X$ and $d_\Y$ are the shortest path distances of the input graphs $\X$ and $\Y$, and $\varepsilon = 2e-3$.
The source graphs are generated randomly from the Barabási–Albert model using the corresponding function of the networkx library \citep{hagberg2020networkx}, with an edge parameter of $2$.
Targets graphs are obtained by adding $20 \%$ of noisy edges to the source, so that there is a ground-truth correspondences between the nodes of the source and the target.  
To apply our method, we normalize then diagonalize the shortest path distance matrices so that it takes the form of \cref{prop:diag_costs}.
In our dimension reduction methods, we truncate this diagonalization to the $D$ largest eigenvalues in absolute value.
For the low-rank baseline, we also truncate the costs to the $D = 100$ first eigenspaces, necessary to ensure linear-time computations of this method.
We use uniform weights $\alpha$ and $\beta$ on both source and target graphs, and initialize the algorithms at the trivial plan -- except for low-rank where we use a kmeans heuristics for initialization.
The accuract reported in \cref{fig:graph_benchmark}, called \emph{EGW loss}, refers to the value of \cref{eq:rgw_def} computed using the matching $\pi$ obtained after convergence.

\paragraph{Word embedding alignment.} Results of \cref{table:nlp_results} rely on the balanced entropic GW problem:
\begin{equation*}
     \min_{\pi \in \measplus{\X \times \Y}} \int (c_\X(x,x') - c_\Y(y,y'))^2\diff\pi(x,y)\diff\pi(x',y') + \iota_{\Pi(\alpha,\beta)}(\pi) + \varepsilon \KL(\pi \Vert \alpha \otimes \beta),
\end{equation*}
with cosine similarity as base costs:
\begin{equation*}
   c_\X(x,x') = \scal{\frac{x}{\norm{x}},\frac{x'}{\norm{x'}}} \quad \text{ and } \quad     c_\Y(y,y') = \scal{\frac{y}{\norm{y}},\frac{y'}{\norm{y'}}}. 
\end{equation*}
These have a scalar structure, so we can directly apply the decomposition of \cref{prop:scalar_decomp} on the spherical projection of the inputs.
We use a multiscale scheme by solving a coarse problem containing only the $2k$ most frequent words of each language, using the coarse result as initialization for the fine-scale one.
We then estimate a rigid alignment of the space from the optimal matching, as detailed in \cite{alvarez2018gromov}.
Finally, we match each word of the source language to a word of the target language using the CSLS matching method \citep{conneau2017word}.
The reported precision corresponds to the proportion of source words whose matched target is a valid translation in the dictionaries provided by the MUSE dataset. 

\paragraph{Semi-balanced shape matching.} The explicit optimization problem we used for \cref{fig:shape_matching} is:
\begin{equation*}
     \min_{\pi \in \measplus{\X \times \Y}} \int (c_\X(x,x') - c_\Y(y,y'))^2\diff\pi(x,y)\diff\pi(x',y') + \iota_{\alpha}(\pi_1) + \varepsilon \KL(\pi \Vert \alpha \otimes \beta),
\end{equation*}
with $\varepsilon=5e-4$. The costs $c_\X$ and $c_\Y$ are sums of squared norms:
\begin{equation*}
    c_\X(x,x') = \norm{x - x'}^2_{\text{geodesic}} + \norm{x - x'}^2_{\text{exp}} \quad \text{ and } c_\X(y,y') = \norm{y - y'}^2_{\text{geodesic}} + \norm{y - y'}^2_{\text{exp}}.
\end{equation*}
The norm $\norm{\cdot}_{\text{geodesic}}$ correspond to Euclidean approximations of geodesic distances as in \cite{panozzo2013weighted}.
We follow the same method as detailed in \cite{houry2026gromov}, yielding Euclidean embeddings in dimension $D = 8$.
On the other hand, the terms $\norm{\cdot}_{\text{exp}}^2$ represent squared-norm approximations in dimension $D = 20$ of the conditionally negative cost: 
\begin{equation*}
d_\X(x,x') = 2 - 2 \exp \left(- \frac{\norm{x - x'}}{0.5}\right) \quad \text{ and } \quad     d_\Y(y,y') = 2 - 2 \exp \left(- \frac{\norm{y - y'}}{0.5}\right),
\end{equation*}
using kernel PCA, to compute approximate Euclidean embeddings, as advocated in \cite{houry2026gromov}.
This way, $c_\X$ and $c_\Y$ combine both global geodesic and local Euclidean information.
We initialized the solver using landmarks and chose the weights $\alpha$ and $\beta$ proportional to average area of triangles surrounding each vertex.
For the low-rank baseline (in \cref{fig:shape_matching_appendix}), we used a rank of $R = 100$.

\paragraph{Fused matching of 3D point clouds.} To compute the matchings of \cref{fig:livingroom_matching}, we minimize the fused unbalanced GW objective:
\begin{multline*}
    \min_{\pi \in \measplus{\X \times \Y}} \int (c_\X(x,x') - c_\Y(y,y'))^2\diff\pi(x,y)\diff\pi(x',y') \\ + \lambda \int \norm{{\rm RGB}(x) - {\rm RGB}(y)}^2 \diff\pi(x,y)
    + \rho \KL(\pi_1 \Vert \alpha) + \rho \KL(\pi_2 \Vert \beta)  + \varepsilon \KL(\pi \Vert \alpha \otimes \beta),
\end{multline*}
where $c_\X, c_\Y$ are $20$-dimensional approximations of classical euclidean distances (following the kernel PCA method of \cite{houry2026gromov}). 
The fused features ${\rm RGB}(x) \in [0,1]^3$ correspond to the color of the point $x$.
We chose $\varepsilon = 2e-3$, $\rho = 1e-1$, and $\lambda = 4e-4$. 
We applied multiscaling, sampling $10k$ random points from each input to build the coarse problem.
We chose $\alpha$ and $\beta$ as uniform weights and initialized the solver at the trivial plan.

\paragraph{Fracture reconstruction.} In \cref{fig:fracture_matching}, we optimize the weighted GW objective: 
\begin{multline}
\label{eq:fracture_objective}
     \min_{\pi \in \measplus{\X \times \Y}} \int \sum_{k} \one_{x \in C_k} \one_{x' \in C_k} (\norm{x-x'}^2 - \norm{y-y'}^2)^2\diff\pi(x,y)\diff\pi(x',y') \\ + \varepsilon \KL(\pi \Vert \alpha \otimes \beta),
\end{multline}
where $\one_{x \in C_k}$ is equal to $1$ if $x$ belongs to the fragment $C_k$ and $0$ otherwise.
We take $\varepsilon = 10^{-4}$ and $\alpha$, $\beta$ as uniform weights.
The coarse optimization is performed on random subsamples of the inputs with $3k$ points.
The best coarse plan (in terms of the objective of \cref{eq:fracture_objective}) is selected over $30$ random initializations.
The final output is then obtained by computing the fine-scale problem initialized at this best coarse plan.

\section{Additional Results}
\label{appendix:additional_results}

The main part of our article describes a general strategy to tackle matching problems.
This section provides additional ingredients to apply these results to classical matching instances, as well as existence results for the decompositions we use in our main theorem.
In \cref{appendix:regularizations}, we provide examples of constraint functions $\Reg$ and the associated duals.
In \cref{appendix:general_matching}, we justify the existence of DC decomposition for general matching, and factorizations of quadratic objectives in a difference of quadratic operators.
In \cref{appendix:block_separable}, we detail the formulation of our dual theorems when the quadratic operators are block-separable.
Finally, in \cref{appendix:gw_decomposition}, we explicit the block-separable decomposition of GW objectives, which are the formula we ultimately implement in our solvers.
Before that, we introduce additional and definitions:

\paragraph{Weak-* continuity.} A sequence of measures $(\alpha_n) \in \meas{\X}$ converges weakly-* to $\alpha \in \meas{\X}$ (denoted by $\alpha_n \rightharpoonup \alpha$) if for any continuous function $f \in \C(\X)$, $\int\!f \diff\alpha_n \rightarrow \int\!f \diff\alpha$.
A function $\mathcal{F}: \meas{\X \times \Y} \to V$ is weak-* continuous if $\mathcal{F}(\alpha_n) \rightarrow \mathcal{F}(\alpha)$ for all sequence $\alpha_n \rightharpoonup \alpha$.
Bounded closed spaces of $\meas{\X \times \Y}$ are compact for the weak-* topology; in particular, if $\Engy$ is weak-* continuous and goes to infinity at infinity, the infimum: 
\begin{equation*}
    \inf_{\pi \in \meas{\X \times \Y}} \Engy(\pi)
\end{equation*}
is always attained.

\paragraph{Trivial measure.} Given $\alpha \in \meas{\X}$ and $\beta \in \meas{\Y}$, we denote by $\alpha \otimes \beta \in \meas{\X \times \Y}$ the \emph{trivial} measure with marginals $\alpha$ and $\beta$:
\begin{equation*}
    \text{For all } c \in \C(\X \times \Y), \quad \int c \diff\alpha\otimes \beta = \int c(x,y) \diff\alpha(x) \diff\beta(y).
\end{equation*}
We will notably use this notation to simplify the expression of quadratic matching objectives:
\begin{equation*}
    \int k \diff\pi\otimes\pi = \int k(x,y,x',y')\diff\pi(x,y)\diff\pi(x',y').
\end{equation*}

\paragraph{Generalized KL-divergence.} In this work, we adopt the setting of \cite{sejourne2022generalized} and use the following definition of KL-divergence:
\begin{equation*}
    \KL(\pi \Vert \alpha \otimes \beta) := \int \log\left(\frac{\diff\pi}{\diff\alpha\otimes\beta}\right)\diff\beta - \int \diff\alpha\otimes\beta + \int \diff\pi.
\end{equation*}
The terms $\int \diff\pi$ and $\diff\alpha\otimes\beta$ cancel when both measures have the same mass, but are required in unbalanced settings to avoid biasing matching objectives towards $0$-mass plans.

\paragraph{Convex indicator function.} Given a convex set $A$, we denote by $\iota_{A}$ its convex indicator function:
\begin{equation*}
    \iota_{A}(x) = \begin{cases}
        0 & \text{if } x \in A, \\
        + \infty &\text{otherwise}.
    \end{cases}
\end{equation*}

\paragraph{Adjoints of integral operators.} Since an integral operator $A : \pi(\X \times \Y) \to \Hilb_A$ is a continuous linear function, it admits an adjoint $A^\top : \Hilb_A \to \C(\X \times \Y)$ that is given by:
\begin{equation*}
    \text{For all } ~~ \Xi \in \Hilb_A ~~ \text{ and } ~~ (x,y) \in \X \times \Y, \quad (A^\top\Xi)(x,y) = \scal{\Xi, A(x,y)}.
\end{equation*}

\subsection{Classical choices of $\Reg$}
\label{appendix:regularizations}

In \cref{table:regularizers,table:regularizers_dual}, we introduce the main constraint terms $\Reg$ encountered in both matching and registration literatures, along with their duals (following the notations of \cref{theorem:tolanddualitygeneral}).
The two most standards are optimal transport and closest point.
In the OT case, the matching $\pi$ must belong to: 
\begin{equation}
\label{eq:couplings}
   \coupling{\alpha, \beta} := \left\{ \pi \in \measplus{\X \times \Y}, ~~ \pi_1 := \int \diff\pi(\cdot, y) = \alpha ~\text{ and }~ \pi_2 := \int \diff\pi(x, \cdot) = \beta \right\},
\end{equation}
the set of \emph{couplings} between $\alpha$ and $\beta$.
Closest point matchings, on the other hand, are obtained by matching each point of $\alpha$ to its nearest neighbor in $\beta$ based on a cost function $c$:
\begin{equation*}
    \pi_{\rm CP}(x,y) = \begin{cases}
        \alpha(x) & \text{if }  y \in \arg\min_x c(x,y), \\
        0 & \text{otherwise}.
        \end{cases}
\end{equation*}
Such matchings are actually the solutions of a problem of the form $\min_\pi \int c\diff\pi + \iota_{\alpha}(\pi_1)$, where $\iota_{\alpha}(\pi_1)$ only enforces the marginal constraint on $\alpha$.
Therefore, this popular method in registration is entirely compatible with our framework with the choice $\Reg(\pi) = \iota_{\alpha}(\pi_1)$.
Other important cases combine marginal penalization and entropic (or KL-divergence) regularization:
\begin{multline}
    \label{eq:csiszar_regs}
    \Reg(\pi) = \Div{\varphi_1}{\pi_1}{\alpha} + \Div{\varphi_2}{\pi_2}{\beta} + \varepsilon \KL(\pi \Vert \alpha \otimes \beta) \\ \text{ with } \quad \Div{\varphi}{\mu}{\nu} := \begin{cases} \int \varphi \left(\frac{\diff\mu}{\diff \nu}\right)\diff\nu & \text{ if } \mu \ll \nu, \\
    + \infty & \text{otherwise}.
    \end{cases}
\end{multline}
The function ${\rm D}_{\varphi}$ is called the Csisz\'ar divergence associated to the function $\varphi$: it is well-defined as soon as $\varphi$ is convex and $\varphi(1) = 0$.
As explained in \cite{sejourne2022generalized}, their dual problems admit the following form, the space of dual variables being $\Phi = \C(\X) \times \C(\Y)$:
\begin{equation*}
  \mathcal{D}_\Reg((f,g),c,\alpha,\beta) = - \int \varphi_1^*(-f)\diff\alpha - \int  \varphi_2^*(-g)\diff\beta - \varepsilon\int(e^{(f\oplus g - c)/\varepsilon} - 1)\,\diff\alpha\otimes\beta,
\end{equation*}
where $\varphi^*_1$ and  $\varphi^*_2$ are the convex conjugates of $\varphi_1$ and $\varphi_2$.
The unbalanced case of \cref{eq:unbalanced_reg} corresponds to $D_\varphi = \KL$ (with $\varphi(r) = r \log(r) - r + 1$).

\begin{table*}[h!]
  \caption{List of classical regularizers.}
  \label{table:regularizers}
\begin{tabularx}{\textwidth}{ll}
\toprule
Regularizer & $\Reg(\pi)$  \\
\midrule
Classical OT & $\iota_{\Pi(\alpha,\beta)}(\pi)$ \\
Closest Point & $\iota_{\alpha}(\pi_1)$ \\
Entropic OT & $\iota_{\Pi(\alpha,\beta)}(\pi) + \varepsilon \KL(\pi \Vert \alpha \otimes \beta)$ \\
Semi-Balanced EOT & $\iota_{\alpha}(\pi_1) + \varepsilon \KL(\pi \Vert \alpha \otimes \beta)$ \\
Unbalanced OT & $\rho_1 \KL(\pi_1, \Vert \alpha) + \rho_2 \KL(\pi_2 \Vert \beta) + \varepsilon \KL(\pi \Vert \alpha \otimes \beta)$ \\
Partial OT &  $\rho_1 \TV(\pi_1 \Vert \alpha) + \rho_2 \TV(\pi_2 \Vert \beta) + \varepsilon \KL(\pi \Vert \alpha \otimes \beta)$ \\
\bottomrule
\end{tabularx}
\end{table*}

\begin{table*}[h!]
  \caption{Dual functions of the regularizers.}
  \label{table:regularizers_dual}
\begin{tabularx}{\textwidth}{ll}
\toprule
Regularizer & $\mathcal{D}_\Reg((f,g), \alpha, \beta)$ \\
\midrule
Classical OT & $ \int f\diff\alpha + \int g \diff \beta + \iota_{f \oplus g \leq c}(f,g)$ \\
Closest Point & $\int f \diff \alpha + \iota_{f(x) = \min_y c(x,y)}(f)$ \\
Entropic OT & $\int f \diff\alpha + \int g \diff \beta - \varepsilon\int(e^{(f\oplus g - c)/\varepsilon} - 1)\,\diff\alpha\otimes\beta$ \\
Semi-Balanced EOT & $\int f \diff \alpha - \varepsilon\int(e^{(f\oplus 0 - c)/\varepsilon} - 1) \diff\alpha\otimes\beta $ \\
Unbalanced OT & $\rho_1\!\int(1 \!-\!e^{-f/\rho_1})\diff\alpha + \rho_2\!\int(1\!-\!e^{-g/\rho_2})\diff\beta
    - \varepsilon\!\int\bigl(e^{(f\oplus g - c)/\varepsilon} \!-\! 1\bigr)\diff\alpha\otimes\beta$ \\
Partial OT &  $\int \min(\rho_1, f) \diff \alpha + \int \min(\rho_2, g) \diff \beta - \varepsilon\!\int\bigl(e^{(f\oplus g - c)/\varepsilon} - 1\bigr)\diff\alpha\otimes\beta $  \\
\bottomrule
\end{tabularx}
\end{table*}

\subsection{General DC Decomposition of Matching Problems}
\label{appendix:general_matching}

Here, we motivate our study of DC matching problems by providing several theorems showing the existence of DC decompositions for general matching objectives.

\paragraph{General matching approximation by DC functions.} DC functions are dense in spaces of compactly-supported continuous functions: therefore, continuous matching objectives can always be approximated by a difference of convex functions on compact domains.
If the matching objective grows arbitrarily at infinity, the minimization can always be restricted to a compact set and the class of DC functions adequately approximates any sufficiently regular matching problem.
We formalize this property in the following proposition.

\begin{restatable}{proposition}{densitydcfunctions}
\label{prop:density_dc_functions}
Let assume that $\Engy$ is continuous for the weak-* convergence and goes to infinity at infinity.
Then, for any precision $\delta > 0$, there exists a bounded convex set $V \subset \measplus{\X \times \Y}$ containing the minima of $\Engy$ and two convex functions $\Cvxa_\delta$, $\Cvxb_\delta$ continuous on their domain such that:
\begin{equation*}
\text{For all } \pi \in V, \quad \abs{\Engy(\pi) - \Cvxa_{\delta}(\pi) + \Cvxb_{\delta}(\pi)} < \delta,
\end{equation*}
and $ \Cvxa_{\delta}(\pi) - \Cvxb_{\delta}(\pi) = + \infty$ for $\pi \notin V$.
\end{restatable}

Besides, when $\Cvxa$ and $\Cvxb$ are superlinear, their convex conjugates are continuous:

\begin{restatable}{lemma}{lemmacontinuityconjugate}
\label{lemma:continuity_conjugate}
Let assume that $\Cvxa$ is superlinear:
\begin{equation*}
    \frac{\Cvxa(\pi_n)}{\TV(\pi_n)} \longrightarrow + \infty \quad \text{ for all } (\pi_n) \in \meas{\X \times \Y} \quad \text{ such that } \quad \TV(\pi_n) \longrightarrow \infty,
\end{equation*}
with $\TV(\pi)$ denoting the total variation of $\pi$.
Then, the convex conjugate $\Cvxb^*$ is continuous on $\C(\X \times \Y)$ for the uniform norm.
\end{restatable}
Since adding the same term to $\Cvxa$ and $\Cvxb$ do not modify their difference, we can always enforce the DC decomposition to have superlinear terms by adding the function $\pi \mapsto \TV(\pi)^2$ to both $\Cvxa$ and $\Cvxb$.
Therefore, we can always enforce the continuity of the convex conjugates $\Cvxa^*$ and $\Cvxb^*$ without any additional hypothesis -- a crucial requirement for our reformulation theorems.

When $\X$ and $\Y$ are finite spaces, the set $\meas{\X \times \Y}$ is identified with a Euclidean matrix space: if $\Engy$ is of class $\C^2$, \cite{hartman1959functions,hiriart1985generalized} proved that it admits an exact DC decomposition over the whole space.
In this case, the DC hypothesis is satisfied without any approximation.
We refer to \cite{hiriart1985generalized} for other characterizations of DC decomposable functions.

\paragraph{Factorizability of quadratic operators.} In the discrete case, the spectral theorem ensures that quadratic operators always admit a factorization like \cref{eq:loss_factorized} which is given by the decomposition of the kernel matrix $[k(x_i,y_j,x_k,y_l)]_{(ij),(kl)}$ in positive and negative eigenspaces.
Similarly, since the images of $A$ and $B$ are isomorphic to Euclidean spaces $\R^\DD$ and $\R^\EE$, the singular value decompositions of each matrix $[A(x_i,y_j)_d]_{ij}$ and $[B(x_i,y_j)_e]_{ij}$ (with $1 \leq d \leq \DD$ and $1 \leq e \leq \EE$) provides a block-separable factorization of QM objectives in the finite setting.

When $\alpha$ and $\beta$ are continuous, this reasoning does not hold anymore.
Indeed, spectral decompositions are notions related to Hilbert spaces, whereas $\meas{\X \times \Y}$ is only a Banach space.
However, when using entropic regularization, the $\Reg\QM$ optimization can actually be restricted to the space of measures with squared-integrable density with respect to $\alpha \otimes \beta$:
\begin{equation*}
    \mathcal{M}^{L^2}(\alpha\otimes\beta) := \left\{ \pi \in \meas{\X \times \Y} ~~ \text{ s.t. } ~ \exists f \in L^2(\X \times \Y, \alpha\otimes\beta), \quad \pi = f \diff\alpha\otimes\beta  \right\}.
\end{equation*}
Since $L^2(\X \otimes \Y, \alpha\otimes\beta)$ is a Hilbert space, this allows us to diagonalize the kernel and prove:
\begin{restatable}{proposition}{existencediagqm}
\label{prop:existence_diag_qm}
If $\Reg$ contains an entropic term, $k$ is continuous and $\alpha,\beta$ have compact support, then there exists two linear operators $A, B : \mathcal{M}_+^{L^2}(\alpha\otimes\beta) \to L^2(\X \times \Y, \alpha\otimes \beta)$ such that:
\begin{equation*}
    \Reg\QM(\alpha, \beta) = \inf_{\pi \in \mathcal{M}_+^{L^2}(\alpha\otimes\beta)} \tfrac{1}{2} \norm{A[\pi]}^2 - \tfrac{1}{2} \norm{B[\pi]}^2 + \Reg(\pi).
\end{equation*}
with $\mathcal{M}_+^{L^2}(\alpha\otimes\beta)$ denoting the subset of positive measures of  $\mathcal{M}^{L^2}(\alpha\otimes\beta)$.
\end{restatable}
The setting is slightly different from the hypotheses of our article, but our duality results can be adapted to this $L^2$ case.
It evidences the naturality of our factorization property when working with entropic-regularized QM, suggesting that the reformulations we propose in \cref{section:quadratic_matching_case} are more than a numerical trick and display a fundamental property of quadratic matching.

\subsection{Detail of the Block-Separable Case}
\label{appendix:block_separable}

Theorem~\ref{thm:qot_as_registration} states the existence of embedding functions such that $\Reg\QM$ is equivalent to our new formula.
In the separable case, these embeddings correspond to the functions $\Phi, \Psi$ of \cref{definition:separable}, but the block-separable case is more involved: we explain how to build them from the block-decomposition of the operators $A$ and $B$. 
Following the notations of \cref{definition:separable}, we denote the embeddings of each block operator $A_i$ and $B_j$ as follows:
\begin{align*}
    A_{i}(x,y) = \Phi_{A,i}(x) \Psi_{A,i}(y)^\top, \\
    B_{j}(x,y) = \Phi_{B,j}(x) \Psi_{B,j}(y)^\top.
\end{align*}

\paragraph{The finite-dimensional case.}
In the finite-dimensional case, the vectors $\Phi_B(x)$, $\Psi_B(y)$, $\Phi_A(y)$ and $\Psi_B(y)$ are obtained by concatenating the corresponding block embeddings as follows:
\begin{equation}
\label{eq:block_embeddings}
\begin{split}
    \Phi_A(x) = \begin{pmatrix} \Phi_{A,1}(x) \\ \dots \\ \Phi_{A,\KK}(x) \end{pmatrix}, & \quad\quad
    \Psi_A(y) = \begin{pmatrix} \Psi_{A,1}(y) \\ \dots \\ \Psi_{A,\KK}(y) \end{pmatrix}, \\
    \Phi_{B}(x) = \begin{pmatrix} \Phi_{B,1}(x) \\ \dots \\ \Phi_{B,\LL}(x) \end{pmatrix}, & \quad\quad
    \Psi_B(y) = \begin{pmatrix} \Psi_{B,1}(y) \\ \dots \\ \Psi_{B,\LL}(y) \end{pmatrix}.    
\end{split}
\end{equation}
The space of linear operators $U$ corresponds to block-diagonal matrices $\Gamma = {\rm diag}(\Gamma_1,\dots, \Gamma_{\LL})$, and the scalar product of \cref{thm:qot_as_registration} coincides with the following block-wise matrix operation:
\begin{equation*}
\scal{\Gamma,B(x,y)} = \scal{\Phi_B(x), \Gamma \Psi_B(y)} = 
    \begin{pmatrix}
        \Phi_{B,1}(x) & \dots & \Phi_{B,\LL}(x) 
    \end{pmatrix}
    \begin{pmatrix}
        \Gamma_1 & \dots & 0 \\
        \vdots & \ddots & \vdots \\
        0 & \dots & \Gamma_{\LL}
    \end{pmatrix}
    \begin{pmatrix}
        \Psi_{B,1}(y) \\ \vdots \\ \Psi_{B,\LL}(y) 
    \end{pmatrix}.
\end{equation*}
Similarly, the convex potentials take the form $\Xi = {\rm diag}(\Xi_1,\dots, \Xi_{\KK})$, so that:
\begin{equation*}
\scal{\Xi,A(x,y)} = \scal{\Phi_A(x), \Xi \Psi_A(y)} = 
    \begin{pmatrix}
        \Phi_{A,1}(x) & \dots & \Phi_{A,\KK}(x) 
    \end{pmatrix}
    \begin{pmatrix}
        \Xi_1 & \dots & 0 \\
        \vdots & \ddots & \vdots \\
        0 & \dots & \Xi_{\KK}
    \end{pmatrix}
    \begin{pmatrix}
        \Psi_{A,1}(y) \\ \vdots \\ \Psi_{A,\KK}(y) 
    \end{pmatrix}.
\end{equation*}

\paragraph{Algebraic formalization.} We rigorously define this finite-dimensional intuition of block-wise decomposition and generalizes it to Hilbert-valued operators $A$ and $B$.
Given a sequence of Hilbert vectors $u_1 \in \Hilb_1,\dots, u_\LL \in \Hilb_\LL$ living in distinct spaces, we use the following matrix notation to denote their direct sum (the natural formalization of vector concatenation):
\begin{equation*}
   \begin{pmatrix} u_1\\ 
     \dots \\ 
    u_\LL \\
    \end{pmatrix} := u_1 \oplus \dots \oplus u_\LL ~~\in \Hilb_1 \oplus \dots\oplus \Hilb_\LL. 
\end{equation*}
It comes with the natural scalar product:
\begin{equation*}
    \scal{   \begin{pmatrix} u_1\\ 
     \dots \\ 
    u_\LL \\
    \end{pmatrix},  
    \begin{pmatrix} v_1\\ 
     \dots \\ 
    v_\LL \\
    \end{pmatrix} } = \scal{u_1,v_1} + \dots + \scal{u_\LL, v_\LL}.
\end{equation*}
For any $\Hilb_\X := \Hilb_{\X,1} \oplus \dots \oplus \Hilb_{\X,\LL}$ and $\Hilb_\Y := \Hilb_{\Y,1} \oplus \dots \oplus \Hilb_{\Y,\LL}$,
let us denote the set of \emph{block-wise} Hilbert-Schmidt operators from $\Hilb_\Y$ to $\Hilb_\X$ as:
\begin{equation*}
    \HS_{\oplus}(\Hilb_\Y, \Hilb_\X) := \HS(\Hilb_{\Y,1}, \Hilb_{\X,1}) \oplus \dots \oplus \HS(\Hilb_{\Y,n}, \Hilb_{\X,\LL}) \subset \HS(\Hilb_\Y, \Hilb_\X),
\end{equation*}
where any $\Gamma = \Gamma_1,\dots, \Gamma_n \in  \HS_{\oplus}(\Hilb_\Y, \Hilb_\X)$ acts linearly on $\Hilb_\Y$ as follows:
\begin{equation*}
\Gamma 
\begin{pmatrix}
    u_1 \\
    \dots \\
    u_\LL
\end{pmatrix} = 
\begin{pmatrix}
    \Gamma_1 u_1 \\
    \dots \\
    \Gamma_\LL u_\LL
\end{pmatrix} \in \Hilb_\X.
\end{equation*}
We now assume that $A$ and $B$ are block-separable in the sense of \cref{definition:separable}.
First, note that block-separability is equivalent to the existence of separable functions $A_i : \X \times \Y \to \Hilb_{A,i}$ and $B_j: \X \times \Y \to \Hilb_{B,j}$ such that, for all $x \in \X$ and all $y \in \Y$:
\begin{equation*}
 A(x,y) = \begin{pmatrix} A_1(x,y)\\ 
     \dots \\ 
    A_\KK(x,y) \\
    \end{pmatrix} 
    \quad \text{ and } \quad  B(x, y) = \begin{pmatrix} B_1(x,y) \\ 
     \dots \\ 
    B_\LL(x,y) \\
    \end{pmatrix}.
\end{equation*}
Since each $A_i$ and $B_j$ is separable, we can write for all $x \in \X$, $y \in \Y$ and all $1 \leq i \leq \KK$, $1 \leq j \leq \LL$:
\begin{align*}
    A_{i}(x,y) = \Phi_{A,i}(x) \Psi_{A,i}(y)^\top & &\text{ where } & &\Phi_{A,i}(x) \in \Hilb^{\Phi}_{A,i} && \text{ and } && \Psi_{A,i}(y) \in \Hilb^{\Psi}_{A,i}, \\
    B_{j}(x,y) = \Phi_{B,j}(x) \Psi_{B,j}(y)^\top && \text{ where } && \Phi_{B,j}(x) \in \Hilb^{\Phi}_{B,j} && \text{ and } && \Psi_{B,j}(y) \in \Hilb^{\Psi}_{B,j}.
\end{align*}
Therefore, for each $x \in \X$ and $y \in \Y$, we have:
\begin{equation}
\label{eq:block_separable_ops}
    A(x, y) = \begin{pmatrix} \Phi_{A,1}(x)\Psi_{A,1}(y)^{\top} \\ 
     \dots \\ 
    \Phi_{A,K}(x)\Psi_{A,K}(y)^{\top} \\
    \end{pmatrix} 
    \quad \text{ and } \quad  B(x, y) = \begin{pmatrix} \Phi_{B,1}(x)\Psi_{B,1}(y)^{\top} \\ 
     \dots \\ 
    \Phi_{B,L}(x)\Psi_{B,L}(y)^{\top} \\
    \end{pmatrix}.
\end{equation}
For all $x \in \X$ and $y \in \Y$, we define the concatenated embeddings $\Phi_A(x)$,$\Psi_A(y)$,$\Phi_{B}(x)$ and $ \Psi_B(y)$ as in \cref{eq:block_embeddings} (where the matrix notation now represent the direct sum of the rows).
With these definitions, \cref{thm:qot_as_registration} becomes:

\begin{restatable}{theorem}{qotasregistrationblockseparable}
\label{thm:qot_as_registration_blockseparable}
Let assume that $A$ and $B$ are block-separable with concatenated embeddings $\Phi_{A} \in \Hilb^{\Phi}_A$, $\Psi_{A} \in \Hilb^{\Psi}_A$, $\Phi_{B} \in \Hilb^{\Phi}_B$ and $\Psi_{B} \in \Hilb^{\Psi}_B$.
Let $\mathbf{R}$ and $\mathbf{S}$ be as in \cref{prop:outer_gamma}.
Then, 
\begin{equation*}
\Reg\QM(\alpha,\beta) ~~=~~  \min_{\Gamma \in \HS_{\oplus}(\Hilb^{\Psi}_B, \Hilb^\Phi_B)} ~~ \min_{\pi \in \measplus{\X \times \Y}}\; - \int \scal{\Phi_B(x), \Gamma \Psi_B(y)} \diff \pi(x,y) + \mathbf{R}(\pi) + \mathbf{S}(\Gamma).
\end{equation*}
Similarly, with the notations of \cref{theorem:tolanddualitygeneral}, the following equivalences hold:
\begin{multline*}
    \Reg\QM(\alpha, \beta)  ~~ = ~~ \inf_{\Gamma \in  \HS_{\oplus}(\Hilb^{\Psi}_B, \Hilb^\Phi_B)}~~ \sup_{\Xi \in  \HS_{\oplus}(\Hilb^{\Psi}_A, \Hilb^\Phi_A)} ~~ \sup_{\phi \in \Phi} ~~  \mathcal{G}(\Gamma, \Xi, \phi) \\
    \text{with} \quad  \mathcal{G}(\Gamma, \Xi, \phi) :=
            \mathcal{D}_\Reg\bigl(\phi;\, c_{\Gamma, \Xi},\, \alpha, \beta\bigr) + \tfrac{1}{2} \norm{\Gamma}^2 - \tfrac{1}{2} \norm{\Xi}^2
\end{multline*}
and
\begin{multline*}
    \Reg\QM(\alpha, \beta)  ~~ = ~~ \inf_{\Gamma \in  \HS_{\oplus}(\Hilb^{\Psi}_B, \Hilb^\Phi_B)}~~ \sup_{\Xi \in  \HS_{\oplus}(\Hilb^{\Psi}_A, \Hilb^\Phi_A)} ~~ \min_{\pi \in \measplus{\X \times \Y}} ~~ 
     \mathcal{F}(\Gamma, \Xi, \phi), \\
    \text{with} \quad \mathcal{F}(\Gamma, \Xi, \phi) := \int c_{\Gamma, \Xi} \diff\pi + \Reg(\pi) + \tfrac{1}{2} \norm{\Gamma}^2 - \tfrac{1}{2} \norm{\Xi}^2,
\end{multline*}
the effective cost $c_{\Gamma, \Xi}$ being equal to:
\begin{equation*}
    c_{\Gamma,\Xi}(x,y) = - \scal{\begin{pmatrix}
        \Phi_A(x) \\ \Phi_B(x) \end{pmatrix},\begin{pmatrix}
        -\Xi \Psi_A(y) \\ \Gamma \Psi_B(y) \end{pmatrix}}.
\end{equation*}
\end{restatable}

We provide the proof of this result in \cref{appendix:proofs}, along with the proof of \cref{thm:qot_as_registration}.
This formalism properly defines the block-wise operations that define the costs involved in this theorem.
It notably highlights the fact that the space of transforms on which optimization is performed only contains a small fraction of all linear operators $\HS(\Hilb^{\Psi}_A, \Hilb^\Phi_A)$ and $\HS(\Hilb^{\Psi}_B, \Hilb^\Phi_B)$: this allows for efficient numerical implementations that we detail in \cref{appendix:subsection:block_update_implementation}.

\subsection{Explicit decompositions of Gromov-Wasserstein kernels}
\label{appendix:gw_decomposition}

As discussed in \cref{appendix:general_matching}, for discrete inputs, eigendecomposition of the QM kernel is sufficient to obtain block-sparse embeddings.
However, this process is costly and intractable on inputs with moderate size.
For Gromov-Wasserstein and its variants, the kernel structure allows to compute embeddings without relying on its diagonalization, making GW particularly suited to computational applications.
In some settings, the block-separability also holds in the infinite case, evidencing the naturality of the block-separable formalism when applied to Gromov-Wasserstein.
In the following, we factorize GW kernels of the form:
\begin{equation*}
    k(x,y,x',y') =  \bigl(c_\X(x,x') - c_\Y(y,y')\bigr)^2,
\end{equation*}
depending on the properties of the base costs $c_\X$ and $c_\Y$.
Our computations rely on the following characterization lemma:
\begin{restatable}{lemma}{kernelfactorization}
    \label{lem:kernel_factorization}
    Let $k : \C((\X \times \Y)^2)$ be symmetric and let $A$, $B$ be two integral operators. The identity
    \begin{equation*}
        \frac{1}{2} \int k\,\diff\pi\otimes\pi \;=\; \frac{1}{2} \norm{A[\pi]}^2 - \frac{1}{2} \norm{B[\pi]}^2
    \end{equation*}
    holds for all $\pi \in \meas{\X \times \Y}$ if and only if, for all $x,x'\in \X$ and $y,y' \in \Y$:
    \begin{equation*}
        k(x,y,x',y') \;=\; \scal{A(x,y), A(x',y')} - \scal{B(x,y), B(x',y')}.
    \end{equation*}
\end{restatable}

\paragraph{GW with inner products and squared norms.} When $c_\X$ and $c_\Y$ are scalar products, $k$ is a polynomial.
After rearrangement of the monomials, we obtain the desired decomposition.

\begin{proposition}
    \label{prop:scalar_decomp}
      The scalar product kernel is equivalent to:
    \begin{multline*}
        \bigl(\scal{x, x'}^2 - \scal{y, y'}^2\bigr)^2 \;=\; \scal{A(x,y), A(x',y')} - \scal{B(x,y), B(x',y')} \\
        \text{ with } \quad 
        A(x,y) =
        \begin{pmatrix}
            xx^\top                              \\
            yy^\top                              \\
        \end{pmatrix}
         \text{ and }
        B(x,y) = \sqrt{2}\, xy^\top,
    \end{multline*}
    the matrix notation denoting the direct sum of the cells (as detailed in \cref{appendix:block_separable}).
\end{proposition}
The block embeddings of $A$ and $B$ are obtained by separating the $x$ and $y$ terms of the formula:
\begin{equation*}
    \Phi_A(x) = \begin{pmatrix}
        xx^\top \\ 1
    \end{pmatrix},
    \quad   
    \Psi_A(y) = \begin{pmatrix}
        1 \\ y y^\top
    \end{pmatrix},
    \quad \Phi_B(x) = \sqrt[4]{2} x 
    \quad \text{ and } \quad \Phi_B(y) = \sqrt[4]{2} y.
\end{equation*}

A similar decomposition can be obtained when costs are differences of scalar products:

\begin{proposition}
    \label{prop:general_decomp}
    Let $c_\X: \X \times \X \rightarrow \R$ and $c_\Y: \Y \times \Y \rightarrow \R$ such that for all $x,x' \in \X$ and $y,y' \in \Y$:
    \begin{equation*}
     c_\X(x,x') = \scal{x_+, x'_+} - \scal{x_-, x'_-} 
     \quad  \text{and} \quad 
     c_\Y(y,y') = \scal{y_+, y'_+} - \scal{y_-, y'_-}.
    \end{equation*}
    Then,
    \begin{multline*}
        \bigl(c_\X(x, x') - c_\Y(y,y')\bigr)^2 \;=\; \scal{A(x,y), A(x',y')} - \scal{B(x,y), B(x',y')} \\
        \text{ with } \quad 
    A(x, y) = \begin{pmatrix} x_+ x_+^\top \\ 
     x_- x_-^\top \\ 
    y_+ y_+^\top  \\
    y_- y_-^\top  \\
    \sqrt{2}x_- y_+^\top \\
    \sqrt{2} x_+ y_-^\top
    \end{pmatrix} 
    \quad \text{ and } \quad 
    B(x,y) = \begin{pmatrix} 
    \sqrt{2} x_+ x_-^\top  \\
    \sqrt{2} y_+ y_-^\top  \\
    \sqrt{2} x_+  y_+^\top \\
    \sqrt{2} x_- y_-^\top
\end{pmatrix}.
\end{multline*}
\end{proposition}    

Once again, we can build the embeddings $\Phi_A(x)$, $\Psi_A(y)$, $\Phi_B(x)$ and $\Psi_B(y)$ by separating the terms in $x$ and $y$ adequately.
Although difference of scalar products are never used explicitly in practice, most GW instances can be reduced to this cost structure.  
Notably, squared Euclidean norms can be rewritten in this form:
\begin{equation*}
    \norm{x - x'}^2 =  \scal{x_+, x'_+} - \scal{x_-, x'_-}, \quad \text{ with } \quad x_+ = \tfrac{1}{\sqrt{2}} (\norm{x}^2 + 1) 
    ~ \text{ and } ~ x_- = \begin{pmatrix}
        \sqrt{2} x \\ \tfrac{1}{\sqrt{2}} (\norm{x}^2 - 1)
    \end{pmatrix}.
\end{equation*}
Therefore, the squared norm case (where $c_\X, c_\Y = \norm{\cdot}^2$) admits an explicit block-separable factorization as well.

\begin{remark}[The balanced case]
When we restrict $\pi$ to the set $\coupling{\alpha,\beta}$ of couplings (defined in \cref{{eq:couplings}}), the kernel decompositions can be further simplified.
Indeed, most terms of $A$ and $B$ depend only on $x$ or on $y$ and remain constant when the marginals of $\pi$ are fixed. 
Since they do not play a role in the minimization, we can simplify the previous computations in the balanced case:
\begin{gather*}
 \frac{1}{2} \int (\scal{x,x'} - \scal{y,y'})^2\diff\pi\otimes\pi = {\rm Cste}(\alpha,\beta) - \norm{\int xy ^\top \diff\pi}^2, \\
 \frac{1}{2} \int (\norm{x-x'}^2 - \norm{y-y'})^2\diff\pi\otimes\pi = {\rm Cste}(\alpha,\beta) - 4 \norm{\int xy ^\top \diff\pi}^2 - 2 \int \norm{x}^2 \norm{y}^2 \diff\pi,
\end{gather*}
giving the same factorization as in \cite{rioux2024entropic,zhang2024gromov}.
Similarly, the objective of \cref{prop:general_decomp} becomes:
\begin{equation*}
    \frac{1}{2} \int (c_\X(x,x') - c_\Y(y,y'))^2\diff\pi\otimes\pi = {\rm Cste}(\alpha,\beta) + \norm{\int \begin{pmatrix} x_-y_+^\top \\ x_+y_-^\top \end{pmatrix} \diff\pi}^2 - \norm{\int \begin{pmatrix} x_+y_+^\top \\ x_-y_-^\top \end{pmatrix} \diff\pi}^2.
\end{equation*}

\end{remark}

\paragraph{Embeddable costs.} The previous decompositions are specific to scalar products or squared norms, but they can be extended to any costs that can be represented as a squared norm or scalar product in a certain embedding space.
We refer to such costs as \emph{embeddable costs}, that encompasses two important categories: costs of positive and conditionally of negative type. 
We briefly introduce them here, refering to \cite[Appendix C]{bekka2007kazhdan} for a detailed introduction to this rich theory.

\begin{definition}[Cost of positive type]
\label{def:positive_type}
A cost $c: \X \times \X \rightarrow \R^+$ is of positive type if it is symmetric and is such that for any finite collection of points $x_1,\dots, x_n \in \X$ and coefficients $\lambda_1,\dots,\lambda_n \in \R$:
\begin{equation*}
\sum_{i=1}^n \lambda_i \lambda_j c(x_i, x_j) \geq 0.
\end{equation*}
\end{definition}

\begin{definition}[Cost conditionally of negative type]
\label{def:cnt}
A non-negative cost $c: \X \times \X \rightarrow \R^+$ is conditionally of negative type (CNT) if it is symmetric, satisfies $c(x, x) = 0$ for all $x \in \X$ and 
is such that for any finite collection of points $x_1,\dots, x_n \in \X$ and coefficients $\lambda_1,\dots,\lambda_n \in \R$:
\begin{equation*}
  \sum_{i=1}^n \lambda_i = 0 ~~\Longrightarrow~~ \sum_{i=1}^n \lambda_i \lambda_j c(x_i, x_j) \leq 0~.
\end{equation*}
\end{definition}

Many important examples fall in these categories.
First, note that these categories are stable by sum and rescaling by positive coefficients.
Moreover, if two costs $c_1, c_2$ are positive, then $c_1 \cdot c_2$ is positive as well: in particular, any non-negative integer power $c^k$ of a positive cost is positive.
On the other hand, tree distances, hyperbolic geodesic distances, spherical distances and all (non-squared) Euclidean distances are CNT.
If $c$ is CNT, then $c^p$ is CNT as well if $0 < p \leq 2$: in particular, $c(x,x') = \norm{x-x'}^p$ if CNT for any positive $p$ smaller than $2$.
Lastly, CNT and positive costs are connected by the following characterization:

\begin{theorem}[\cite{schoenberg1938metric}]
Let $c \in \C(\X \times \X)$ be symmetric and satisfies $c(x,x) = 0$ for all $x \in \X$.
Let $x_0 \in \X$.
The following conditions are equivalent:
\begin{enumerate}
    \item The cost $c$ is CNT,
    \item The cost $\Tilde{c}: x,x' \in \X\times\X \mapsto c(x,x_0) + c(x_0,x') - c(x,x')$ is positive,
    \item The cost $e^{-\lambda c}: x,x' \in \X\times\X \mapsto \exp(-\lambda c(x,x'))$ is positive for every $\lambda \geq 0$.
\end{enumerate}
\end{theorem}

A remarkable property of these costs is that they are indeed embeddable:

\begin{theorem}[\cite{aronszajn1950theory}]
\label{theorem:positive_embedding}
Let $c \in \C(\X \times \X)$ be symmetric. The cost $c$ is positive if and only if there exists a real Hilbert space $\Hilb$ with a continuous mapping $\varphi : \X \rightarrow \Hilb$ such that:
\begin{equation*}
   \text{For all } x,x' \in \X,  \quad c(x, x') ~=~ \scal{\varphi(x), \varphi(x')}_{\Hilb}.
\end{equation*}
\end{theorem}

\begin{theorem}[\cite{schoenberg1938metric}]
\label{theorem:cnt_embedding}
Let $c \in \C(\X \times \X)$ be symmetric and satisfies $c(x,x) = 0$ for all $x \in \X$. The cost $c$ is CNT if and only if there exists a real Hilbert space $\Hilb$ and a continuous mapping $\varphi : \X \rightarrow \Hilb$ such that:
\begin{equation*}
    \text{For all } x,x' \in \X,     \quad c(x, x') ~=~ \norm{\varphi(x) - \varphi(x')}_{\Hilb}^2.
\end{equation*}
\end{theorem}

Although these embeddings are potentially infinite-dimensional, when working with a finite number of points, we can always restrict them to a finite-dimensional subset of the initial Hilbert space. 
In practice, they can be computed by diagonalizing the cost matrix defined by $c(x,x')$ in the positive case or by $c(x,x_0) + c(x_0,x') - c(x,x')$ in the CNT case.
For inputs of size $\NN$, these embeddings are $\NN$-dimensional.
Their dimension can then be reduced using approximation techniques, as we explain in \cref{appendix:subsection:acceleration_techniques}.

\paragraph{Arbitrary costs on finite spaces.} When $\X$ and $\Y$ are finite, the spectral theorem allows us to reexpress any arbitrary cost as a difference of scalar product, allowing us to apply the computations we made in the Euclidean case:
\begin{proposition}
\label{prop:diag_costs}
    Let assume that $\X$ is finite, and let $c: \X \times \X \longrightarrow \R$.
    Then, there exists two functions $\varphi_+ : \X \longrightarrow \R^{D_+}$ and $\varphi_- : \X \longrightarrow \R^{D_-}$ such that:  
    \begin{equation*}
        \text{For all } x, x' \in \X, \quad c(x,x') = \scal{\varphi_+(x), \varphi_+(x')} - \scal{\varphi_-(x), \varphi_-(x')}.
    \end{equation*}
\end{proposition}
Therefore, up to the diagonalization of the costs $c_\X$ and $c_\Y$, any finite instance of Gromov-Wasserstein fits in the setting of \cref{prop:general_decomp}.
Unlike diagonalizing the operator $k$ explicitely, diagonalizing the costs is numerically tractable and has a negligible computational cost compared to solving the matching problem: this result is the key to efficiently apply our solver to GW. 

\paragraph{Weighted kernels.} We finally provide an extension of our computations to reweighted kernels: 
\begin{equation*}
    k_w(x,y,x',y') = \scal{w_\X(x), w_\X(x')}  \scal{w_\Y(y), w_\Y(y')}  k(x,y,x',y'),
\end{equation*}
where $k$ is block-separable.
The additional scalar terms adjust the matching penalization based on an affinity score between $x \leftrightarrow x'$ and $y \leftrightarrow y'$, offering some flexibility in the modeling of GW problems and allowing to prioritize local geometry more than plain GW.
In notably models the case of frature matching of \cref{fig:fracture_matching}, where $w_\X(x) = [\one_{x \in C_k}]_k$ is a hot encoding of the fragment $C_k$ to which each point $x$ belongs (and without weighting on the points of the non-fractured target $\Y$).

The embeddings of $k_w$ can be obtained from those of $k$ by taking the tensor product of each block coordinate with the weight vector $w_\X$ and $w_\Y$:
\begin{gather*}
   \Phi_A(x) = \begin{pmatrix}
        \Phi_{A,1}(x) \\ \dots \\ \Phi_{A,\KK}(x)
    \end{pmatrix}
    \quad \implies \quad
    \Phi_{A,w}(x) =  \begin{pmatrix}
        \Phi_{A,1}(x)w_\X(x)^\top \\ \dots \\ \Phi_{A,\KK}(x) w_\X(x)^\top
    \end{pmatrix}, \\
    \Phi_B(x) = \begin{pmatrix}
        \Phi_{B,1}(x) \\ \dots \\ \Phi_{B,\LL}(x)
    \end{pmatrix}
    \quad \implies \quad
    \Phi_{B,w}(x) =  \begin{pmatrix}
        \Phi_{B,1}(x)w_\X(x)^\top \\ \dots \\ \Phi_{B,\LL}(x) w_\X(x)^\top
    \end{pmatrix},
\end{gather*}
the embeddings $\Psi_{A}(y)$ and $\Psi_{B}(y)$ being defined similarly.
More generally this construction also provides a recipe to compute the embeddings of product kernels $k' = k_1 k_2$, where both $k_1$ and $k_2$ are block-separable, since any block-separable kernel can be expressed as a sum of terms of the form $ \scal{w_\X(x), w_\X(x')}  \scal{w_\Y(y), w_\Y(y')} =  \scal{w_\X(x)w_\Y(y)^\top, w_\X(x')w_\Y(y')^\top}$.

\section{Detailed Implementation}
\label{appendix:detailed_implementation}
This section is dedicated to the practical implementation of our proposed solver.
Specifically, we detail how the high-level description of \cref{fig:schema_algo} translates into practical and modular algorithm, how to perform computations as efficiently as possible, and how to initialize the potentials $\Gamma_0, \Xi_0$ in practice.
 
In the following, we assume that the inputs measures $\alpha$ and $\beta$ are discrete with $\NN$ and $\MM$ points, respectively:
\begin{equation*}
    \alpha = \sum_{i=1}^{\NN} \alpha_i \delta_{x_i} \quad \text{ and } \quad \beta = \sum_{j=1}^{\MM} \beta_j \delta_{y_j}.
\end{equation*}
We also assume that the Hilbert spaces $\Hilb_A$ and $\Hilb_B$ are Euclidean, since the spaces $\img{A}$ and $\img{B}$ are generated by the images of the finite-dimensional space $\meas{\X \times \Y} \cong \R^{\NN \times \MM}$.

\subsection{Fast Implementation of Block Updates}
\label{appendix:subsection:block_update_implementation}

In the main article, we explain how block-separable kernels allow for efficient and modular implementations of our solver.
We provide detail about this fact here, explaining how to exploit the block structure to keep computations as fast as possible.

\paragraph{The separable case.} Before treating the general case, we first detail the algorithms when the operators $A$ and $B$ are separable to provide the intuition of this approach.
If these operators satisfy:
\begin{equation*}
    A(x,y) = \Phi_A(x)\Psi_A(y)^\top \quad \text{ and } \quad B(x,y) = \Phi_B(x)\Psi_B(y)^\top,
\end{equation*}
then for any plan $\pi$, $A[\pi] = {\rm Corr}_A[\pi]$ and $B[\pi] = {\rm Corr}_B[\pi]$ (following the notations of \cref{eq:primal_corr}).
Moreover, all $\Reg\OT$ computations take the scalar product form $c_{\Gamma,\Xi}(x,y) = - \scal{x_{\rm eff}, y_{\rm eff}}$ where $x_{\rm eff}$ and $y_{\rm eff}$ are given by:
\begin{equation*}
   x_{\rm eff} = \begin{pmatrix}
        \Phi_\Cvxa(x) \\ \Phi_\Cvxb(x) \end{pmatrix}
                \quad \text{ and } \quad 
                y_{\rm eff}=
    \begin{pmatrix} 
        -\Xi \Psi_\Cvxa(y) \\ \Gamma \Psi_\Cvxb(y) \end{pmatrix}.
\end{equation*}
\Cref{alg:ugw} only interacts with the $\Reg\OT$ solver through these parameters $x_{\rm eff}$, $y_{\rm eff}$ and through the cross-correlation operations involved in $A[\pi]$, $B[\pi]$.
Therefore, we can encode all these operations in a black-box $\Reg\OT$ solver specialized to costs of the form $c(x,y) = -\scal{x, y}$ that we denote by \texttt{Sink} (by reference to the Sinkhorn algorithm that serves as our $\Reg\OT$ solver in our practical implementation). 
Our method then takes the form of \cref{alg:ugw_separable}, 
\texttt{Sink.optimization\_step()} runs a few warm-started iterations of this solver and $\texttt{Sink.Corr(}\Phi, \Psi\texttt{)}$ returns $\int \Phi(x)\Psi(y)^\top \diff\pi(x,y)$.
Note that, in practice, OT solvers do not directly implement such cross-correlations, but rather transport integrals of the form:
\begin{equation}
\label{eq:transport_integrals}
    v(x) = \int u(y) \diff\pi(x,y) \quad \text{ or } \quad   v(y) = \int u(x) \diff\pi(x,y).
\end{equation}
Therefore, to compute the cross-correlation between two embedding vectors $\Phi(x) \in \R^\DD$ and $\Psi(y) \in \R^\EE$, two choices are available:
\begin{equation*}
    \int_\X \Phi(x) \left( \int_\Y (\Psi(y)^\top \diff \pi(x,y)) \right) \quad \text{ or } \quad  \int_\Y \Psi(y)^\top \left( \int_\X (\Phi(x) \diff \pi(x,y)) \right).
\end{equation*}
The inner integral bearing the main computational cost, it is more efficient to perform this integral on the vector with smallest dimension.
We detail this procedure in \cref{alg:cross_corr}, where \texttt{Sink.integrate} implements the integrals of \cref{eq:transport_integrals} (and is performed internally by the chosen OT backend).
When using linear-memory OT backends, the cross-correlation computation has a complexity of $O(\NN \MM \min(\DD, \EE) + (\NN + \MM) \DD \EE)$ in time and $O(\NN \DD + \MM \EE)$ in memory.

\begin{algorithm}[htb]
    \caption{$\Reg\QM$ solver (separable formulation)}
    \label{alg:ugw_separable}
    \KwIn{as in \cref{alg:ugw}, with $A$, $B$ given by $\Phi_A,\Psi_A$ and $\Phi_B, \Psi_B$; inner iterations $n_{inner}$.}
    \KwOut{$(\Gamma, \Xi)$ and the solver \texttt{Sink}, whose potentials encode the plan $\widehat\pi$.}
    Initialize $\Gamma_0$, $\Xi$ and the $\Reg\OT$ solver \texttt{Sink}\;
    $\texttt{Sink.x}\gets (\Phi_A(x), \Phi_B(x))^\top$\;
    \For{$t = 0, 1, 2, \dots$ until $\norm{\Gamma_{t+1} - \Gamma_t} \le \delta$}{
     \For{$l = 1, \dots, n_{inner}$}{
        $\texttt{Sink.y} \gets (-\Xi\Psi_A(y), \Gamma_t\Psi_B(y))^\top$\; \texttt{Sink.optimization\_step()}\;
        $\Xi \gets \Xi + \tau\,\bigl(\texttt{Sink.Corr(}\Phi_A(x), \Psi_A(y)\texttt{)} - \Xi\bigr)$\;
        }
        $\Gamma_{t+1} \gets \texttt{Sink.Corr(} \Phi_B(x), \Psi_B(y) \texttt{)}$\;
    }
\end{algorithm}

\begin{algorithm}[H]
    \caption{$\texttt{Sink.Corr}$ implementation.}
    \label{alg:cross_corr}
    \KwIn{Sinkhorn solver \texttt{Sink}; $\Phi \in \R^{\NN \times \DD}$; $\Psi \in \R^{\MM \times \EE}$.}
    \KwOut{$\Omega \in \R^{\DD \times \EE}$, cross-correlation of $\Phi$ and $\Psi$ for the current transport plan of \texttt{Sink}.} 
    \uIf{$\DD \leq \EE$}{
        $\overline{\Phi} \gets \texttt{Sink.integrate(} \Phi,~ \text{axis}=0 \texttt{)}$\;
        $\Omega \gets \sum_i \overline{\Phi}_{i} \Psi_i^\top $\;
    } 
    \uElse{
         $\overline{\Psi} \gets  \texttt{Sink.integrate(} \Psi,~\text{axis}=1 {)}$ \;
         $\Omega \gets \sum_j \Phi_j \overline{\Psi}_j^\top $\;
    }
\end{algorithm}

\Cref{alg:ugw_separable} shows that, in the separable case, all computations involve either standard matrix-vector operations or standard interactions with the black-box \texttt{Sink} solver, demonstrating the modularity of our method as advertised in the main article.
The block-separable case yields heavier calculations, but the structure is similar.
We explain it in the following, where we assume that $A$ and $B$ have the structure of \cref{eq:block_separable_ops} and $D_{A,i}$, $E_{A,i}$, $D_{B,j}$, $E_{B,j}$ are the dimensions of $\Phi_{A,i}$, $\Psi_{A,j}$, $\Phi_{B,j}$, $\Psi_{B,j}$.

\paragraph{Block cross-correlations.} The operators $A$ and $B$ can be expressed as block-wise cross-correlation operators:
\begin{equation*}
    A[\pi] = \mbox{Block-Corr}_A[\pi] := \begin{pmatrix} \Corr_{A,1}[\pi]\\ 
     \dots \\ 
    \Corr_{A,K}[\pi] \\
    \end{pmatrix} \quad \text{ and } B[\pi] = \mbox{Block-Corr}_B[\pi] := \begin{pmatrix} \Corr_{B,1}[\pi]\\ 
     \dots \\ 
    \Corr_{B,L}[\pi] \\
    \end{pmatrix}, 
\end{equation*}
where for all $i$ and $j$,
\begin{equation*}
\Corr_{A,i}[\pi] = \int \Phi_{A,i}(x) \Psi_{A,i}(y)^{\top}\diff\pi(x,y) \quad \text{ and } \quad \Corr_{B,j}[\pi] = \int \Phi_{B,j}(x) \Psi_{B,j}(y)^{\top}\diff\pi(x,y).
\end{equation*}
Therefore, each block of $A[\pi]$ and $B[\pi]$ can be computed using the same standard operation as for separable costs.

\paragraph{Effective cost parametrization.} Similarly, the effective cost $c_{\Gamma,\Xi}$ can be computed block-wise. 
Rather than applying the potentials $\Xi$ and $\Gamma$ on only one of the embeddings $\Phi_A(x)$ or $\Psi_B(y)$ as in \cref{eq:effective_cost_decomp}, we can split the computation over both sides so that: 
\begin{multline*}
    \scal{\Xi, A(x,y)} = \scal{\begin{pmatrix} \Xi_{\X,1} \Phi_{A,1}(x) \\ 
     \dots \\ 
    \Xi_{\X,K} \Phi_{A,K}(x) \\
    \end{pmatrix} , \begin{pmatrix} \Xi_{\Y,1} \Psi_{A,1}(y) \\ 
     \dots \\ 
    \Xi_{\Y,K} \Psi_{A,K}(y) \\
    \end{pmatrix}} \\ \text{ for any } \quad \Xi_{\X,i}, \Xi_{\Y,i} \quad \text{ such that } \quad \Xi = \begin{pmatrix} \Xi_{\X,1}^\top \Xi_{\Y,1}\\ \dots \\ \Xi_{\X,K}^\top \Xi_{\Y,K} \end{pmatrix}, 
\end{multline*}
with a similar formula for $ \scal{\Gamma, B(x,y)}$.
As before, we can write $c_{\Gamma,\Xi}(x,y) = - \scal{x_{\rm eff}, y_{\rm eff}}$, with
\begin{multline*}
x_{\rm eff} = \mbox{Block-Proj}_{\Gamma_\X, \Xi_\X} (\Phi_A(x), \Phi_B(x)) :=  \begin{pmatrix} \Xi_{\X,1} \Phi_{A,1}(x)\\ \dots \\ \Xi_{\X,K} \Phi_{A,K}(x) \\  \Gamma_{\X,1} \Phi_{B,1}(x) \\ \dots \\  \Gamma_{\X,L} \Phi_{B,L}(x) \end{pmatrix} \\  
\text{ and } \quad y_{\rm eff} = \mbox{Block-Proj}_{\Gamma_\Y, \Xi_\Y} (\Psi_A(y), \Psi_B(y)) :=  \begin{pmatrix} -\Xi_{\Y,1} \Psi_{A,1}(y)\\ \dots \\ -\Xi_{\Y,K} \Psi_{A,K}(y) \\  \Gamma_{\Y,1} \Psi_{B,1}(y) \\ \dots \\  \Gamma_{\Y,L} \Psi_{B,L}(y) \end{pmatrix}.
\end{multline*}
The decomposition of $\Xi$ and $\Gamma$ on both $x$ and $y$ allows to minimize the dimension of the vectors $x_{\rm eff}$ and $y_{\rm eff}$, since the cost of each Sinkhorn iteration scales linearly with this dimension.
For each block, we therefore choose $\Xi_{\X,i} = \Xi_i^{\top}$ and $\Xi_{\Y,j} = 1$  if $D_{A,i} > E_{A,i}$,  $\Xi_{\X,i} = 1$ and $\Xi_{\Y,j} = \Xi_i$ otherwise.
Using these definitions, our method now takes the form of \cref{alg:ugw_sinkhorn_block}.
 The complexity of each operation is equal to $O(N^2 (D_A \vee E_A + D_B \vee E_B))$ in time and $O(N (D_A \vee E_A + D_B \vee E_B))$ in memory, with:
\begin{equation*}
  D_A \vee E_A = \sum_j \min(D_{A,i}, E_{A,i}) \quad \text{ and } \quad   D_B \vee E_B = \sum_j \min(D_{B,j}, E_{B,j}), 
\end{equation*}
avoiding the quadratic dependency in the dimension of naive implementations.

\begin{algorithm}[H]
    \caption{$\Reg\QM$ solver (block-separable formulation)}
    \label{alg:ugw_sinkhorn_block}
    \KwIn{as in \cref{alg:ugw}, with $A$, $B$ given by $\Phi_A,\Psi_A$ and $\Phi_B, \Psi_B$; inner iterations $n_{inner}$.}
    \KwOut{$(\Gamma, \Xi)$ and the solver \texttt{Sink}, whose potentials encode the plan $\widehat\pi$.}
    Initialize $\Gamma_0$, $\Xi$ and the $\Reg\OT$ solver \texttt{Sink}\;
    \For{$t = 0, 1, 2, \dots$ until $\norm{\Gamma_{t+1} - \Gamma_t} \le \delta$}{
     \For{$l = 0, 1, 2, \dots n_{inner}$}{
        $\texttt{Sink.x} \gets \mbox{Block-Proj}_{\Gamma_k, \Xi_k}(\Phi_A(x), \Phi_B(x))$\;
        $ \texttt{Sink.y} \gets \mbox{Block-Proj}_{\Gamma_k, \Xi_k}(\Psi_A(y), \Psi_B(y))$\;
        \texttt{Sink.optimization\_step()}\;
         $\Xi \gets \Xi + \tau\,\bigl(\texttt{Sink.Block-Corr(}\Phi_A(x), \Psi_A(y)\texttt{)} - \Xi\bigr)$\;
        }
        $\Gamma_{k+1} \gets \texttt{Sink.Block-Corr(}\Phi_B(x), \Psi_B(y)\texttt{)}$\;
    }
\end{algorithm}

\begin{algorithm}[H]
    \caption{$\texttt{Sink.Block-Corr}$ implementation.}
    \label{alg:block_cross_corr}
    \KwIn{Sinkhorn solver \texttt{Sink}; $\Phi =(\Phi_1,\dots, \Phi_K)$; $\Psi = (\Psi_1, \dots, \Psi_K)$.}
    \KwOut{$\Omega = (\Omega_1,\dots, \Omega_K)$.} 
    \For{$k = 0, \dots, K$}{
        $\Omega_k \gets \texttt{Sink.Corr(} \Phi_k, \Psi_k\texttt{)}$\;
    }
    $\Omega \gets (\Omega_1,\dots, \Omega_K)$\;
\end{algorithm}

\subsection{Advanced Implementation Techniques} 
\label{appendix:subsection:acceleration_techniques}

\paragraph{Further memory optimization.} In \cref{prop:general_decomp}, some blocks embeddings of $A$ and $B$ are terms of the form $f(x)g(x)^\top$ ($x_+ x_+^\top$, $x_+ x_-^\top$, etc). 
In \cref{alg:ugw_sinkhorn_block}, the only operations involving these embeddings take the form:
\begin{equation*}
\int f(x)g(x)^\top \diff\pi(x, y) = \int f(x)g(x)^\top \diff\pi_1(x) 
\quad \text{ and } \quad 
f(x)g(x)^\top \Xi = f(x) \bigl(g(x)^\top \Xi \bigr).
\end{equation*}
These are standard computations and efficiently implemented in scientific computing libraries, we can compute such integrals when needed without storing the tensor product explicitly, simply representing the corresponding embedding as a tuple $(f(x), g(x))$.
This approach reduces the memory footprint of these embeddings from $\NN\dim(f)\dim(g)$ to $\NN (\dim(f) + \dim(g))$, significantly increasing scalability when $f$ and $g$ are high dimensional.
Using the same techniques on $y$ tems, this allows our Gromov-Wasserstein implementation to scale linearly in the input dimension for the Euclidean cases of \cref{prop:scalar_decomp,prop:general_decomp}.

\paragraph{Low-dimensional embedding computations.} The embeddings $\Phi_A(x)$, $\Psi_A(y)$, $\Phi_B(x)$ and $\Psi_B(y)$ of \cref{appendix:gw_decomposition} are usually not known explicitly but can be obtained from the diagonalization of the cost matrices $[c_\X(x_i,x_k)]_{ik}$ and $[c_\Y(y_j,y_l)]_{jl}$.
In general, these embeddings are high-dimensional and unsuited to large-scale applications. 
It is therefore beneficial to rely on dimension reduction techniques, such as truncated PCA.
The KeOps library \citep{charlier2021kernel} allows to perform truncated eigendecompositions efficiently and with linear memory footprint, preventing any memory overflow.
Alternatively, approximations methods like Generalized Nystr\"om \citep{nakatsukasa2020fast} provide low-dimensional approximations of these embeddings in linear time and memory.
The cost of these methods is negligible compared to the rest of our pipeline.

\paragraph{Acceleration with multiscaling.} The optimal potentials $\Gamma$ and $\Xi$ depend on the geometry of the inputs $\alpha$ and $\beta$, rather than their precise pointwise configuration.
Therefore, optimal potentials $(\Gamma^*_c, \Xi^*_c)$ obtained by replacing the inputs with coarse approximations $\alpha_c \approx \alpha$ and $\beta_c \approx \beta$ remain close to the true optimum $(\Gamma^*, \Xi^*)$.
This motivates a natural multiscale strategy: solve the coarse problem, then initialize the fine-scale problem using the coarse output $(\Gamma^*_c, \Xi^*_c)$.
Similarly, standard OT multiscaling techniques \citep{feydy2020fast} can be applied in order to extrapolate an initialization $\phi_0$ from the coarse dual parameters $\phi^*_c$.
This multiscaling scheme substantially reduces the number of iterations needed for convergence.
To construct the coarse approximation of the input, three choices are available.
We can take a random subsample of the original points, use a k-means clustering to approximate the fine-scale measures, or use a predefined subset of the initial inputs.

\paragraph{Mask-based sparsification.} In practice, optimal plans concentrate their mass on small fractions of all possible pairs, most of their coefficient being negligible.
Some OT solvers exploit this behavior by accepting a sparsity mask that indicates which coefficients are expected to be nearly zero: they ignore them during computation, significantly fastening the resolution of the problem.
In classical OT, two points that are far apart are unlikely to be matched, so that sparsity masks can be computed based on the position of points in the space.
For quadratic matching, this technique does not apply directly since the two inputs generally live in uncomparable spaces. 
However, the multiscale strategy provides a natural alternative: since the coarse plan $\pi^*_c$ approximates $\pi^*$, a near-zero entry $\pi^*_c(x_c,y_c) \approx 0$ indicates that any fine-scale pair $(x,y)$ with $x \approx x_c$ and $y \approx y_c$ is also unlikely to be matched.
This allows sparsity masks to be estimated from the coarse solution, further accelerating the fine-scale solution.
We describe the exact procedure in \cref{alg:sparse_mask}.

 \begin{algorithm}[H]
     \caption{Sparsity mask computation}
     \label{alg:sparse_mask}
     \KwIn{Coarse inputs $x^c_i, y^c_j$, fine inputs $x_i, y_j$; optimal coarse potentials $\Xi^c$, $\Gamma^c$; optimal coarse plan $\pi^c$; tolerance parameter $\delta > 0$.}
     \KwOut{Sparsity mask $M \in \{0, 1\}^{\NN \times \MM}$.}
     $\Tilde{x} \gets (\Phi_A(x), \Phi_B(x))^\top$\;
     $\Tilde{y} \gets (-\Xi^c\Psi_A(y), \Gamma^c\Psi_B(y))^\top$\;
     $\Tilde{x}^c \gets (\Phi_A(x^c), \Phi_B(x^c))^\top$\;
     $\Tilde{y}^c \gets (-\Xi^c\Psi_A(y^c), \Gamma^c\Psi_B(y^c))^\top$\;
     $I \gets [\arg\min_u \norm{\Tilde{x}_i - \Tilde{x}^c_u}^2]_i$\;
     $J \gets [\arg\min_v \norm{\Tilde{y}_j - \Tilde{y}^c_v}^2]_j$\;
     $M \gets [\pi^c_{I(i)J(j)} > \delta]_{ij}$.
 \end{algorithm}

\subsection{Initialization of potentials}
\label{appendix:initialization_potentials}

The initialization of the solver is a crucial element for matching problems. 
Indeed, since $\Reg\QM$ is non-convex, the output of our solver depends on the initial potential $\Gamma_0 \in \Hilb_B$.
Moreover, although the initialization of $\Xi_0$ does not influence the local minimum reached by our method, it can significantly accelerate convergence and must also be taken into account.
Our framework allows for various choices of initializations that we detail here.
We provide the formulas for $\Gamma_0$, which can be straightforwardly adapted to initializations of $\Xi_0$.
We focus on the separable case, since each block of a block-separable potential can be initialized using these formulas.  

\paragraph{Trivial initialization.} The classical choice in the literature is to initialize the optimization at the trivial plan $\pi = \alpha \otimes \beta$. 
We then choose $\Gamma_0$ as the corresponding cross-correlation:
\begin{equation*}
    \Gamma_0 = \left( \int \Phi(x)\diff\alpha(x) \right) \left( \int \Psi(y) \diff\beta(y)\right)^\top.
\end{equation*}
This choice has the advantage to not induce any bias in the optimization trajectory.
Empirically, it generally provides satisfying results in the balanced setting but tends to be poorly suited to the unbalanced or in challenging settings (e.g. when inputs have many symmetries).

\paragraph{Landmarks.} If some correspondences are already known between some remarkable points, or \emph{landmarks}, we can exploit this information to initialize our algorithm more accurately. 
Given a sequence $(x_1,y_1),\dots, (x_K, y_K)$ of landmarks, we simply set:
\begin{equation*}
    \Gamma_0 = \frac{1}{K} \sum_i \Phi(x_i) \Psi(y_i)^\top.
\end{equation*}
This initialization guides the optimization towards a plan that matches each $x_i$ to a point close to $y_i$.

\paragraph{Random initialization.} When no prior information is known about the optimal matching, the optimization landscape can be explored by sampling the space of possible $\Gamma$ and keep the best solution over several random initializations.
Let $\Phi(x) \in \R^\DD$ and $\Phi(y) \in \R^\EE$.
Let $M(\Phi), M(\Psi)$ and $V(\Phi), V(\Psi)$ be the mean and  coordinate-wise variance of $\Phi, \Psi$ for the measure $\alpha, \beta$:
\begin{gather*}
M(\Phi) = \int \Phi \diff\alpha \quad \text{ and } \quad V(\Phi) = \left[ \int (\Phi_d - M(\Phi_d))^2 \diff\alpha \right]_{1 \leq d \leq \DD} \\
M(\Psi) = \int \Psi \diff\beta \quad \text{ and } \quad V(\Psi) = \left[ \int (\Psi_e - M(\Psi)_e)^2 \diff\beta \right]_{1 \leq e \leq \EE}
\end{gather*}
For any $\Gamma \in \Hilb_B$, if $\Gamma$ is the cross-correlation matrix of a coupling $\pi \in \coupling{\alpha, \beta}$, then: 
\begin{equation*}
\text{For all }  1\leq d \leq \DD \quad \text{and} \quad 1\leq e \leq \EE, \quad   \abs{\Gamma_{de} - M(\Phi)_d M(\Psi)_e } \leq \sqrt{V(\Phi)_d V(\Psi)_e}. 
\end{equation*}
Therefore, we can sample each coordinate $[\Gamma_0]_{de}$ uniformly from the interval $[M(\Phi)_d M(\Psi)_e \pm \sqrt{V(\Phi)_d V(\Psi)_e}]$.
These bounds hold in the balanced case $\pi \in \coupling{\alpha, \beta}$. 
They remain true in the unbalanced setting as well if we know that optimal plans have smaller margins than the inputs $\alpha, \beta$, i.e. they can destroy mass but not create it. 

\section{Proofs}
\label{appendix:proofs}

\subsection{Proofs of section \ref{section:measure_matching}}

\thmccvdualitygeneral*

\begin{proof}
Since $\Cvxb$ is a proper convex lower semi-continuous function, it is equal to its biconjugate: by the Fenchel-Moreau theorem, $\Cvxb^{**} = \Cvxb$.
Therefore, 
\begin{equation*}
\text{For all } \pi \in \meas{\X \times \Y}, \quad  \Cvxb(\pi) = \sup_{c_\Cvxb \in \C(\X \times \Y)} \int c_\Cvxb(x,y)\diff\pi(x,y) - \Cvxb^*(c_\Cvxb).
\end{equation*}
Consequently,
\begin{equation*}
\Cvxa(\pi) - \Cvxb(\pi) + \Reg(\pi) =  \inf_{c_\Cvxb \in \C(\X \times \Y)} - \int c_\Cvxb(x,y)\diff\pi(x,y) + \Cvxa(\pi) + \Reg(\pi) + \Cvxb^*(c_\Cvxb),
\end{equation*}
providing the first equation of the result.
If $\X \subset \R^{\DD}$ and $\Y \subset \R^{\EE}$ are compact subset of Euclidean spaces, then $\X \times \Y$ is a compact subset of $\R^{\DD \times \EE}$.

To prove the second part of the theorem, we first use the fact that sums of separable functions $x,y \mapsto \sum_i f_i(x)g_i(y)$ are dense in $\C(\X \times \Y)$.
Moreover, by the  Stone--Weierstrass theorem, polynomial functions are dense in $\C(\X)$ and in $\C(\Y)$ for the uniform convergence.
Since we assumed that $\Cvxb^*$ is continuous for the uniform convergence, we can restrict the infimum to polynoms of the form:
\begin{equation*}
    \text{For all } x \in \X \text{ and } y \in \Y, \quad p_\Gamma(x,y) = \sum_{d,e} \Gamma_{d,e} \Phi_{d}(x) \Psi_{e}(y),
\end{equation*}
where for all sequences of exponents $d = (d_1,\dots,d_\DD)$ and $e = (e_1,\dots,e_\EE)$, we define:
\begin{equation*}
    \Phi_{d} : x \in \X \mapsto x_1^{d_1}\dots x_\DD^{d_\DD} \in \R \quad \text{ and } \quad \Psi_{e} : y \in \Y \mapsto y_1^{e_1}\dots y_\EE^{e_\EE} \in \R,
\end{equation*}
and $\Gamma_{d,e}$ is a collection of real values with only a finite number of non-zero elements. 
Since the sequences involved in the previous decompositions are countable, and up to a reweighting of the terms $\Gamma_{d,e}$, we can write the sequences of monomials as sequences:
\begin{equation*}
    \Phi(x) = \left(\frac{1}{d_1!\dots d_\DD!} \Phi_d(x) \right)_{d \in \N_0^\DD} \quad \text{ and }  \quad \Phi(y) = \left(\frac{1}{e_1! \dots e_\EE!} \Psi_e(y) \right)_{e \in \N_0^{\EE}},
\end{equation*}
where the reweighting ensures that these sequences are squared-integrable.
Therefore, for all $x$ and $y$, $\Phi(x), \Phi(y) \in l^2(\R)$ where $l^2(\R)$ denotes the Hilbert space:
\begin{equation*}
    l^2(\R) := \left\{(u_n)_{n \geq 0} \in \R^{\N} \quad \text{ s.t. } \quad \sum_n u_n^2 < + \infty \right\}.
\end{equation*}
Similarly, if only a finite number of $\Gamma_{d,e}$ are non-zero, the following linear transform is a Hilbert-Schmidt operator of $l^2(\R)$:
\begin{equation*}
    \Gamma : u \in l^2(\R) \longmapsto \Gamma u := \left(\sum_{d} \Gamma_{d,e} u_{d} \right)_{e} \in l^2(\R).
\end{equation*}
With these notations, the previous polynom writes:
\begin{equation*}
    p_{\Gamma}(x,y) = \scal{\Phi(x), \Gamma \Psi(y)}.
\end{equation*}
By change of variable $c_\Cvxb \leftarrow  p_{\Gamma}$ in \cref{eq:general_matching_dualccv}, we obtain the desired formula.
Since the function $\Gamma \mapsto p_\Gamma$ is linear, we also have the convexity of $\mathbf{S}'(\Gamma) = \mathbf{S}(p_\Gamma)$, proving the theorem.
\end{proof}

\tolanddualitygeneral*

The proof is a consequence of the classical convex duality theorem:

\begin{theorem}[Rockafellar--Fenchel Theorem]
\label{thm:rf}
Let $\mathcal{F} : \meas{\X \times \Y} \to \R \cup \{+\infty\}$ and $\mathcal{G} : \meas{\X \times \Y} \to \R \cup \{+\infty\}$ be proper convex lower semi-continuous functions.
Let assume that there is a $\pi_0 \in \meas{\X \times \Y}$ such that both $\mathcal{F}(\pi_0),\mathcal{G}(\pi_0)$ are finite and $\mathcal{G}$ is continuous at $\pi_0$.
Then,
\begin{equation*}
    \inf_{\pi \in \meas{\X \times \Y}}\; \mathcal{F}(\pi) + \mathcal{G}(\pi)
        \;=\; \sup_{c \in \C(\X \times \Y)}\; - \mathcal{F}^*(-c) - \mathcal{G}^*(c).
    \end{equation*}
\end{theorem}

\begin{proof}[Proof of \cref{theorem:tolanddualitygeneral}]
We apply the Rockafellar--Fenchel theorem with:
\begin{equation*}
    \mathcal{F}(\pi) = -\int c_\Cvxb \diff\pi + \Reg(\pi) + \iota_{\measplus{\X \times \Y}}(\pi) \quad \text{ and } \quad \mathcal{G}(\pi) = \Cvxa(\pi).
\end{equation*}
By hypothesis on $\Cvxa$, $\mathcal{G}$ is continuous on its domain.
Moreover, for any $\pi_0 \in \measplus{\X \times \Y}$ such that $\Engy(\pi_0) + \Reg(\pi_0) < + \infty$, both $\mathcal{F}(\pi_0)$ and $\mathcal{G}(\pi_0)$ are finite.
Therefore, the Rockafellar--Fenchel hypotheses are indeed satisfied.
By definition of $\mathcal{D}_\Reg$, we have for all $c \in \C(\X \times \Y)$:
\begin{equation*}
    -\mathcal{F}^*(-c) = \min_{\pi \in \measplus{\X \times \Y}} \int c \diff\pi - \int c_\Cvxb \diff\pi + \Reg(\pi) = \sup_{\phi \in \Phi} \mathcal{D}_\Reg\bigl(\phi;\, c - c_\Cvxb,\, \alpha, \beta\bigr).
\end{equation*}
Therefore:
\begin{equation*}
    \min_{\pi \in \measplus{\X \times \Y}} \int c_\Cvxb\diff\pi + \Reg(\pi) + \Cvxa(\pi) = \sup_{c_\Cvxa \in \C(\X \times \Y)} \sup_{\phi \in \Phi} \mathcal{D}_\Reg\bigl(\phi;\, c_\Cvxa - c_\Cvxb,\, \alpha, \beta\bigr) - \Cvxa^*(c_\Cvxa),
\end{equation*}
providing the first identity.
To prove the second one, we simply remark that, by definition of $\mathcal{D}_\Reg$:
\begin{equation*}
    \sup_{\phi \in \Phi} ~~ 
            \mathcal{D}_\Reg\bigl(\phi;\, c_\Cvxa - c_\Cvxb,\, \alpha, \beta\bigr)  = \min_{\pi \in \measplus{\X \times \Y}} \int (c_\Cvxa - c_\Cvxb) \diff\pi + \Reg(\pi).
\end{equation*}
Finally, we prove the last part of the theorem by using the same technique as in \cref{theorem:ccv_duality_general}, approximating continuous functions on $\X \times \Y$ by polynoms.
\end{proof}

\subsection{Proofs of Section~\ref{section:quadratic_matching_case}}

The results of this section are based on the explicit convex conjugate of convex quadratic forms:
\begin{lemma}
\label{lem:fquad_conjugate}
    Let $A : \meas{X\times Y} \to \Hilb_A$ be an integral operator.
    The convex conjugate of the quadratic operator $\Cvxa(\pi) := \tfrac{1}{2}\norm{A[\pi]}^2$ is
    \begin{equation*}
        \Cvxa^*(c) \;=\;
        \begin{cases}
            \tfrac{1}{2}\norm{A[\pi]}^2 & \text{if } c = A^\top A\pi, \\
            +\infty & \text{if }  c \notin \img{A^\top A}.
        \end{cases}
    \end{equation*}
\end{lemma}
\begin{proof}
Let $c \in \C(\X \times \Y)$.
By definition,  
\begin{equation*}
    \Cvxa^*(c) = \sup_{\pi \in \meas{X \times Y}} f(\pi) \quad \text{ where } \quad f(\pi) =  \int c \diff \pi - \tfrac{1}{2} \norm{A[\pi]}^2.
\end{equation*}
The function $f$ is differentiable with $\nabla f(\pi) = c - A^\top A[\pi]$: its supremum is attained if and only if there is a $\pi$ such that $\nabla f(\pi) = 0$, i.e. $c = A^\top A[\pi]$.
In that case, a direct evaluation of $f$ shows that:
\begin{equation*}
     \Cvxa^*(c) = \int (A^\top A[\pi]) \diff \pi - \tfrac{1}{2}\norm{A[\pi]}^2 = \tfrac{1}{2} \norm{A[\pi]}^2.
\end{equation*}
Otherwise, the function $f$ do not admit any maximum, and its supremum is $+ \infty$.
\end{proof}

\outermin*

\begin{proof}
We apply \cref{eq:general_matching_dualccv} to the quadratic case, plugging the convex conjugate of $\Cvxb := \tfrac{1}{2}\norm{B[\pi]}^2$:
\begin{equation*}
\min_{\pi' \in \measplus{\X \times \Y}} ~ \min_{\pi \in \measplus{\X \times \Y}}\; -\int (B^\top B[\pi']) \diff \pi(x,y) + \mathbf{R}(\pi) + \frac{1}{2} \norm{B[\pi']}^2,
\end{equation*}
where we used the fact that the expression is finite only if there exists a $\pi' \in \meas{\X \times \Y}$ such that $c_\Cvxb = B^\top B[\pi']$.
By change of variable $\Gamma \leftarrow B[\pi']$, we then write:
\begin{equation*}
\min_{\Gamma \in \img{B}} ~ \min_{\pi \in \measplus{\X \times \Y}}\; -\int (B^\top \Gamma) \diff \pi(x,y) + \mathbf{R}(\pi) + \frac{1}{2} \norm{\Gamma}^2.
\end{equation*}
Since $(B^\top \Gamma)(x,y) = \scal{\Gamma, B(x,y)}$ for all $x \in \X$ and $y \in \Y$, we obtain the desired formula.
The second identity of the theorem is a consequence of \cref{theorem:tolanddualitygeneral}, using the same changes of variable as above ($\Xi \leftarrow A[\pi']$).
\end{proof}

\qotasregistration*

We state a more precise version of \cref{thm:qot_as_registration} in \cref{thm:qot_as_registration_blockseparable}:

\qotasregistrationblockseparable*

\begin{proof}
In the separable case, the first identity is a direct consequence of \cref{eq:outer_gamma} and the identity:
\begin{equation*}
    \scal{\Gamma, B(x,y)} = \scal{\Gamma, \Phi_B(x)\Psi_B(y)^\top} = \scal{\Phi_B(x), \Gamma \Psi_B(y)}.
\end{equation*}
The block-separable case is treated analogously. Let us write $\norm{B} = \sum_{i=1}^{\LL} \norm{B_i}$ with $B_i(x,y) = \Phi_{B,i}(x) \Psi_{B,i}(y)^\top$ for all $1 \leq i \leq \LL$ and all $x \in \X$, $y \in \Y$. 
Then, \cref{eq:outer_gamma} is equivalent to:
\begin{multline*}
        \min_{\Gamma_1 \in \img{B_1}} \dots \min_{\Gamma_\LL \in \img{B_\LL}} \; \min_{\pi \in \measplus{\X \times \Y}} - \int \sum_{i=1}^{\LL} \scal{\Phi_{B,i}(x), \Gamma_i \Psi_{B,i}(y)} \diff \pi(x,y) \\ + \mathbf{R}(\pi) +  \tfrac{1}{2} \sum_{i=1}^{\LL} \norm{\Gamma_i}^2.
\end{multline*}
Using the definitions of \cref{appendix:block_separable}, we directly obtain the first identity \cref{thm:qot_as_registration_blockseparable}.
The other identities follow from \cref{prop:outer_gamma}, using the above identity for $A$.
\end{proof}

\subsection{Proofs of the Appendix}

\densitydcfunctions*

\begin{proof}
First, since $\Engy$ goes to infinity at infinity, all its minima are contained in a closed bounded ball $V$ centered on the origin, which is a convex set of $\measplus{\X \times \Y}$.
It is compact for the weak-* convergence, since it is bounded and closed.
Let ${\rm DC}(V) \subset \C(V)$ be the set functions on $V$ that are decomposable as a difference of continuous convex functions.
It satisfies the hypotheses of the Stone-Weierstrass theorem:
\begin{itemize}
    \item All constant functions belong to ${\rm DC}(V)$.
    \item Linear functions belong to ${\rm DC}(V)$: therefore, for all $\pi, \pi' \in V$, there exists $\mathcal{F} \in {\rm DC}(V)$ such that $\mathcal{F}(\pi) \neq \mathcal{F}(\pi')$.
    \item The set ${\rm DC}(V)$ is trivially stable by addition and multiplication by a scalar.
    \item If $\mathcal{F}_1$ and $\mathcal{F}_2$ are in ${\rm DC}(V)$, then the product  $\mathcal{F}_1 \mathcal{F}_2$ belongs to ${\rm DC}(V)$ as well.
\end{itemize}
To prove the last property, we first write:
\begin{equation*}
\mathcal{F}_1 \mathcal{F}_2 = \tfrac{1}{2} (\mathcal{F}_1 + \mathcal{F}_2)^2 - \tfrac{1}{2} (\mathcal{F}_1 - \mathcal{F}_2)^2.
\end{equation*}
We only need to prove that squared DC functions are also DC.
To do so, we remark that if $\mathcal{F} \in {\rm DC}(V)$ we can write $\mathcal{F} = \mathcal{A} - \mathcal{B}$ where $\mathcal{A}$ and $ \mathcal{B}$ are both \emph{positive} continuous convex functions.
Indeed, since they are continuous, these functions are bounded on the compact set $V$: if we add the same sufficiently large constant to $\mathcal{A}$ and $\mathcal{B}$, we ensure their positivity while keeping their difference identical. 
Then,
\begin{equation*}
    \mathcal{F}^2 = 2(\mathcal{A}^2 + \mathcal{B}^2) - (\mathcal{A} + \mathcal{B})^2 \in {\rm DC}(V),
\end{equation*}
since all the terms of the sum are convex as the square of positive convex functions.

Since $V$ is compact for the weak-* convergence, the Stone-Weierstrass theorem states that ${\rm DC}(V)$ is dense in $\C(V)$.
Therefore, for any weak-* continuous function $\Engy$ on $\measplus{\X \times \Y}$ and any $\delta > 0$, there are two convex functions $\Cvxa_\delta$ and $\Cvxb_\delta$ satisfying the desired approximation on $V$. 
By replacing $\Cvxa_\delta$ with $\Cvxa_\delta + \iota_{V}$, we obtain the desired result. 
\end{proof}

\lemmacontinuityconjugate*

\begin{proof}
Convex conjugates are lower-semi continuous as a limit of continuous functions $c \mapsto \int c\diff\pi - \Cvxa(\pi)$.
Lower-semi continuous functions on Banach spaces are continuous on the interior of their domain, as proven e.g. in \cite[Section 3.3, proposition 4]{aubin1979compactness}.
Since $(\C(\X \times \Y), \norm{\cdot}_{\infty})$ is a Banach space, we only need to prove that the domain of $\Cvxa^*$ is the whole space $\C(\X \times \Y)$ to conclude.

Let $c \in \C(\X \times \Y)$.
By duality between $(\C(\X \times \Y), \norm{\cdot}_{\infty})$ and $(\meas{ \X \times \Y}, \TV)$, the following inequality holds:
\begin{equation*}
  \text{For all } \pi \in \meas{\X \times \Y}, \quad  \int c \diff \pi \leq \norm{c}_{\infty} \TV(\pi).
\end{equation*}
Therefore,
\begin{equation*}
\Cvxa^*(c) = \sup_{\pi \in \meas{\X \times \Y}} \int c \diff\pi - \Cvxa(\pi) \leq \sup_{\pi \in \meas{\X \times \Y}} \TV(\pi) \left(\norm{c}_{\infty} - \frac{\Cvxa(\pi)}{\TV(\pi)} \right).
\end{equation*}
By superlinearity, the set of $\pi$ such that $\Cvxa(\pi) /\TV(\pi) \leq \norm{c}_{\infty}$ is bounded: therefore, the supremum can be restricted to a subset of $\pi$ with bounded total variation, and the previous bound ensures that $\Cvxa^*(c) < + \infty$.
Therefore, the domain of $\Cvxa^*$ is $\C(\X \times \Y)$: putting everything together, it demonstrates the continuity of $\Cvxb^*$ on $\C(\X \times \Y)$ for the uniform norm.
\end{proof}

\existencediagqm*

\begin{proof}
Let  $\mu := \alpha \otimes \beta$.
Any $\pi$ such that $\KL(\pi \Vert \mu) < + \infty$ admits a density with respect to $\mu$, i.e. $\pi(x,y) = f_{\pi}(x,y)\diff\mu(x,y)$ with $f_{\pi} \in L^1(\X\times\Y, \mu)$.
The space of squared-integrable functions $L^2(\X\times\Y, \mu)$ is dense in $L^1(\X\times\Y, \mu)$, so the original problem is equivalent to:
\begin{equation*}
    \Reg\QM(\alpha, \beta) = \inf_{f \in L^2(\X\times\Y, \mu)} ~~ \frac{1}{2}\int f(v) \left( \int k(v,w)f(w)\diff\mu(w) \right)\diff\mu(v) + \Reg(f \diff\mu).
\end{equation*}
Since $k$ is continuous on the compact space $\X \times \Y$, it is bounded and we have the following:
\begin{equation*}
    \int k(v,w)^2\diff\mu(v) \mu(w) < +\infty.
\end{equation*}
Therefore, the operator $f \in L^2(\X\times\Y, \mu) \mapsto \mathcal{K}f:= \int k(\cdot,w)f(w)\diff\mu(w)) \in L^2(\X\times\Y, \mu)$ is an auto-adjoint Hilbert-Schmidt integral operator.
In particular, it is compact and the spectral theorem applies: there exists an orthonormal basis $(e_i) \in L^2(\X\times\Y, \mu)$ and eigenvalues $(\lambda_i) \in \R$ such that,
\begin{equation*}
    \text{For all } f \in L^2(\X\times\Y, \mu), v \in \X\times \Y, \quad \mathcal{K}f(v) = \sum_i \lambda_i e_i(v) \int f(w) e_i(w)\diff\mu(w).
    \end{equation*}
We define $A, B: L^2(\X\times\Y, \mu) \to L^2(\X\times\Y, \mu)$ as:
\begin{align*}
    A: f \longmapsto & \int \sum_{\lambda_i > 0} \abs{\lambda_i} e_i e_i(w) f(w) \diff\mu(w), \\
    B: f \longmapsto & \int \sum_{\lambda_i < 0} \abs{\lambda_i} e_i \int e_i(w) f(w) \diff\mu(w).
\end{align*}
Then, we have $\mathcal{K}f = Af - Bf$ and:
\begin{equation*}
\int k(v,w)f(v)f(w)\diff\mu(v)\diff\mu(w) = \norm{Af}_{L^2}^2 - \norm{B[f]}_{L^2}^2.
\end{equation*}
Therefore, after identifying $A$ and $B$ with their action on measures $f\diff\mu \in \mathcal{M}_+^{L^2}(\mu)$:
\begin{equation*}
    \Reg\QM(\alpha, \beta) = \inf_{\pi \in \mathcal{M}_+^{L^2}(\mu)} ~~ \tfrac{1}{2}\norm{A[\pi]} - \tfrac{1}{2} \norm{B[\pi]}^2 + \Reg(\pi).
\end{equation*}
\end{proof}

\kernelfactorization*

\begin{proof}
Given $A$ and $B$, we define the function:
\[
    k_{A,B}(x,y,x',y') := \scal{A(x,y), A(x',y')} - \scal{B(x,y), B(x',y')}.
\]
Let $k \in \C((\X \times \Y)^2)$ be a symmetric function.
If $k = k_{A,B}$ pointwise, then by bilinearity of the inner product,
\[
    \int k\,\diff\pi\otimes\pi
    = \int k_{A,B}\,\diff\pi\otimes\pi
    = \norm{\int A\,\diff\pi}^2 - \norm{\int B\,\diff\pi}^2.
\]
Conversely, assume the integral identity holds for every $\pi \in \meas{\X \times \Y}$. Applying it to $\pi = \delta_{(x,y)}$ gives $k(x,y,x,y) = k_{A,B}(x,y,x,y)$. 
Finally, applying it to $\pi = \delta_{(x,y)} + \delta_{(x',y')}$ gives:
\begin{multline*}
    k(x,y,x,y) + k(x',y',x',y') + 2 k(x,y,x',y') \\= k_{A,B}(x,y,x,y) + k_{A,B}(x',y',x',y') + 2k_{A,B}(x,y,x',y'), 
\end{multline*}
which implies that $k(x,y,x',y') = k_{A,B}(x,y,x',y')$.
\end{proof}

\section{Congergence of the dual algorithm}
\label{appendix:end-to-end-cv}
\subsection{Notation and setting}
\label{subsec:cv_setting}

\paragraph{Objects.}
Let $\mathcal P:=\measplus{\X\times\Y}$. Marginal constraints, when present,
are encoded in $\Reg$, so that $\operatorname{dom}\Reg$ is the set of
admissible plans. Let $A:\meas{\X\times\Y}\to\Hilb_A$ and $B:\meas{\X\times\Y}\to\Hilb_B$ be the integral operators of \cref{eq:loss_factorized}; they are linear, so that $B[\pi-\pi']=B[\pi]-B[\pi']$ for $\pi,\pi'\in\mathcal P$. Their adjoints are
$A^\top\Xi=\scal{\Xi,A(\cdot,\cdot)}$ and $B^\top\Gamma=\scal{\Gamma,B(\cdot,\cdot)}$;
the variables $\Gamma$ and $\Xi$ range over $\Hilb_B$ and $\Hilb_A$. We use
the objective of \cref{eq:general_matching}, its linearization at fixed
$\Gamma$ (the inner problem of \cref{eq:outer_gamma}) and the value of that
inner problem:
\begin{align}
    \Engy_\Reg(\pi)&:=\tfrac12\norm{A[\pi]}^2-\tfrac12\norm{B[\pi]}^2+\Reg(\pi),
    &\Reg\QM(\alpha,\beta)&=\inf_{\pi\in\mathcal P}\Engy_\Reg(\pi),
    \label{eq:unified_E}\\
    Q_\Gamma(\pi)&:=\tfrac12\norm{A[\pi]}^2-\scal{\Gamma,B[\pi]}+\Reg(\pi),
    &\Lambda(\Gamma)&:=\tfrac12\norm{\Gamma}^2+\inf_{\pi\in\mathcal P}Q_\Gamma(\pi).
    \label{eq:unified_Q}
\end{align}
By \cref{prop:outer_gamma}, $\Reg\QM(\alpha,\beta)=\inf_\Gamma\Lambda(\Gamma)$;
$\Lambda$ is the Lyapunov function of the outer loop. For a cost
$c\in\C(\X\times\Y)$ we write
\begin{equation}
    \label{eq:unified_V}
    \Reg\OT(c):=\inf_{\pi\in\mathcal P}\bigl\{\scal{c,\pi}+\Reg(\pi)\bigr\}
\end{equation}
for the value of the $\Reg\OT$ problem \cref{eq:dualizable_reg}, seen as a
function of its cost for fixed $(\alpha,\beta)$. With the effective cost
$c_{\Gamma,\Xi}:=A^\top\Xi-B^\top\Gamma$ of \cref{prop:outer_gamma}, the reduced
objective of \cref{section:numerical_optim} is
\begin{equation}
    \label{eq:unified_h}
    h_\Gamma(\Xi)
    :=
    -\tfrac12\norm{\Xi}^2+\Reg\OT(c_{\Gamma,\Xi})
    =
    \sup_{\phi\in\Phi}\mathcal D_\Reg(\phi;\,c_{\Gamma,\Xi},\alpha,\beta)-\tfrac12\norm{\Xi}^2,
\end{equation}
the second equality being the dual representation \cref{eq:dualizable_reg}.
\paragraph{An elementary identity.}
For all $a,b,\Xi$ in a Hilbert space,
\begin{equation}
    \tfrac12\norm a^2-\tfrac12\norm b^2-\scal{\Xi,a-b}
    =\tfrac12\norm{a-\Xi}^2-\tfrac12\norm{b-\Xi}^2 .
    \label{eq:square_id_2}
\end{equation}
It follows by expanding the squares and is used repeatedly below.

\emph{Weak duality.} Since $\scal{\Xi,a}-\tfrac12\norm\Xi^2\le\tfrac12\norm a^2$ for all $a,\Xi\in\Hilb_A$, for every $\Gamma$ and $\Xi$
\begin{align*}
    h_\Gamma(\Xi)
    =\inf_{\pi\in\mathcal P}\bigl\{\scal{\Xi,A[\pi]}-\tfrac12\norm\Xi^2-\scal{\Gamma,B[\pi]}+\Reg(\pi)\bigr\}
    \le\inf_{\pi\in\mathcal P}Q_\Gamma(\pi),
\end{align*}
hence
\begin{equation}
    \label{eq:unified_weak_duality}
    \sup_{\Xi\in\Hilb_A}h_\Gamma(\Xi)\le\inf_{\pi\in\mathcal P}Q_\Gamma(\pi).
\end{equation}
Equality is proved in \cref{prop:unified_danskin} when $h_\Gamma$ is differentiable, and is not needed elsewhere.

\paragraph{Standing assumptions.}
Throughout the appendix we assume:
\begin{enumerate}
    \item[\textbf{(S1)}] $\Reg$ is proper, convex and lower semicontinuous,
    and $\Reg(\pi_0)<\infty$ for some $\pi_0\in\mathcal P$;
    \item[\textbf{(S2)}] $\Reg\QM(\alpha,\beta)>-\infty$.
\end{enumerate}
Under \textbf{(S1)}--\textbf{(S2)}, $\Lambda$ is finite everywhere
(\cref{lem:unified_structure}), and we set
$\Delta^{\mathrm{out}}:=\Lambda(\Gamma_0)-\Reg\QM(\alpha,\beta)\in[0,\infty)$
for the initial point $\Gamma_0$ of the algorithm.

\paragraph{$\Reg\OT$ oracle and certificate.}
A plan $\widehat\pi\in\mathcal P$ is \emph{$\delta_{\Reg\OT}$-accurate} at cost
$c$ if
\begin{equation}
    \label{eq:unified_ot_oracle}
    \scal{c,\widehat\pi}+\Reg(\widehat\pi)\le\Reg\OT(c)+\delta_{\Reg\OT}.
\end{equation}
This is a primal notion: with a dual solver, $\widehat\pi$ is the plan
recovered from the current potentials and $\delta_{\Reg\OT}$ is its
primal--dual gap, which upper-bounds the left-hand side minus $\Reg\OT(c)$ by
weak duality. It requires $\Reg(\widehat\pi)<\infty$; for KL or TV marginal
penalties the recovered plan qualifies, while for a hard marginal constraint
it must first be rounded onto the constraint set, e.g.\ by
\citet[Alg.~2]{altschuler2017near}, the moments $A[\widehat\pi]$,
$B[\widehat\pi]$ being then those of the rounded plan. For a pair
$(\Xi,\widehat\pi)$ with $\widehat\pi$ $\delta_{\Reg\OT}$-accurate at
$c_{\Gamma,\Xi}$ we define the computable certificate
\begin{equation}
    \label{eq:unified_cert_def}
    \operatorname{Cert}_\Gamma(\Xi,\widehat\pi):=\tfrac12\norm{A[\widehat\pi]-\Xi}^2+\delta_{\Reg\OT},
\end{equation}
which is the quantity tested by the inner stopping rule of \cref{alg:ugw};
\cref{lem:unified_certificate} shows that it bounds
$Q_\Gamma(\widehat\pi)-\inf_\pi Q_\Gamma$.

\paragraph{Indices for \cref{alg:ugw}.}
We analyze \cref{alg:ugw} with: constant accuracy
$\delta_{\Reg\OT}$, constant step $\tau$, certificate-based inner stopping
rule. At outer iteration $t$, write $\Xi_{t,k}$ for the value of $\Xi$ before
the $(k+1)$-th $\Reg\OT$ solve (so $\Xi_{t,0}=\Xi_{t-1,K_{t-1}}$ is carried
over from the previous outer iteration), $\widehat\pi_{t,k}$ for the
$\delta_{\Reg\OT}$-accurate plan returned at cost $c_{\Gamma_t,\Xi_{t,k}}$,
and
\begin{equation}
    \label{eq:unified_Kt}
    K_t:=\min\bigl\{k\ge0:\ \operatorname{Cert}_{\Gamma_t}(\Xi_{t,k},\widehat\pi_{t,k})\le\delta^2/4\bigr\}\in\N\cup\{+\infty\}.
\end{equation}
The inner loop of iteration $t$ thus makes $K_t+1$ $\Reg\OT$ solves and $K_t$
updates $\Xi_{t,k+1}=\Xi_{t,k}+\tau(A[\widehat\pi_{t,k}]-\Xi_{t,k})$; it
terminates if and only if $K_t<\infty$, in which case
$\widehat\pi_t:=\widehat\pi_{t,K_t}$, $\Gamma_{t+1}:=B[\widehat\pi_t]$, and
the algorithm returns $\widehat\pi_t$ when $\norm{\Gamma_{t+1}-\Gamma_t}\le\delta$.
When $\Gamma=\Gamma_t$ is fixed we drop the index $t$ and write $\Xi_k$,
$\widehat\pi_k$. For the entropic regularizers of
\cref{table:regularizers_dual}, $A[\widehat\pi]-\Xi$ is the gradient in $\Xi$
of $\mathcal D_\Reg(\phi;c_{\Gamma,\Xi},\alpha,\beta)-\tfrac12\norm\Xi^2$ at
the potentials $\phi$ encoding $\widehat\pi$, which is why the $\Xi$-update is
an (inexact) gradient step.

\subsection{The fixed-\texorpdfstring{$\Gamma$}{Gamma} certificate: proof of Lemma~\ref{lem:unified_certificate}}
\label{subsec:unified_fixed_gamma}

The following lemma justifies the inner stopping rule of \cref{alg:ugw}.

\begin{lemma}[Inexact supergradient and primal certificate]
    \label{lem:unified_certificate}
    Fix $(\Gamma,\Xi)\in\Hilb_B\times\Hilb_A$, let $\widehat\pi$ be
    $\delta_{\Reg\OT}$-accurate at $c_{\Gamma,\Xi}$, and set
    $r:=A[\widehat\pi]-\Xi$. Then, for every $\Xi'\in\Hilb_A$,
    \begin{equation}
        \label{eq:unified_approx_first_order}
        h_\Gamma(\Xi')\le h_\Gamma(\Xi)+\scal{A[\widehat\pi]-\Xi,\Xi'-\Xi}-\tfrac12\norm{\Xi'-\Xi}^2+\delta_{\Reg\OT},
    \end{equation}
    and
    \begin{equation}
        \label{eq:unified_primal_certificate}
        Q_\Gamma(\widehat\pi)-\inf_{\pi\in\mathcal P}Q_\Gamma(\pi)
        \le\tfrac12\norm{A[\widehat\pi]-\Xi}^2+\delta_{\Reg\OT}
        =\operatorname{Cert}_\Gamma(\Xi,\widehat\pi).
    \end{equation}
\end{lemma}

\begin{proof}
    Let $\Xi'\in\Hilb_A$. Since $c_{\Gamma,\Xi'}=c_{\Gamma,\Xi}+A^\top(\Xi'-\Xi)$, testing the infimum defining $\Reg\OT(c_{\Gamma,\Xi'})$ at $\widehat\pi$ and using \cref{eq:unified_ot_oracle},
    \begin{align*}
        \Reg\OT(c_{\Gamma,\Xi'})
        &\le\scal{c_{\Gamma,\Xi},\widehat\pi}+\Reg(\widehat\pi)+\scal{A[\widehat\pi],\Xi'-\Xi}
        \le\Reg\OT(c_{\Gamma,\Xi})+\scal{A[\widehat\pi],\Xi'-\Xi}+\delta_{\Reg\OT}.
    \end{align*}
    Subtracting $\tfrac12\norm{\Xi'}^2=\tfrac12\norm\Xi^2+\scal{\Xi,\Xi'-\Xi}+\tfrac12\norm{\Xi'-\Xi}^2$ from both sides gives \cref{eq:unified_approx_first_order}.

    For \cref{eq:unified_primal_certificate}, let $\pi\in\mathcal P$. The accuracy of $\widehat\pi$ tested against $\pi$ reads
    \[
        -\scal{\Gamma,B[\widehat\pi]-B[\pi]}+\Reg(\widehat\pi)-\Reg(\pi)
        \le-\scal{\Xi,A[\widehat\pi]-A[\pi]}+\delta_{\Reg\OT},
    \]
    so that, by \cref{eq:unified_Q} and then \cref{eq:square_id_2} with $(a,b)=(A[\widehat\pi],A[\pi])$,
    \begin{align*}
        Q_\Gamma(\widehat\pi)-Q_\Gamma(\pi)
        &\le\tfrac12\norm{A[\widehat\pi]}^2-\tfrac12\norm{A[\pi]}^2-\scal{\Xi,A[\widehat\pi]-A[\pi]}+\delta_{\Reg\OT}\\
        &=\tfrac12\norm{A[\widehat\pi]-\Xi}^2-\tfrac12\norm{A[\pi]-\Xi}^2+\delta_{\Reg\OT}
        \le\tfrac12\norm{A[\widehat\pi]-\Xi}^2+\delta_{\Reg\OT}.
    \end{align*}
    Taking the infimum over $\pi$ proves \cref{eq:unified_primal_certificate}.
\end{proof}

\subsection{Outer convergence: proof of Theorem~\ref{thm:main_convergence}(i)}
\label{subsec:unified_main_generic}

\begin{lemma}[Square completion and outer Lyapunov function]
    \label{lem:unified_structure}
    For every $(\Gamma,\pi)\in \Hilb_B\times\mathcal P$,
    \begin{equation}
        \label{eq:unified_square}
        \tfrac12\norm{\Gamma}^2+Q_\Gamma(\pi)
        =
        \Engy_\Reg(\pi)
        +\tfrac12\norm{\Gamma-B[\pi]}^2.
    \end{equation}
    Consequently, the function $\Lambda$ of \cref{eq:unified_Q} satisfies
    \[
        \Lambda(\Gamma)
        =
        \inf_{\pi\in\mathcal P}
        \Bigl\{\Engy_\Reg(\pi)+\tfrac12\norm{\Gamma-B[\pi]}^2\Bigr\},
    \]
    is finite under \textbf{(S1)}--\textbf{(S2)}, and
    \begin{equation}
        \label{eq:unified_Lambda}
        \Lambda(\Gamma)\ge\Reg\QM(\alpha,\beta)
        \quad\text{for every }\Gamma\in\Hilb_B,
        \qquad
        \inf_{\Gamma}\Lambda(\Gamma)=\Reg\QM(\alpha,\beta).
    \end{equation}
\end{lemma}

\begin{proof}
    Expanding $\tfrac12\norm{\Gamma-B[\pi]}^2=\tfrac12\norm\Gamma^2-\scal{\Gamma,B[\pi]}+\tfrac12\norm{B[\pi]}^2$ and adding $\Engy_\Reg(\pi)$, the terms $\pm\tfrac12\norm{B[\pi]}^2$ cancel and give \cref{eq:unified_square}; taking the infimum over $\pi$ gives the expression of $\Lambda$.

    Since the square is nonnegative, $\Lambda(\Gamma)\ge\inf_\pi\Engy_\Reg(\pi)=\Reg\QM(\alpha,\beta)>-\infty$ by \textbf{(S2)}, and $\Lambda(\Gamma)\le\Engy_\Reg(\pi_0)+\tfrac12\norm{\Gamma-B[\pi_0]}^2<\infty$ for the plan $\pi_0$ of \textbf{(S1)}, whose moments are finite since $A$ and $B$ are bounded. Finally, choosing $\Gamma=B[\pi]$ kills the square, so $\Lambda(B[\pi])\le\Engy_\Reg(\pi)$ for every $\pi$; together with the lower bound, this gives \cref{eq:unified_Lambda}.
\end{proof}

The identity \cref{eq:unified_square} shows that
$\tfrac12\norm\Gamma^2+Q_\Gamma$ majorizes $\Engy_\Reg$ with equality at
$\Gamma=B[\pi]$, so that the $\Gamma$-update of \cref{alg:ugw} is a
majorization--minimization (equivalently, DC) step. Its fixed points are the
plans $\pi^\star$ with $\pi^\star\in\arg\min_\pi Q_{B[\pi^\star]}(\pi)$, the
DC-critical plans of \cref{section:numerical_optim}; the algorithm produces
them approximately, in the following sense.

\begin{definition}[$(\delta,\eta)$-criticality]
\label{def:unified_criticality}
We call $\widehat\pi\in\mathcal P$ \emph{$(\delta,\eta)$-critical}, with
$\delta,\eta\ge0$, if for every $\pi\in\mathcal P$,
\begin{equation}
    \label{eq:unified_criticality}
    \Engy_\Reg(\widehat\pi)
    \le
    \Engy_\Reg(\pi)
    +\tfrac12\norm{B[\pi-\widehat\pi]}^2
    +\delta\norm{B[\pi-\widehat\pi]}
    +\eta.
\end{equation}
\end{definition}
By \cref{eq:unified_square}, $\widehat\pi$ is $(\delta,\eta)$-critical if and
only if it is an $\eta$-minimizer of the convex function
$\pi\mapsto Q_{B[\widehat\pi]}(\pi)+\delta\norm{B[\pi-\widehat\pi]}$; the case
$(\delta,\eta)=(0,0)$ is DC criticality.

\begin{lemma}[One outer step]
    \label{lem:outer_step}
    Let $\Gamma\in\Hilb_B$, $\eta\ge0$, and let $\widehat\pi\in\mathcal P$ be
    an $\eta$-minimizer of $Q_\Gamma$, that is,
    \begin{equation}
        \label{eq:unified_outer_iteration}
        Q_\Gamma(\widehat\pi)\le\inf_{\pi\in\mathcal P}Q_\Gamma(\pi)+\eta.
    \end{equation}
    Set $\Gamma^+:=B[\widehat\pi]$. Then
    \begin{enumerate}
        \item[\rm(i)] $\Lambda(\Gamma)-\Lambda(\Gamma^+)\ge\tfrac12\norm{\Gamma^+-\Gamma}^2-\eta$;
        \item[\rm(ii)] $\widehat\pi$ is $\bigl(\norm{\Gamma^+-\Gamma},\eta\bigr)$-critical.
    \end{enumerate}
\end{lemma}

\begin{proof}
    We first prove (i). By \cref{eq:unified_outer_iteration},
    \begin{equation}
        \label{eq:outer_lambda_lower}
        \Lambda(\Gamma)
        =\tfrac12\norm{\Gamma}^2+\inf_{\pi}Q_\Gamma(\pi)
        \ge\tfrac12\norm{\Gamma}^2+Q_\Gamma(\widehat\pi)-\eta.
    \end{equation}
    Applying the square identity \cref{eq:unified_square} at
    $(\Gamma,\widehat\pi)$, and using $B[\widehat\pi]=\Gamma^+$,
    \[
        \tfrac12\norm{\Gamma}^2+Q_\Gamma(\widehat\pi)
        =\Engy_\Reg(\widehat\pi)+\tfrac12\norm{\Gamma-\Gamma^+}^2,
    \]
    so that \cref{eq:outer_lambda_lower} becomes
    \begin{equation}
        \label{eq:outer_lambda_lower_final}
        \Lambda(\Gamma)
        \ge\Engy_\Reg(\widehat\pi)+\tfrac12\norm{\Gamma-\Gamma^+}^2-\eta.
    \end{equation}
    On the other hand, evaluating the infimum defining $\Lambda(\Gamma^+)$ at
    $\widehat\pi$ and applying \cref{eq:unified_square} at
    $(\Gamma^+,\widehat\pi)$, where the square vanishes because
    $\Gamma^+=B[\widehat\pi]$,
    \[
        \Lambda(\Gamma^+)
        \le\tfrac12\norm{\Gamma^+}^2+Q_{\Gamma^+}(\widehat\pi)
        =\Engy_\Reg(\widehat\pi).
    \]
    Subtracting this inequality from \cref{eq:outer_lambda_lower_final}
    proves (i).

    We now prove (ii). Fix an arbitrary $\pi\in\mathcal P$. By
    \cref{eq:unified_outer_iteration}, $Q_\Gamma(\widehat\pi)\le Q_\Gamma(\pi)+\eta$.
    Applying \cref{eq:unified_square} at $(\Gamma,\widehat\pi)$ and at
    $(\Gamma,\pi)$, this inequality reads
    \begin{equation}
        \label{eq:outer_step_square_ineq}
        \Engy_\Reg(\widehat\pi)+\tfrac12\norm{\Gamma-\Gamma^+}^2
        \le\Engy_\Reg(\pi)+\tfrac12\norm{\Gamma-B[\pi]}^2+\eta.
    \end{equation}
    Set $d:=B[\pi-\widehat\pi]=B[\pi]-\Gamma^+$. Then
    $\Gamma-B[\pi]=(\Gamma-\Gamma^+)-d$, and expanding the square,
    \[
        \tfrac12\norm{\Gamma-B[\pi]}^2
        =\tfrac12\norm{\Gamma-\Gamma^+}^2-\scal{\Gamma-\Gamma^+,d}+\tfrac12\norm d^2.
    \]
    Substituting this into \cref{eq:outer_step_square_ineq}, the term
    $\tfrac12\norm{\Gamma-\Gamma^+}^2$ cancels on both sides, and we obtain
    \begin{align*}
        \Engy_\Reg(\widehat\pi)
        &\le\Engy_\Reg(\pi)+\tfrac12\norm d^2+\scal{\Gamma^+-\Gamma,d}+\eta\\
        &\le\Engy_\Reg(\pi)+\tfrac12\norm{B[\pi-\widehat\pi]}^2
        +\norm{\Gamma^+-\Gamma}\,\norm{B[\pi-\widehat\pi]}+\eta,
    \end{align*}
    where the last inequality is Cauchy--Schwarz. This is
    \cref{eq:unified_criticality} with $(\delta,\eta)$ replaced by
    $(\norm{\Gamma^+-\Gamma},\eta)$, that is, $\widehat\pi$ is
    $(\norm{\Gamma^+-\Gamma},\eta)$-critical.
\end{proof}

\begin{theorem}[Outer convergence of \cref{alg:ugw}; precise form of \cref{thm:main_convergence}(i)]
    \label{thm:unified_outer}
    Assume \textbf{(S1)}--\textbf{(S2)} and let $\delta>0$. Run
    \cref{alg:ugw} with any $\Reg\OT$ solver, any step size $\tau$ and any
    accuracy $\delta_{\Reg\OT}$ such that every inner loop terminates
    ($K_t<\infty$ for all $t$). Let
    \begin{equation}
        \label{eq:unified_outer_count}
        T_\delta:=\max\Bigl\{1,\Bigl\lceil\frac{4\Delta^{\mathrm{out}}}{\delta^2}\Bigr\rceil\Bigr\}.
    \end{equation}
    Then:
    \begin{enumerate}
        \item[\rm(i)] for every outer iteration $t$,
        $Q_{\Gamma_t}(\widehat\pi_t)\le\inf_\pi Q_{\Gamma_t}(\pi)+\delta^2/4$,
        \begin{equation}
            \label{eq:unified_outer_descent}
            \Lambda(\Gamma_t)-\Lambda(\Gamma_{t+1})
            \ge\tfrac12\norm{\Gamma_{t+1}-\Gamma_t}^2-\tfrac{\delta^2}4,
        \end{equation}
        and $\widehat\pi_t$ is $\bigl(\norm{\Gamma_{t+1}-\Gamma_t},\delta^2/4\bigr)$-critical;
        \item[\rm(ii)] the stopping test $\norm{\Gamma_{t+1}-\Gamma_t}\le\delta$
        is met at some outer iteration $t_{\mathrm{stop}}<T_\delta$, and the
        returned plan $\widehat\pi_{t_{\mathrm{stop}}}$ is
        $(\delta,\delta^2/4)$-critical.
    \end{enumerate}
\end{theorem}

\begin{proof}
    Fix an outer iteration $t$. By definition \cref{eq:unified_Kt} of $K_t$,
    the plan $\widehat\pi_t=\widehat\pi_{t,K_t}$ is
    $\delta_{\Reg\OT}$-accurate at $c_{\Gamma_t,\Xi_{t,K_t}}$ and
    $\operatorname{Cert}_{\Gamma_t}(\Xi_{t,K_t},\widehat\pi_t)\le\delta^2/4$.
    Hence \cref{lem:unified_certificate} gives
    $Q_{\Gamma_t}(\widehat\pi_t)\le\inf_\pi Q_{\Gamma_t}(\pi)+\delta^2/4$,
    which is the first claim of (i), and \cref{lem:outer_step} applied with
    $(\Gamma,\widehat\pi,\eta)=(\Gamma_t,\widehat\pi_t,\delta^2/4)$, for
    which $\Gamma^+=B[\widehat\pi_t]=\Gamma_{t+1}$, gives the other two.

    For (ii), suppose that the stopping test fails for every $t<T_\delta$.
    Summing \cref{eq:unified_outer_descent} over $t=0,\dots,T_\delta-1$ and
    using $\Lambda(\Gamma_{T_\delta})\ge\Reg\QM(\alpha,\beta)$ from
    \cref{eq:unified_Lambda},
    \[
        \tfrac12\sum_{t=0}^{T_\delta-1}\norm{\Gamma_{t+1}-\Gamma_t}^2
        \le\Lambda(\Gamma_0)-\Lambda(\Gamma_{T_\delta})+T_\delta\frac{\delta^2}4
        \le\Delta^{\mathrm{out}}+T_\delta\frac{\delta^2}4 .
    \]
    Therefore
    \[
        \min_{0\le t<T_\delta}\norm{\Gamma_{t+1}-\Gamma_t}^2
        \le\frac1{T_\delta}\sum_{t=0}^{T_\delta-1}\norm{\Gamma_{t+1}-\Gamma_t}^2
        \le\frac{2\Delta^{\mathrm{out}}}{T_\delta}+\frac{\delta^2}2
        \le\delta^2,
    \]
    because $T_\delta\ge4\Delta^{\mathrm{out}}/\delta^2$. This contradicts the
    assumption that $\norm{\Gamma_{t+1}-\Gamma_t}>\delta$ for all
    $t<T_\delta$, so the stopping test is met at some
    $t_{\mathrm{stop}}<T_\delta$. At that iteration,
    $\norm{\Gamma_{t_{\mathrm{stop}}+1}-\Gamma_{t_{\mathrm{stop}}}}\le\delta$
    and (i) show that $\widehat\pi_{t_{\mathrm{stop}}}$ is
    $(\delta',\delta^2/4)$-critical for some $\delta'\le\delta$, hence
    $(\delta,\delta^2/4)$-critical, because \cref{eq:unified_criticality} is
    monotone in both parameters.
\end{proof}

\Cref{thm:unified_outer} is conditional only on the termination of the
inner loops, which is observed at run time; \cref{cor:unified_oracle_complexity}
gives a smoothness condition under which termination is guaranteed, with an
explicit bound on $K_t$.

\subsection{Inner-loop complexity: proof of Theorem~\ref{thm:main_convergence}(ii)}
\label{subsec:unified_nested}

At fixed $\Gamma$, the inner loop of \cref{alg:ugw} is an inexact gradient
ascent on $h_\Gamma$, with inexact gradient
$\widehat g_k:=A[\widehat\pi_k]-\Xi_k$ (\cref{lem:unified_certificate}).

\begin{lemma}[Strong concavity of the reduced objective]
    \label{lem:unified_strong_concavity}
    Assume \textbf{(S1)} and fix $\Gamma$. Then $h_\Gamma<+\infty$ on
    $\Hilb_A$. If moreover $\Reg\OT(c_{\Gamma,\Xi})>-\infty$ for every
    $\Xi\in\Hilb_A$, so that $h_\Gamma$ is finite, then $h_\Gamma$ is
    $1$-strongly concave and upper semicontinuous on $\Hilb_A$, and attains
    its maximum at a unique point $\Xi_\Gamma^\star$.
\end{lemma}

\begin{proof}
    The bound $h_\Gamma<+\infty$ is \cref{eq:unified_weak_duality} tested at the plan $\pi_0$ of \textbf{(S1)}. Assume $h_\Gamma$ finite. For fixed $\pi$, $\Xi\mapsto\scal{\Xi,A[\pi]}-\scal{\Gamma,B[\pi]}+\Reg(\pi)$ is affine and continuous, so $\Xi\mapsto\Reg\OT(c_{\Gamma,\Xi})$, an infimum of such functions, is concave and upper semicontinuous, and $h_\Gamma=\Reg\OT(c_{\Gamma,\cdot})-\tfrac12\norm\cdot^2$ is $1$-strongly concave and upper semicontinuous. A finite, upper semicontinuous, strongly concave function on a Hilbert space attains its supremum at exactly one point.
\end{proof}

\begin{proposition}[Gradient of $h_\Gamma$ and exact fixed point]
    \label{prop:unified_danskin}
    Fix $\Gamma$ and assume that $h_\Gamma$ is finite on $\Hilb_A$ and
    differentiable at $\Xi\in\Hilb_A$. If $\pi_{\Gamma,\Xi}$ is any minimizer
    of $\pi\mapsto\scal{c_{\Gamma,\Xi},\pi}+\Reg(\pi)$ over $\mathcal P$,
    then
    \begin{equation}
        \label{eq:unified_gradient_formula}
        \nabla h_\Gamma(\Xi)=A[\pi_{\Gamma,\Xi}]-\Xi.
    \end{equation}
    In particular, if $h_\Gamma$ is differentiable on $\Hilb_A$ and the
    $\Reg\OT$ problem at $c_{\Gamma,\Xi^\star_\Gamma}$ admits a minimizer
    $\pi^\star_\Gamma$, where $\Xi^\star_\Gamma$ is the maximizer of
    \cref{lem:unified_strong_concavity}, then
    \begin{equation}
        \label{eq:unified_xi_fixed_point}
        \Xi^\star_\Gamma=A[\pi^\star_\Gamma],
        \qquad
        \pi^\star_\Gamma\in\operatorname*{argmin}_{\pi\in\mathcal P}Q_\Gamma(\pi),
        \qquad
        \max_{\Xi\in\Hilb_A}h_\Gamma(\Xi)=\min_{\pi\in\mathcal P}Q_\Gamma(\pi).
    \end{equation}
\end{proposition}

\begin{proof}
    Write $v(\Xi):=\Reg\OT(c_{\Gamma,\Xi})$, so that
    $h_\Gamma=v-\tfrac12\norm\cdot^2$ and $v$ is finite on $\Hilb_A$.

    \emph{A supergradient of $v$.} Let $\pi:=\pi_{\Gamma,\Xi}$ and let
    $\Xi'\in\Hilb_A$. Since $\pi$ is admissible in the infimum defining
    $v(\Xi')$, and since $c_{\Gamma,\Xi'}=c_{\Gamma,\Xi}+A^\top(\Xi'-\Xi)$,
    \begin{equation}
        \label{eq:danskin_supergradient}
        v(\Xi')
        \le\scal{c_{\Gamma,\Xi'},\pi}+\Reg(\pi)
        =\scal{c_{\Gamma,\Xi},\pi}+\Reg(\pi)+\scal{A[\pi],\Xi'-\Xi}
        =v(\Xi)+\scal{A[\pi],\Xi'-\Xi},
    \end{equation}
    where the last equality is the optimality of $\pi$ at the cost
    $c_{\Gamma,\Xi}$. Thus $A[\pi]$ is a supergradient of the concave
    function $v$ at $\Xi$.

    \emph{Identification of the gradient.} As $v$ is concave and differentiable at $\Xi$, its only supergradient there is $\nabla v(\Xi)$; hence $\nabla v(\Xi)=A[\pi]$ and $\nabla h_\Gamma(\Xi)=A[\pi]-\Xi$.

    \emph{Exact fixed point.} If $h_\Gamma$ is differentiable on $\Hilb_A$, then $\nabla h_\Gamma(\Xi^\star_\Gamma)=0$ at its maximizer, so \cref{eq:unified_gradient_formula} gives $\Xi^\star_\Gamma=A[\pi^\star_\Gamma]$. The plan $\pi^\star_\Gamma$ is $0$-accurate at $c_{\Gamma,\Xi^\star_\Gamma}$ with $r=A[\pi^\star_\Gamma]-\Xi^\star_\Gamma=0$, so \cref{eq:unified_primal_certificate} yields $\pi^\star_\Gamma\in\arg\min_\pi Q_\Gamma(\pi)$. Finally, by \cref{eq:unified_h}, the optimality of $\pi^\star_\Gamma$ and $\Xi^\star_\Gamma=A[\pi^\star_\Gamma]$,
    \begin{align*}
        h_\Gamma(\Xi^\star_\Gamma)
        &=-\tfrac12\norm{\Xi^\star_\Gamma}^2+\scal{\Xi^\star_\Gamma,A[\pi^\star_\Gamma]}-\scal{\Gamma,B[\pi^\star_\Gamma]}+\Reg(\pi^\star_\Gamma)
        =Q_\Gamma(\pi^\star_\Gamma),
    \end{align*}
    so \cref{eq:unified_weak_duality} is an equality, attained on both sides.
\end{proof}

Differentiability of $h_\Gamma$ is part of the smoothness assumption
\textbf{(H$_L$)} below, and holds for the entropic regularizers on finite
spaces, where the $\Reg\OT$ minimizer is moreover unique and $C^1$ in $\Xi$
(\cref{prop:ugw_klkl_regularities}).

\paragraph{Smoothness assumption.}
We now assume that the reduced objective is smooth, in one of two forms.
\begin{enumerate}
    \item[\textbf{(H$_L$)}] $h_\Gamma$ is finite on $\Hilb_A$, and for some
    $L\ge1$ either (\emph{global form}) $h_\Gamma$ is $L$-smooth on
    $\Hilb_A$, or (\emph{local form}) there is $R>0$ such that
    $\norm{\Xi^\star_\Gamma}\le R$, the inner iterates satisfy
    $\norm{\Xi_k}\le R$ for all $k$, $\delta_{\Reg\OT}\le\tfrac12$, and
    $h_\Gamma$ is twice differentiable with
    $-LI\preceq\nabla^2h_\Gamma\preceq-I$ on the ball
    $\mathcal B_\Gamma:=\{\norm\Xi\le3R+1\}$.
\end{enumerate}
The local form is the one met in practice (\cref{prop:ugw_klkl_regularities}). The proofs below use $L$-smoothness only through the descent lemma
\begin{equation}
    \label{eq:descent_lemma}
    \bigl|h_\Gamma(\Xi')-h_\Gamma(\Xi)-\scal{\nabla h_\Gamma(\Xi),\Xi'-\Xi}\bigr|\le\tfrac L2\norm{\Xi'-\Xi}^2,
\end{equation}
which holds whenever $\nabla h_\Gamma$ is $L$-Lipschitz along the segment $[\Xi,\Xi']$, and they apply it only along four families of segments:
\begin{enumerate}
    \item[(a)] $[\Xi_k,\Xi_k+d]$ with $\norm d\le\sqrt{2\delta_{\Reg\OT}/L}\le1$ (\cref{lem:unified_nested_gradient_error});
    \item[(b)] $[\Xi_k,\Xi_{k+1}]$ and (c) $[\Xi_k,\Xi^\star_\Gamma]$ (\cref{prop:inner_geometric});
    \item[(d)] $[\Xi_k,\Xi_k+\tfrac1L\nabla h_\Gamma(\Xi_k)]$ (\cref{prop:inner_geometric}).
\end{enumerate}
In the local form, both endpoints of (b) and (c) lie in $\{\norm\Xi\le R\}$ and both endpoints of (a) lie in $\{\norm\Xi\le R+1\}$; since balls are convex, these segments lie in $\mathcal B_\Gamma$, where $\nabla^2h_\Gamma\succeq-LI$, so \cref{eq:descent_lemma} holds along them. The following lemma handles (d).

\begin{lemma}[The exact gradient step stays in the ball]
    \label{lem:exact_step_in_ball}
    In the local form of \textbf{(H$_L$)}, for every inner iterate $\Xi_k$,
    \[
        \bigl\|\Xi_k+\tfrac1L\nabla h_\Gamma(\Xi_k)-\Xi^\star_\Gamma\bigr\|\le\norm{\Xi_k-\Xi^\star_\Gamma}\le2R,
    \]
    so that the segment $[\Xi_k,\Xi_k+\tfrac1L\nabla h_\Gamma(\Xi_k)]$ lies in $\{\norm\Xi\le3R\}\subset\mathcal B_\Gamma$.
\end{lemma}

\begin{proof}
    Set $d:=\Xi_k-\Xi^\star_\Gamma$, so that $\norm d\le2R$ and the segment $[\Xi^\star_\Gamma,\Xi_k]$ lies in $\{\norm\Xi\le R\}$, where $-LI\preceq\nabla^2h_\Gamma\preceq-I$. Since $\nabla h_\Gamma(\Xi^\star_\Gamma)=0$ ($h_\Gamma$ is differentiable and maximal there),
    \[
        \nabla h_\Gamma(\Xi_k)=-Hd,
        \qquad
        H:=-\int_0^1\nabla^2h_\Gamma(\Xi^\star_\Gamma+ud)\,\diff u,
    \]
    and $H$ is symmetric with $I\preceq H\preceq LI$, hence $0\preceq I-\tfrac1LH\preceq(1-\tfrac1L)I$. Therefore
    \[
        \bigl\|\Xi_k+\tfrac1L\nabla h_\Gamma(\Xi_k)-\Xi^\star_\Gamma\bigr\|
        =\bigl\|(I-\tfrac1LH)\,d\bigr\|\le\norm d\le2R,
    \]
    so the exact gradient step lies in $\{\norm\Xi\le3R\}$, and so does the segment joining it to $\Xi_k$, by convexity.
\end{proof}

Hence every argument below is valid in both forms of \textbf{(H$_L$)}, and we do not distinguish them.

\begin{lemma}[OT accuracy controls the gradient error]
    \label{lem:unified_nested_gradient_error}
    Assume \textbf{(H$_L$)} at $\Gamma$. Let $\widehat\pi_k$ be a
    $\delta_{\Reg\OT}$-accurate plan at $c_{\Gamma,\Xi_k}$ and define
    $\widehat g_k:=A[\widehat\pi_k]-\Xi_k$. Then
    \begin{equation}
        \label{eq:unified_nested_gradient_error}
        \norm{
            \widehat g_k-\nabla h_\Gamma(\Xi_k)
        }
        \le
        \sqrt{2L\delta_{\Reg\OT}}.
    \end{equation}
\end{lemma}

\begin{proof}
    Set $e_k:=\widehat g_k-\nabla h_\Gamma(\Xi_k)$. Combining \cref{eq:unified_approx_first_order} (without its quadratic term) with the descent lemma \cref{eq:descent_lemma} along $[\Xi_k,\Xi_k+d]$ gives, for every $d$ with $\norm d\le1$ (family (a)),
    \[
        -\scal{e_k,d}\le\tfrac L2\norm d^2+\delta_{\Reg\OT}.
    \]
    If $e_k\ne0$, take $d=-s\,e_k/\norm{e_k}$ with $s\in(0,1]$: then $s\norm{e_k}\le\tfrac L2s^2+\delta_{\Reg\OT}$. For $\delta_{\Reg\OT}>0$ the choice $s=\sqrt{2\delta_{\Reg\OT}/L}\le1$ gives $s\norm{e_k}\le2\delta_{\Reg\OT}$, i.e.\ $\norm{e_k}\le\sqrt{2L\delta_{\Reg\OT}}$; for $\delta_{\Reg\OT}=0$, letting $s\downarrow0$ gives $e_k=0$.
\end{proof}

\begin{proposition}[Inner loop: geometric convergence and certification]
    \label{prop:inner_geometric}
    Assume \textbf{(S1)}--\textbf{(S2)} and \textbf{(H$_L$)} at $\Gamma$,
    and run the inner loop of \cref{alg:ugw} at $\Gamma$ with step
    $\tau=1/L$ and accuracy $\delta_{\Reg\OT}$. Let
    $\widehat g_k:=A[\widehat\pi_k]-\Xi_k$,
    $\Delta_k:=h_\Gamma(\Xi^\star_\Gamma)-h_\Gamma(\Xi_k)$ and
    $q:=1-1/L$. Then, for every $k\ge0$:
    \begin{enumerate}
        \item[\rm(i)] $\Delta_{k+1}\le q\,\Delta_k+\delta_{\Reg\OT}$, hence
        \begin{equation}
            \label{eq:unified_smooth_constant_error}
            \Delta_k\le q^k\Delta_0+L\,\delta_{\Reg\OT};
        \end{equation}
        \item[\rm(ii)] the certificate satisfies
        \begin{equation}
            \label{eq:unified_smooth_certificate_bound}
            \operatorname{Cert}_\Gamma(\Xi_k,\widehat\pi_k)
            =\tfrac12\norm{\widehat g_k}^2+\delta_{\Reg\OT}
            \le2L\Delta_k+(2L+1)\delta_{\Reg\OT};
        \end{equation}
        \item[\rm(iii)] for $\eta\in(0,1]$, if $\delta_{\Reg\OT}=\eta/(2C_L)$
        with $C_L:=2L^2+2L+1$, then
        $\operatorname{Cert}_\Gamma(\Xi_k,\widehat\pi_k)\le\eta$ for every
        \begin{equation}
            \label{eq:unified_smooth_count}
            k\ge N_\Xi(\eta):=\Bigl\lceil L\log\max\Bigl\{1,\frac{4L\Delta_0}{\eta}\Bigr\}\Bigr\rceil,
        \end{equation}
        so that the inner loop stopped at level $\eta$ makes at most
        $N_\Xi(\eta)+1$ $\Reg\OT$ solves. The same holds with $\Delta_0$
        replaced by any upper bound, and
        $\Delta_0\le\tfrac L2\norm{\Xi_0-\Xi^\star_\Gamma}^2$
        ($\le2LR^2$ in the local form).
    \end{enumerate}
\end{proposition}

\begin{proof}
    Throughout, let $\psi:=-h_\Gamma$. By \cref{lem:unified_strong_concavity}, $\psi$ is $1$-strongly convex on $\Hilb_A$, and by \textbf{(H$_L$)} the descent lemma \cref{eq:descent_lemma} holds along the segments (a)--(d) listed after \textbf{(H$_L$)}. We write
    \[
        g_k:=\nabla h_\Gamma(\Xi_k)=-\nabla\psi(\Xi_k)
        \qquad\text{and}\qquad
        e_k:=\widehat g_k-g_k
    \]
    for the exact gradient and the gradient error, so that the $\Xi$-update of \cref{alg:ugw} with $\tau=1/L$ reads
    \begin{equation}
        \label{eq:inner_update_decomposed}
        \Xi_{k+1}=\Xi_k+\tfrac1L\,\widehat g_k=\Xi_k+\tfrac1L\,(g_k+e_k),
    \end{equation}
    and \cref{lem:unified_nested_gradient_error} gives
    \begin{equation}
        \label{eq:inner_error_bound}
        \norm{e_k}^2\le2L\,\delta_{\Reg\OT}.
    \end{equation}

    \medskip
    \noindent\emph{One inexact gradient step.}
    By \cref{eq:inner_update_decomposed}, $\Xi_{k+1}-\Xi_k=\tfrac1L(g_k+e_k)$. The descent lemma \cref{eq:descent_lemma} for $\psi$ along the segment $[\Xi_k,\Xi_{k+1}]$ (family (b)) gives
    \begin{align*}
        \psi(\Xi_{k+1})
        &\le\psi(\Xi_k)+\scal{\nabla\psi(\Xi_k),\Xi_{k+1}-\Xi_k}+\tfrac L2\norm{\Xi_{k+1}-\Xi_k}^2\\
        &=\psi(\Xi_k)-\tfrac1L\scal{g_k,g_k+e_k}+\tfrac1{2L}\norm{g_k+e_k}^2 .
    \end{align*}
    We expand the last two terms:
    \begin{align*}
        -\tfrac1L\scal{g_k,g_k+e_k}
        &=-\tfrac1L\norm{g_k}^2-\tfrac1L\scal{g_k,e_k},\\
        \tfrac1{2L}\norm{g_k+e_k}^2
        &=\tfrac1{2L}\norm{g_k}^2+\tfrac1L\scal{g_k,e_k}+\tfrac1{2L}\norm{e_k}^2 .
    \end{align*}
    The two cross terms $\mp\tfrac1L\scal{g_k,e_k}$ cancel, and we are left with
    \begin{equation}
        \label{eq:inner_one_step}
        \psi(\Xi_{k+1})\le\psi(\Xi_k)-\tfrac1{2L}\norm{g_k}^2+\tfrac1{2L}\norm{e_k}^2 .
    \end{equation}
    Hence, the exact part of the step decreases $\psi$ by $\tfrac1{2L}\norm{g_k}^2$, and the error can cost at most $\tfrac1{2L}\norm{e_k}^2$.

    \medskip
    \noindent\emph{The gradient dominates the gap.}
    By $1$-strong convexity of $\psi$, for every $\Xi'\in\Hilb_A$,
    \[
        \psi(\Xi')\ge\psi(\Xi_k)+\scal{\nabla\psi(\Xi_k),\Xi'-\Xi_k}+\tfrac12\norm{\Xi'-\Xi_k}^2 .
    \]
    The right-hand side is a quadratic function of $\Xi'$; it is minimized at $\Xi'=\Xi_k-\nabla\psi(\Xi_k)$, where its value is $\psi(\Xi_k)-\tfrac12\norm{\nabla\psi(\Xi_k)}^2$. Hence the inequality holds a fortiori with this minimal value on the right, for every $\Xi'$; taking $\Xi'=\Xi^\star_\Gamma$ on the left,
    \[
        \psi(\Xi^\star_\Gamma)\ge\psi(\Xi_k)-\tfrac12\norm{\nabla\psi(\Xi_k)}^2,
    \]
    that is,
    \[
        \norm{g_k}^2=\norm{\nabla\psi(\Xi_k)}^2\ge2\bigl[\psi(\Xi_k)-\psi(\Xi^\star_\Gamma)\bigr]=2\Delta_k .
    \]
    We record this as
    \begin{equation}
        \label{eq:unified_PL}
        \norm{g_k}^2\ge2\Delta_k .
    \end{equation}

    \medskip
    \noindent\emph{The recursion (i).}
    Subtract $\psi(\Xi^\star_\Gamma)$ from both sides of \cref{eq:inner_one_step}; since $\Delta_j=\psi(\Xi_j)-\psi(\Xi^\star_\Gamma)$ for every $j$, this gives
    \[
        \Delta_{k+1}\le\Delta_k-\tfrac1{2L}\norm{g_k}^2+\tfrac1{2L}\norm{e_k}^2 .
    \]
    We now use \cref{eq:unified_PL} to bound $-\tfrac1{2L}\norm{g_k}^2\le-\tfrac1{2L}\cdot2\Delta_k=-\tfrac1L\Delta_k$, and \cref{eq:inner_error_bound} to bound $\tfrac1{2L}\norm{e_k}^2\le\tfrac1{2L}\cdot2L\delta_{\Reg\OT}=\delta_{\Reg\OT}$. Therefore
    \[
        \Delta_{k+1}\le\Delta_k-\tfrac1L\Delta_k+\delta_{\Reg\OT}=q\,\Delta_k+\delta_{\Reg\OT},
        \qquad q=1-\tfrac1L .
    \]
    Applying this inequality $k$ times, from $\Delta_0$ to $\Delta_k$, we obtain
    \[
        \Delta_k
        \le q\Delta_{k-1}+\delta_{\Reg\OT}
        \le q^2\Delta_{k-2}+q\,\delta_{\Reg\OT}+\delta_{\Reg\OT}
        \le\cdots
        \le q^k\Delta_0+\delta_{\Reg\OT}\sum_{j=0}^{k-1}q^j .
    \]
    Since $0\le q<1$, the geometric sum satisfies $\sum_{j=0}^{k-1}q^j\le\sum_{j=0}^{\infty}q^j=\frac1{1-q}=L$, which proves \cref{eq:unified_smooth_constant_error}.

    \medskip
    \noindent\emph{The certificate bound (ii).}
    By definition, $\operatorname{Cert}_\Gamma(\Xi_k,\widehat\pi_k)=\tfrac12\norm{\widehat g_k}^2+\delta_{\Reg\OT}$, so we need to bound $\norm{\widehat g_k}=\norm{g_k+e_k}$. We first bound the exact gradient by the gap. The descent lemma \cref{eq:descent_lemma} for $h_\Gamma$ along the segment $[\Xi_k,\Xi_k+\tfrac1Lg_k]$ (family (d)) gives
    \begin{align*}
        h_\Gamma\bigl(\Xi_k+\tfrac1Lg_k\bigr)
        &\ge h_\Gamma(\Xi_k)+\scal{g_k,\tfrac1Lg_k}-\tfrac L2\bigl\|\tfrac1Lg_k\bigr\|^2\\
        &=h_\Gamma(\Xi_k)+\tfrac1L\norm{g_k}^2-\tfrac1{2L}\norm{g_k}^2
        =h_\Gamma(\Xi_k)+\tfrac1{2L}\norm{g_k}^2 .
    \end{align*}
    Since $\Xi^\star_\Gamma$ maximizes $h_\Gamma$, the left-hand side is at most $h_\Gamma(\Xi^\star_\Gamma)$, and we get
    \[
        \tfrac1{2L}\norm{g_k}^2\le h_\Gamma(\Xi^\star_\Gamma)-h_\Gamma(\Xi_k)=\Delta_k,
        \qquad\text{that is,}\qquad
        \norm{g_k}\le\sqrt{2L\Delta_k}.
    \]
    Together with $\norm{e_k}\le\sqrt{2L\delta_{\Reg\OT}}$ from \cref{eq:inner_error_bound}, the triangle inequality gives
    \[
        \norm{\widehat g_k}\le\norm{g_k}+\norm{e_k}\le\sqrt{2L\Delta_k}+\sqrt{2L\delta_{\Reg\OT}} .
    \]
    Squaring and using $(a+b)^2\le2a^2+2b^2$,
    \[
        \tfrac12\norm{\widehat g_k}^2
        \le\tfrac12\bigl(2\cdot2L\Delta_k+2\cdot2L\delta_{\Reg\OT}\bigr)
        =2L\Delta_k+2L\delta_{\Reg\OT}.
    \]
    Adding $\delta_{\Reg\OT}$ on both sides proves \cref{eq:unified_smooth_certificate_bound}.

    \medskip
    \noindent\emph{The iteration count (iii).}
    Fix $\eta\in(0,1]$ and set $\delta_{\Reg\OT}=\eta/(2C_L)$. Inserting the bound \cref{eq:unified_smooth_constant_error} on $\Delta_k$ into the certificate bound \cref{eq:unified_smooth_certificate_bound},
    \begin{align*}
        \operatorname{Cert}_\Gamma(\Xi_k,\widehat\pi_k)
        &\le2L\bigl(q^k\Delta_0+L\delta_{\Reg\OT}\bigr)+(2L+1)\delta_{\Reg\OT}\\
        &=2L\,q^k\Delta_0+\bigl(2L^2+2L+1\bigr)\delta_{\Reg\OT}\\
        &=2L\,q^k\Delta_0+C_L\,\delta_{\Reg\OT}.
    \end{align*}
    By the choice of $\delta_{\Reg\OT}$, the second term is exactly $C_L\cdot\eta/(2C_L)=\eta/2$. It therefore suffices to show that the first term is at most $\eta/2$ for $k\ge N_\Xi(\eta)$.

    Since $\log(1-x)\le-x$ for $x\in[0,1)$, we have $q^k=(1-\tfrac1L)^k=e^{k\log(1-1/L)}\le e^{-k/L}$ when $L>1$ (and $q^k\le e^{-k/L}$ trivially when $L=1$, since then $q=0$). Hence $2Lq^k\Delta_0\le\eta/2$ holds as soon as
    \[
        2L\,e^{-k/L}\Delta_0\le\frac\eta2
        \quad\Longleftrightarrow\quad
        e^{-k/L}\le\frac{\eta}{4L\Delta_0}
        \quad\Longleftrightarrow\quad
        k\ge L\log\frac{4L\Delta_0}{\eta}
    \]
    (if $\Delta_0=0$ there is nothing to prove). If $4L\Delta_0/\eta\le1$, the right-hand side is nonpositive and the condition holds for every $k\ge0$; otherwise it holds for every $k\ge L\log(4L\Delta_0/\eta)$. Both cases are summarized by
    \[
        k\ge L\log\max\Bigl\{1,\frac{4L\Delta_0}\eta\Bigr\},
    \]
    and since $k$ is an integer, by $k\ge N_\Xi(\eta)$ with $N_\Xi(\eta)$ as in \cref{eq:unified_smooth_count}. This proves $\operatorname{Cert}_\Gamma(\Xi_k,\widehat\pi_k)\le\eta$ for every $k\ge N_\Xi(\eta)$.

    Three remarks complete the proof of (iii).
    \begin{itemize}
        \item \emph{Upper bounds on $\Delta_0$.} If $\Delta_0$ is replaced by any larger number $\overline\Delta_0$ in the argument above, the conclusion still holds, since $2Lq^k\overline\Delta_0\ge2Lq^k\Delta_0$; hence $N_\Xi(\eta)$ computed with $\overline\Delta_0$ is also a valid count.
        \item \emph{Number of solves.} The inner loop of \cref{alg:ugw} stops at the first index $k$ with $\operatorname{Cert}_\Gamma(\Xi_k,\widehat\pi_k)\le\eta$. Since the certificate holds at $k=N_\Xi(\eta)$, this first index is at most $N_\Xi(\eta)$, and the loop makes at most $N_\Xi(\eta)+1$ $\Reg\OT$ solves (one per index $k=0,\dots,N_\Xi(\eta)$).
        \item \emph{The bound $\Delta_0\le\tfrac L2\norm{\Xi_0-\Xi^\star_\Gamma}^2$.} The descent lemma \cref{eq:descent_lemma} for $h_\Gamma$ along the segment $[\Xi^\star_\Gamma,\Xi_0]$ (family (c) with $k=0$), together with $\nabla h_\Gamma(\Xi^\star_\Gamma)=0$, gives
        \begin{align*}
            h_\Gamma(\Xi_0)
            &\ge h_\Gamma(\Xi^\star_\Gamma)+\scal{\nabla h_\Gamma(\Xi^\star_\Gamma),\Xi_0-\Xi^\star_\Gamma}-\tfrac L2\norm{\Xi_0-\Xi^\star_\Gamma}^2\\
            &=h_\Gamma(\Xi^\star_\Gamma)-\tfrac L2\norm{\Xi_0-\Xi^\star_\Gamma}^2,
        \end{align*}
        that is, $\Delta_0\le\tfrac L2\norm{\Xi_0-\Xi^\star_\Gamma}^2$. In the local form, $\norm{\Xi_0-\Xi^\star_\Gamma}\le\norm{\Xi_0}+\norm{\Xi^\star_\Gamma}\le2R$, so $\Delta_0\le\tfrac L2(2R)^2=2LR^2$.
    \end{itemize}
    Finally, the local form of \textbf{(H$_L$)} requires $\delta_{\Reg\OT}\le\tfrac12$; this is satisfied, because $\eta\le1$ and $C_L=2L^2+2L+1\ge5$ for $L\ge1$ give $\delta_{\Reg\OT}=\eta/(2C_L)\le\tfrac1{10}$.
\end{proof}

\begin{corollary}[Oracle complexity of \cref{alg:ugw}; precise form of \cref{thm:main_convergence}(ii)]
    \label{cor:unified_oracle_complexity}
    Assume \textbf{(S1)}--\textbf{(S2)}, let $\delta\in(0,1]$, and suppose
    that \textbf{(H$_L$)} holds at $\Gamma=\Gamma_t$ for every $t<T_\delta$
    with the same constants $L$ (and $R$). Let
    \[
        \overline\Delta
        :=
        \sup_{t<T_\delta}
        \Bigl[h_{\Gamma_t}(\Xi^\star_{\Gamma_t})-h_{\Gamma_t}(\Xi_{t,0})\Bigr]
        \le\frac L2\sup_{t<T_\delta}\norm{\Xi_{t,0}-\Xi^\star_{\Gamma_t}}^2
    \]
    be the largest initial inner gap. Run \cref{alg:ugw} with $\tau=1/L$ and
    $\delta_{\Reg\OT}=\delta^2/(8C_L)$, where $C_L=2L^2+2L+1$. Then every
    inner loop terminates, with
    \[
        K_t\le N_\Xi(\delta^2/4)
        =\Bigl\lceil L\log\max\Bigl\{1,\frac{16L\overline\Delta}{\delta^2}\Bigr\}\Bigr\rceil,
    \]
    the conclusions of \cref{thm:unified_outer} hold, and the total number of
    $\Reg\OT$ solves is at most
    \[
        T_\delta\bigl(N_\Xi(\delta^2/4)+1\bigr)
        =O\Bigl(\frac{\Delta^{\mathrm{out}}L}{\delta^2}\log\frac{L\overline\Delta}{\delta^2}\Bigr).
    \]
\end{corollary}

\begin{proof}
    Fix an outer iteration $t<T_\delta$ and apply \cref{prop:inner_geometric}(iii) at $\Gamma=\Gamma_t$ with the level $\eta:=\delta^2/4\in(0,1]$. Its accuracy requirement $\delta_{\Reg\OT}=\eta/(2C_L)$ is exactly the value prescribed here, since
    \[
        \frac{\eta}{2C_L}=\frac{\delta^2/4}{2C_L}=\frac{\delta^2}{8C_L},
    \]
    and its iteration threshold, computed with the upper bound $\overline\Delta\ge\Delta_{t,0}$ allowed by \cref{prop:inner_geometric}(iii), is
    \[
        N_\Xi(\eta)=\Bigl\lceil L\log\max\Bigl\{1,\frac{4L\overline\Delta}{\eta}\Bigr\}\Bigr\rceil
        =\Bigl\lceil L\log\max\Bigl\{1,\frac{4L\overline\Delta}{\delta^2/4}\Bigr\}\Bigr\rceil
        =\Bigl\lceil L\log\max\Bigl\{1,\frac{16L\overline\Delta}{\delta^2}\Bigr\}\Bigr\rceil,
    \]
    which is the quantity $N_\Xi(\delta^2/4)$ of the statement. Hence the certificate of iteration $t$ is at most $\delta^2/4$ for every inner index $k\ge N_\Xi(\delta^2/4)$, so that the first certified index satisfies $K_t\le N_\Xi(\delta^2/4)<\infty$.

    Since every inner loop terminates, \cref{thm:unified_outer} applies: the run stops at some $t_{\mathrm{stop}}<T_\delta$, and it makes $\sum_{t\le t_{\mathrm{stop}}}(K_t+1)\le T_\delta\bigl(N_\Xi(\delta^2/4)+1\bigr)$ $\Reg\OT$ solves. Finally, $T_\delta\le1+4\Delta^{\mathrm{out}}/\delta^2$ and $N_\Xi(\delta^2/4)+1\le2+L\log\max\{1,16L\overline\Delta/\delta^2\}$, whose product is $O\bigl(\tfrac{\Delta^{\mathrm{out}}L}{\delta^2}\log\tfrac{L\overline\Delta}{\delta^2}\bigr)$. The bound on $\overline\Delta$ in the statement is the last claim of \cref{prop:inner_geometric}(iii), applied at each $t$.
\end{proof}

\subsection{Accelerated versions.}\label{sub:acceleraded_versions}

The previous analysis uses plain gradient ascent for the reduced problem in
$\Xi$. The same framework immediately allows acceleration whenever the
$\Reg\OT$ oracle converges geometrically.

\begin{proposition}[Acceleration with a geometric $\Reg\OT$ oracle]
\label{prop:accelerated_inner}
Fix $\Gamma$ and let $h_\Gamma$ be $L$-smooth on
a convex region containing its maximizer and all iterates. Since $h_\Gamma$ is $1$-strongly convex, set
$\kappa:=$. Assume that a $\zeta$-accurate $\Reg\OT$ solve can be
computed in
\begin{equation}
    N_{\Reg\OT}(\zeta)
    =
    O\!\left(\log\frac{1}{\zeta}\right)
\end{equation}
iterations, up to constants independent of the target accuracy.

Then, for any target inner accuracy $\eta>0$, an accelerated inexact-gradient
method applied to $-h_\Gamma$ reaches
\[
    \operatorname{Cert}_\Gamma(\Xi,\widehat\pi)\le\eta
\]
using
\begin{equation}
    O\!\left(
        \sqrt{\kappa}\,
        \log\frac{L\Delta_0}{\eta}
    \right)
\end{equation}
$\Reg\OT$ calls, provided each call is solved to sufficiently small accuracy,
for instance
\begin{equation}
    \zeta
    =
    O\!\left(
        \frac{\eta}{L\sqrt{\kappa}}
    \right).
\end{equation}
Consequently, the total number of iterations of the underlying
$\Reg\OT$ solver is
\begin{equation}
    O\!\left(
        \sqrt{\kappa}\,
        \log\frac{L\Delta_0}{\eta}\,
        \log\frac{L\sqrt{\kappa}}{\eta}
    \right)
    =
    \widetilde O(\sqrt{\kappa}).
\end{equation}
In particular, arbitrarily accurate $\Reg\OT$ solves preserve, up to
logarithmic factors, the usual accelerated dependence on the condition
number.
\end{proposition}

\begin{proof}
A $\zeta$-accurate $\Reg\OT$ solution provides an inexact first-order oracle
for $h_\Gamma$; in particular, by
\cref{lem:unified_nested_gradient_error},
\[
    \|\widehat g-\nabla h_\Gamma(\Xi)\|
    \le \sqrt{2L\zeta}.
\]
Standard accelerated methods for smooth strongly convex optimization with
an inexact oracle converge in
$O(\sqrt{\kappa}\log(1/\eta))$ iterations, with an error floor proportional
to $\sqrt{\kappa}\,\zeta$ \citep{devolder2013strongly}. Choosing
$\zeta=O(\eta/(L\sqrt{\kappa}))$ makes this error negligible at the accuracy
required by the certificate. Since the $\Reg\OT$ solver is geometrically
convergent, this increased oracle accuracy changes the cost of each call
only logarithmically.
\end{proof}

\section{Entropic unbalanced Gromov--Wasserstein: proof of Theorem~\ref{thm:ugw_end_to_end}}
\label{subsec:unified_ugw}

We now specialize the preceding analysis to entropically regularized
unbalanced Gromov--Wasserstein with KL penalties on both marginals. The UOT
subproblems are solved by generalized Sinkhorn iterations
\cite{chizat2018scaling,pham2020unbalanced}. For standard facts on entropic
UOT we follow \citet[Section~4]{sejourne2023unbalanced}, whose normalization
(entropy relative to $\alpha\otimes\beta$, potentials $(f,g)$) coincides with
ours. We focus on the additional regularity
induced by the Gromov--Wasserstein DC structure and the end-to-end convergence of our algorithm for this specific case.

We denote by $\varepsilon>0$ the entropic regularization parameter, by
$\rho_1,\rho_2>0$ the marginal penalties (\cref{table:regularizers}), by
$\delta_{\Reg\OT}>0$ the accuracy of the UOT oracle
\cref{eq:unified_ot_oracle}, and by $\delta>0$ the DC-criticality tolerance of
\cref{thm:unified_outer}.

We consider discrete measures
\[
    \alpha=\sum_{i=1}^{\NN}\alpha_i\delta_{x_i},
    \qquad
    \beta=\sum_{j=1}^{\MM}\beta_j\delta_{y_j},
\]
with strictly positive weights, and set
\[
    M_\alpha:=\sum_i\alpha_i,
    \qquad
    M_\beta:=\sum_j\beta_j.
\]
Writing $a_{ij}:=A(x_i,y_j)$ and $b_{ij}:=B(x_i,y_j)$, we have
\[
    A[\pi]=\sum_{ij}\pi_{ij}a_{ij},
    \qquad
    B[\pi]=\sum_{ij}\pi_{ij}b_{ij},
\]
and we let
\[
    R_A:=\max_{i,j}\norm{a_{ij}},
    \qquad
    R_B:=\max_{i,j}\norm{b_{ij}}.
\]
The objective is \cref{eq:unified_E} with
\begin{equation}
\label{eq:ugw_klkl_reg}
\Reg_{\mathrm{UOT}}(\pi)
:=
\varepsilon\KL(\pi\Vert\alpha\otimes\beta)
+
\rho_1\KL(\pi_1\Vert\alpha)
+
\rho_2\KL(\pi_2\Vert\beta),
\end{equation}
where $\KL(p\Vert q):=\sum_l\bigl[p_l\log(p_l/q_l)-p_l+q_l\bigr]$ is the
generalized Kullback--Leibler divergence between nonnegative vectors. This is
the entropic UOT functional of \citet[Definition~7]{sejourne2023unbalanced}
with $D_{\varphi_i}=\rho_i\KL$. Sinkhorn rates are stated in the symmetric case
$\rho_1=\rho_2=\rho$.

At fixed $(\Gamma,\Xi)$, define
\begin{equation}
\label{eq:ugw_klkl_inner}
P_{\Gamma,\Xi}(\pi)
:=
\scal{c_{\Gamma,\Xi},\pi}
+
\Reg_{\mathrm{UOT}}(\pi),
\qquad
\mathrm{UOT}(c_{\Gamma,\Xi})
:=
\inf_{\pi\ge0}P_{\Gamma,\Xi}(\pi),
\end{equation}
with
\begin{equation}
\label{eq:ugw_klkl_cost}
c_{\Gamma,\Xi,ij}
=
\scal{\Xi,a_{ij}}
-
\scal{\Gamma,b_{ij}}.
\end{equation}
Thus $\mathrm{UOT}(c)$ is the value $\Reg\OT(c)$ of \cref{eq:unified_V} for
$\Reg=\Reg_{\mathrm{UOT}}$, and the corresponding reduced objective is
\begin{equation}
\label{eq:ugw_reduced_h}
h_\Gamma(\Xi)
:=
-\frac12\norm{\Xi}^2
+
\mathrm{UOT}(c_{\Gamma,\Xi}).
\end{equation}

\paragraph{Regularity with respect to the GW variable.}

\begin{proposition}[Regularity of the reduced UGW problem]
\label{prop:ugw_klkl_regularities}
Fix $\Gamma$ and $R>0$, and set
\[
    C_\Gamma(R)
    :=
    R_AR+R_B\norm{\Gamma},
\]
\begin{equation}
\label{eq:ugw_mass_bound}
M_\Gamma(R)
:=
M_\alpha^{
\frac{\varepsilon+\rho_1}{\varepsilon+\rho_1+\rho_2}
}
M_\beta^{
\frac{\varepsilon+\rho_2}{\varepsilon+\rho_1+\rho_2}
}
\exp\left(
\frac{C_\Gamma(R)}{\varepsilon+\rho_1+\rho_2}
\right).
\end{equation}
Then:
\begin{enumerate}
    \item[\rm (i)]
    for every $\Xi$, the UOT problem \cref{eq:ugw_klkl_inner} has a unique
    minimizer $\pi_{\Gamma,\Xi}$, which is entrywise positive and depends
    continuously differentiably on $\Xi$; the function $h_\Gamma$ is finite
    and twice continuously differentiable on $\Hilb_A$, with
    \begin{equation}
    \label{eq:ugw_klkl_grad_h}
        \nabla h_\Gamma(\Xi)
        =
        A[\pi_{\Gamma,\Xi}]
        -
        \Xi;
    \end{equation}

    \item[\rm (ii)]
    if $\norm{\Xi}\le R$, then
    \begin{equation}
    \label{eq:ugw_klkl_mass}
        \norm{\pi_{\Gamma,\Xi}}_1
        \le
        M_\Gamma(R);
    \end{equation}

    \item[\rm (iii)]
    $h_\Gamma$ is $1$-strongly concave on $\Hilb_A$, and on
    $\{\norm{\Xi}\le R\}$ its Hessian satisfies
    $-L_hI\preceq\nabla^2h_\Gamma\preceq-I$ with
    \begin{equation}
    \label{eq:ugw_klkl_Lh}
        L_h
        :=
        1+
        \frac{R_A^2M_\Gamma(R)}{\varepsilon}.
    \end{equation}
\end{enumerate}
\end{proposition}

\begin{proof}
\emph{(i).} Since $\varepsilon\KL(\cdot\Vert\alpha\otimes\beta)$ is strictly
convex, $P_{\Gamma,\Xi}$ has at most one minimizer. The dual of
\cref{eq:ugw_klkl_inner} \citep[Proposition~2]{sejourne2023unbalanced} is the
maximization over $(f,g)\in\R^{\NN}\times\R^{\MM}$ of a smooth, strictly
concave and coercive function. It therefore has a
unique maximizer $(f^\star,g^\star)$. The minimizer of \cref{eq:ugw_klkl_inner}
is then given by the primal--dual relation
\citep[Equation~(21)]{sejourne2023unbalanced}:
\begin{equation}
\label{eq:ugw_gibbs}
    \pi_{\Gamma,\Xi,ij}
    =\alpha_i\beta_j\exp\Bigl(\frac{f^\star_i+g^\star_j-c_{\Gamma,\Xi,ij}}{\varepsilon}\Bigr)>0 .
\end{equation}
In particular, $\mathrm{UOT}(c_{\Gamma,\Xi})$ and $h_\Gamma$ are finite.

On the open set $\{\pi>0\}$, $P_{\Gamma,\Xi}$ is smooth and
$\pi_{\Gamma,\Xi}$ is characterized by the stationarity condition
\begin{equation}
\label{eq:ugw_stationarity}
    \nabla_\pi P_{\Gamma,\Xi}(\pi)
    =c_{\Gamma,\Xi}+\nabla\Reg_{\mathrm{UOT}}(\pi)=0 .
\end{equation}
Its Jacobian with respect to $\pi$ is
$H_\pi:=\nabla^2\Reg_{\mathrm{UOT}}(\pi)\succeq\varepsilon\operatorname{diag}(1/\pi_{ij})\succ0$,
because the marginal KL terms are convex. By the implicit function theorem,
$\Xi\mapsto\pi_{\Gamma,\Xi}$ is continuously differentiable. Since
$c_{\Gamma,\Xi}=A^\top\Xi-B^\top\Gamma$, differentiating \cref{eq:ugw_stationarity}
gives
\begin{equation}
\label{eq:ugw_implicit_derivative}
    D_\Xi\pi_{\Gamma,\Xi}=-H_{\pi}^{-1}A^\top,
    \qquad\text{with }H_\pi:=H_{\pi_{\Gamma,\Xi}} .
\end{equation}
Finally, $\mathrm{UOT}(c_{\Gamma,\Xi})=P_{\Gamma,\Xi}(\pi_{\Gamma,\Xi})$, where
$(\pi,\Xi)\mapsto P_{\Gamma,\Xi}(\pi)=\scal{\Xi,A[\pi]}-\scal{\Gamma,B[\pi]}+\Reg_{\mathrm{UOT}}(\pi)$
is smooth on $\{\pi>0\}\times\Hilb_A$. By the chain rule, and since the
$\pi$-gradient vanishes at $\pi_{\Gamma,\Xi}$ by \cref{eq:ugw_stationarity}
(envelope theorem), we get $\nabla_\Xi\,\mathrm{UOT}(c_{\Gamma,\Xi})=A[\pi_{\Gamma,\Xi}]$.
This gives \cref{eq:ugw_klkl_grad_h}, in accordance with
\cref{prop:unified_danskin} and \cref{eq:unified_gradient_formula}. Since the
right-hand side is $C^1$, $h_\Gamma$ is $C^2$.

\emph{(ii).} Let $\norm\Xi\le R$, and write $\pi:=\pi_{\Gamma,\Xi}$ and
$m:=\norm{\pi}_1>0$. Since $\pi$ is a minimizer, $t\mapsto P_{\Gamma,\Xi}(t\pi)$
is minimal at $t=1$. Setting its derivative to zero, as in
\citet[Lemma~4]{pham2020unbalanced}, gives
\begin{equation*}
    0=
    \scal{c_{\Gamma,\Xi},\pi}
    +\varepsilon\sum_{ij}\pi_{ij}\log\frac{\pi_{ij}}{\alpha_i\beta_j}
    +\rho_1\sum_i(\pi_1)_i\log\frac{(\pi_1)_i}{\alpha_i}
    +\rho_2\sum_j(\pi_2)_j\log\frac{(\pi_2)_j}{\beta_j}.
\end{equation*}

We first bound the cost term. By Cauchy--Schwarz and the definitions of
$R_A$ and $R_B$,
\begin{align*}
c_{\Gamma,\Xi,ij}
&=
\scal{\Xi,a_{ij}}-\scal{\Gamma,b_{ij}} \
&\ge
-\norm{\Xi}\norm{a_{ij}}
-\norm{\Gamma}\norm{b_{ij}} \
&\ge
-R_AR-R_B\norm{\Gamma}
=
-C_\Gamma(R).
\end{align*}
Since $\pi_{ij}\ge0$, summing this entrywise bound against $\pi$ gives
\begin{align*}
\scal{c_{\Gamma,\Xi},\pi}
&=
\sum_{ij}c_{\Gamma,\Xi,ij}\pi_{ij} \
&\ge
-C_\Gamma(R)\sum_{ij}\pi_{ij}
=
-C_\Gamma(R)m.
\end{align*}

Recall that the marginals of $\pi$ are
$(\pi_1)_i = \sum_j \pi_{ij},
(\pi_2)_j = \sum_i \pi_{ij}$ .
Hence,both have the same total mass as $\pi$:
\begin{align*}
\sum_i (\pi_1)_i
=
\sum_j (\pi_2)_j
=
\sum_{i,j} \pi_{ij}
=
m.
\end{align*}

We can therefore apply the log-sum inequality to each entropy term. Since
\begin{align*}
\sum_{i,j} \alpha_i \beta_j
=
M_\alpha M_\beta,
\end{align*}
it gives
\begin{align*}
\sum_{i,j} \pi_{ij}
\log \frac{\pi_{ij}}{\alpha_i \beta_j}
&\ge
m \log \frac{m}{M_\alpha M_\beta},\\\
\sum_i (\pi_1)_i
\log \frac{(\pi_1)_i}{\alpha_i}
&\ge
m \log \frac{m}{M_\alpha},\\
\sum_j (\pi_2)_j
\log \frac{(\pi_2)_j}{\beta_j}
&\ge
m \log \frac{m}{M_\beta}.
\end{align*}
Combining these inequalities with
$\scal{c_{\Gamma,\Xi},\pi}\ge -C_\Gamma(R)m$ yields
\begin{align*}
0
&\ge
-C_\Gamma(R)m
+\varepsilon m \log \frac{m}{M_\alpha M_\beta}
+\rho_1 m \log \frac{m}{M_\alpha}
+\rho_2 m \log \frac{m}{M_\beta}.
\end{align*}
Dividing by $m>0$ and collecting the logarithmic terms gives
\begin{align*}
0
&\ge
-C_\Gamma(R)
+(\varepsilon+\rho_1+\rho_2)\log m\
&\qquad
-(\varepsilon+\rho_1)\log M_\alpha
-(\varepsilon+\rho_2)\log M_\beta.
\end{align*}

which is \cref{eq:ugw_klkl_mass} after rearranging and exponentiating.

\emph{(iii).} Strong concavity is \cref{lem:unified_strong_concavity}, since
$h_\Gamma$ is finite. Let $\norm\Xi\le R$. Differentiating
\cref{eq:ugw_klkl_grad_h} and using \cref{eq:ugw_implicit_derivative},
\begin{equation}
\label{eq:ugw_hessian_identity}
    \nabla^2 h_\Gamma(\Xi)
    =
    -I
    -
    AH_\pi^{-1}A^\top.
\end{equation}
From $H_\pi\succeq\varepsilon\operatorname{diag}(1/\pi_{ij})$ we get
$H_\pi^{-1}\preceq\varepsilon^{-1}\operatorname{diag}(\pi_{ij})$. Hence, for
$u\in\Hilb_A$, by Cauchy--Schwarz, $\norm{a_{ij}}\le R_A$ and
\cref{eq:ugw_klkl_mass},
\[
    0\le\scal{u,AH_\pi^{-1}A^\top u}
    \le\frac1\varepsilon\sum_{ij}\pi_{ij}\scal{u,a_{ij}}^2
    \le\frac{R_A^2\norm{\pi_{\Gamma,\Xi}}_1}{\varepsilon}\norm u^2
    \le\frac{R_A^2M_\Gamma(R)}{\varepsilon}\norm u^2 .
\]
With \cref{eq:ugw_hessian_identity}, this gives
$-L_hI\preceq\nabla^2h_\Gamma(\Xi)\preceq-I$.
\end{proof}

\begin{remark}[No exponential deterioration in $1/\varepsilon$]
\label{rem:ugw_klkl_lambda}
The mass bound \cref{eq:ugw_mass_bound} contains
$\exp\bigl(C_\Gamma(R)/(\varepsilon+\rho_1+\rho_2)\bigr)$ rather than
$\exp(C/\varepsilon)$. Thus, for fixed positive marginal penalties and
uniformly bounded $R$ and $\norm{\Gamma}$, $M_\Gamma(R)=O(1)$ and
$L_h=O(\varepsilon^{-1})$ as $\varepsilon\downarrow0$.
\end{remark}

\paragraph{Sinkhorn iterations.}
Assume $\rho_1=\rho_2=\rho$ and fix $(\Gamma,\Xi)$. The dual of the UOT problem
\cref{eq:ugw_klkl_inner} is solved by the unbalanced Sinkhorn algorithm
\citep{chizat2018scaling}, \citep[Definition~10]{sejourne2023unbalanced}. For
KL penalties it alternates the two half-steps
\citep[Section~4.2]{sejourne2023unbalanced}
\begin{equation}
    \label{eq:ugw_sinkhorn_halfsteps}
    f_i\leftarrow-\frac{\rho\varepsilon}{\rho+\varepsilon}\log\sum_{j}\beta_j\,e^{(g_j-c_{\Gamma,\Xi,ij})/\varepsilon},
    \qquad
    g_j\leftarrow-\frac{\rho\varepsilon}{\rho+\varepsilon}\log\sum_{i}\alpha_i\,e^{(f_i-c_{\Gamma,\Xi,ij})/\varepsilon},
\end{equation}
and encodes the current plan as
$\pi_{ij}=\alpha_i\beta_je^{(f_i+g_j-c_{\Gamma,\Xi,ij})/\varepsilon}$. Up to
the change of cost $C_{ij}=c_{\Gamma,\Xi,ij}-\varepsilon\log(\alpha_i\beta_j)$,
this is Algorithm~1 of \citet{pham2020unbalanced} with
$(\tau,\eta)=(\rho,\varepsilon)$. The optimal potentials $(f^\star,g^\star)$ are
the fixed points of \cref{eq:ugw_sinkhorn_halfsteps}
\citep[Proposition~3]{sejourne2023unbalanced}. The log-sum-exp maps in
\cref{eq:ugw_sinkhorn_halfsteps} are $1$-Lipschitz for the sup norm. Hence each
half-step is a contraction of factor $\rho/(\rho+\varepsilon)$ towards the
optimum, for any current potentials:
\begin{equation}
    \label{eq:ugw_halfstep_contraction}
    \norm{f^{\mathrm{new}}-f^\star}_\infty\le\frac{\rho}{\rho+\varepsilon}\norm{g-g^\star}_\infty,
    \qquad
    \norm{g^{\mathrm{new}}-g^\star}_\infty\le\frac{\rho}{\rho+\varepsilon}\norm{f-f^\star}_\infty .
\end{equation}
This is the key step in the proof of Theorem~1 of \citet{pham2020unbalanced},
derived there from their Lemmas~1--2. See also
\citet[Theorem~2]{sejourne2023unbalanced}, which states the resulting rate
$(\rho/(\rho+\varepsilon))^2$ per complete iteration.
We start from  $g^{(0)}$ and call one \emph{complete iteration} the update $g^{(s-1)}\to f^{(s)}\to g^{(s)}$ obtained by applying the two half-steps \cref{eq:ugw_sinkhorn_halfsteps} in this order; $\pi^{(s)}$ denotes the plan encoded by $(f^{(s)},g^{(s)})$, and
\[
    D_s:=\max\bigl\{\norm{f^{(s)}-f^\star}_\infty,\norm{g^{(s)}-g^\star}_\infty\bigr\}
    \quad(s\ge1),
    \qquad
    D_0:=\norm{g^{(0)}-g^\star}_\infty
\]
measure the distance to the optimal potentials. Applying
\cref{eq:ugw_halfstep_contraction} $2s-1$ times and then $2s$ times,
\begin{equation}
    \label{eq:ugw_Ds_decay}
    \norm{f^{(s)}-f^\star}_\infty\le\Bigl(\frac{\rho}{\rho+\varepsilon}\Bigr)^{2s-1}D_0,
    \qquad
    \norm{g^{(s)}-g^\star}_\infty\le\Bigl(\frac{\rho}{\rho+\varepsilon}\Bigr)^{2s}D_0,
\end{equation}
and therefore
\begin{equation}
    \label{eq:ugw_Ds_decay_max}
    D_s\le\Bigl(\frac{\rho}{\rho+\varepsilon}\Bigr)^{2s-1}D_0
    \qquad\text{for all }s\ge1 .
\end{equation}

\begin{lemma}[Sinkhorn as a primal UOT oracle]
\label{lem:ugw_sinkhorn_oracle}
Assume $\rho_1=\rho_2=\rho$, fix $\Gamma$ and $\Xi$ with $\norm\Xi\le R$, and let $M_\Gamma(R)$ be as in \cref{prop:ugw_klkl_regularities}. Then:
\begin{enumerate}
    \item[\rm(i)] for every $s\ge1$ such that $D_s\le\varepsilon/2$,
    \begin{equation}
        \label{eq:ugw_sinkhorn_primal_rate}
        P_{\Gamma,\Xi}(\pi^{(s)})-\mathrm{UOT}(c_{\Gamma,\Xi})
        \le\frac{8\,(\varepsilon+2\rho)\,M_\Gamma(R)}{\varepsilon^2}\,D_s^2 ;
    \end{equation}
    \item[\rm(ii)] consequently, with the constant
    \begin{equation}
        \label{eq:ugw_CS_def}
        C_{\mathrm S}:=\max\Bigl\{1,\ \frac{2D_0}{\varepsilon},\ \frac{8(\varepsilon+2\rho)M_\Gamma(R)}{\varepsilon^2}D_0^2\Bigr\},
    \end{equation}
    the plan $\pi^{(s)}$ is $\delta_{\Reg\OT}$-accurate for every $\delta_{\Reg\OT}\le1$ as soon as
    \begin{equation}
        \label{eq:ugw_sinkhorn_complexity}
        s\ge N_{\mathrm S}(\delta_{\Reg\OT}):=\Bigl\lceil\frac12+\frac{\log(C_{\mathrm S}/\delta_{\Reg\OT})}{2\log(1+\varepsilon/\rho)}\Bigr\rceil ;
    \end{equation}
    moreover
    \begin{equation}
        \label{eq:ugw_sinkhorn_complexity_bound}
        N_{\mathrm S}(\delta_{\Reg\OT})\le\frac32+\frac{\rho+\varepsilon}{2\varepsilon}\log\frac{C_{\mathrm S}}{\delta_{\Reg\OT}} .
    \end{equation}
\end{enumerate}
\end{lemma}

\begin{proof}
Let $\pi^\star:=\pi_{\Gamma,\Xi}$. Since
$\nabla P_{\Gamma,\Xi}(\pi^\star)=0$, the primal gap is the Bregman divergence
of $P_{\Gamma,\Xi}$ at $\pi^\star$. Using the linearity of the marginal maps,
\begin{equation}
\label{eq:ugw_primal_gap_identity}
P_{\Gamma,\Xi}(\pi)-\mathrm{UOT}(c_{\Gamma,\Xi})
=
\varepsilon\KL(\pi\Vert\pi^\star)
+\rho\KL(\pi_1\Vert\pi^\star_1)
+\rho\KL(\pi_2\Vert\pi^\star_2).
\end{equation}

\emph{(i).}
By the Gibbs representation of $\pi^{(s)}$ and $\pi^\star$,
\begin{align*}
\left|
\log\frac{\pi^{(s)}_{ij}}{\pi^\star_{ij}}
\right|
&\le
\frac{2D_s}{\varepsilon}
=:\theta .
\end{align*}
The same bound holds for the ratios of both marginals, since each marginal
ratio is a weighted average of the corresponding entrywise ratios. If
$\theta\le1$, the elementary bound

$$
    u\log u-u+1\le2\theta^2,
    \qquad
    e^{-\theta}\le u\le e^\theta,
$$

therefore gives
\begin{align*}
\KL(\pi^{(s)}\Vert\pi^\star),
\quad
\KL(\pi^{(s)}_1\Vert\pi^\star_1),
\quad
\KL(\pi^{(s)}_2\Vert\pi^\star_2)
\le
2\theta^2\norm{\pi^\star}_1.
\end{align*}
Using $\norm{\pi^\star}_1\le M_\Gamma(R)$ and
$\theta=2D_s/\varepsilon$ in \cref{eq:ugw_primal_gap_identity} yields

$$
    P_{\Gamma,\Xi}(\pi^{(s)})-\mathrm{UOT}(c_{\Gamma,\Xi})
    \le
    \frac{8(\varepsilon+2\rho)M_\Gamma(R)}{\varepsilon^2}D_s^2.
$$

\emph{(ii).}
Set $q:=\frac{\rho}{\rho+\varepsilon}$.
By \cref{eq:ugw_Ds_decay_max},
$$
    D_s\le q^{2s-1}D_0.
$$
If $s\geq N_{\mathrm{S}}(\delta_{\Reg\OT})$, then

$$
    q^{2s-1}
    \le
    \frac{\delta_{\Reg\OT}}{C_{\mathrm S}}.
$$

By the definition of $C_{\mathrm S}$ and $\delta_{\Reg\OT}\le1$, this implies
both $D_s\le\varepsilon/2$ and
\begin{align*}
P_{\Gamma,\Xi}(\pi^{(s)})-\mathrm{UOT}(c_{\Gamma,\Xi})
&\le
\frac{8(\varepsilon+2\rho)M_\Gamma(R)}{\varepsilon^2}
D_0^2q^{2(2s-1)} \
&\le
C_{\mathrm S}q^{2s-1}
\le
\delta_{\Reg\OT}.
\end{align*}
This proves \cref{eq:ugw_sinkhorn_complexity}. Finally,
\cref{eq:ugw_sinkhorn_complexity_bound} follows from
$\lceil x\rceil\le x+1$ and

$$
    \log(1+x)\ge\frac{x}{1+x}.
$$

\end{proof}

\paragraph{End-to-end rates.}
For the UGW problem considered here, hypotheses \textbf{(S1)}--\textbf{(S2)}
hold since $\Reg_{\mathrm{UOT}}$ is proper, convex and lower semicontinuous,
while both the GW energy and $\Reg_{\mathrm{UOT}}$ are nonnegative.
Fix $\delta\in(0,1]$ and bounds
$\overline R,\overline\Gamma,\overline G>0$, and set
\begin{align}
    \overline M
    &:=
    \sup_{\norm{\Gamma}\le\overline\Gamma}
    M_\Gamma(3\overline R+1),
    \label{eq:ugw_uniform_M}\\
    \overline L
    &:=
    1+\frac{R_A^2\overline M}{\varepsilon},
    \qquad
    \overline C_L
    :=
    2\overline L^2+2\overline L+1,
    \label{eq:ugw_uniform_L}\\
    \overline\delta_{\Reg\OT}
    &:=
    \frac{\delta^2}{8\overline C_L},
    \label{eq:ugw_uniform_delta}\\
    \overline N_\Xi
    &:=
    \left\lceil
    \overline L
    \log\max\left\{
        1,
        \frac{32\overline L^2\overline R^2}{\delta^2}
    \right\}
    \right\rceil.
    \label{eq:ugw_uniform_NXi}
\end{align}
Let $\overline N_{\mathrm S}$ denote the bound of
\cref{lem:ugw_sinkhorn_oracle}(ii) at accuracy
$\overline\delta_{\Reg\OT}$, with $M_\Gamma(R)$ and $D_0$ replaced by their
uniform bounds $\overline M$ and $\overline G$.

\begin{theorem}[End-to-end complexity for entropic UGW]
\label{thm:ugw_end_to_end}
Assume $\rho_1=\rho_2=\rho$, and run \cref{alg:ugw} with
$\tau=1/\overline L$ and UOT accuracy
$\overline\delta_{\Reg\OT}$, solving each UOT subproblem by Sinkhorn with at
most $\overline N_{\mathrm S}$ complete iterations. Suppose that, throughout
the run,
\begin{equation}
\label{eq:ugw_uniform_inner_region}
    \norm{\Gamma_t}\le\overline\Gamma,
    \qquad
    \max\left\{
        \norm{\Xi^\star_{\Gamma_t}},
        \norm{\Xi_{t,k}}
    \right\}\le\overline R,
    \qquad
    \norm{g^{(0)}_{t,k}-g^\star_{t,k}}_\infty\le\overline G.
\end{equation}
Then every inner loop terminates after at most $\overline N_\Xi$ iterations,
and the algorithm terminates within $T_\delta$ outer iterations with a
$(\delta,\delta^2/4)$-critical plan. In particular,
\begin{equation}
\label{eq:ugw_nested_end_to_end}
    N_{\mathrm{Sink}}^{\mathrm{tot}}
    \le
    T_\delta(\overline N_\Xi+1)\overline N_{\mathrm S}.
\end{equation}

If $\rho$ is fixed and
$\overline R,\overline\Gamma,\overline G=O(1)$ as
$\varepsilon,\delta\to0$, then
\begin{equation}
\label{eq:ugw_nested_final_rate}
    N_{\mathrm{Sink}}^{\mathrm{tot}}
    =
    \widetilde O\left(
        \frac{
            \Delta^{\mathrm{out}}(\rho+\varepsilon)
        }{
            \delta^2\varepsilon^2
        }
    \right).
\end{equation}
If moreover $\NN=\MM=n$ and $A,B$ take values in fixed-dimensional
Euclidean spaces, the total arithmetic complexity is
\begin{equation}
    \widetilde O\left(
        \frac{
            n^2\Delta^{\mathrm{out}}(\rho+\varepsilon)
        }{
            \delta^2\varepsilon^2
        }
    \right).
\end{equation}
\end{theorem}

\begin{proof}
Under the assumed bounds,
\cref{prop:ugw_klkl_regularities}(iii) gives uniformly
\begin{equation}
    -\overline L I
    \preceq
    \nabla^2 h_{\Gamma_t}(\Xi)
    \preceq
    -I
    \qquad
    \text{for all }
    \norm{\Xi}\le3\overline R+1.
\end{equation}
Hence the local form of \textbf{(H$_L$)} holds with
$L=\overline L$.

Moreover,
\cref{lem:ugw_sinkhorn_oracle} and the bounds
\begin{equation}
    M_{\Gamma_t}(\overline R)
    \le
    \overline M,
    \qquad
    \norm{g^{(0)}_{t,k}-g^\star_{t,k}}_\infty
    \le
    \overline G
\end{equation}
show that every UOT call reaches accuracy
$\overline\delta_{\Reg\OT}$ within
$\overline N_{\mathrm S}$ complete Sinkhorn iterations.

It remains to control the initial inner gap. Since
$\nabla h_{\Gamma_t}(\Xi^\star_{\Gamma_t})=0$ and
$h_{\Gamma_t}$ is $\overline L$-smooth on the segment joining
$\Xi_{t,0}$ and $\Xi^\star_{\Gamma_t}$,
\begin{align}
    h_{\Gamma_t}(\Xi^\star_{\Gamma_t})
    -
    h_{\Gamma_t}(\Xi_{t,0})
    &\le
    \frac{\overline L}{2}
    \norm{\Xi_{t,0}-\Xi^\star_{\Gamma_t}}^2\\
    &\le
    2\overline L\,\overline R^2.
\end{align}
Therefore \cref{cor:unified_oracle_complexity}, with oracle accuracy
\[
    \frac{\delta^2}{8\overline C_L}
    =
    \overline\delta_{\Reg\OT},
\]
gives
\[
    K_t\le\overline N_\Xi,
\]
while \cref{thm:unified_outer} gives termination within $T_\delta$ outer
iterations and a $(\delta,\delta^2/4)$-critical plan.

There are at most $\overline N_\Xi+1$ UOT calls per outer iteration and
at most $\overline N_{\mathrm S}$ complete Sinkhorn iterations per call,
which proves \cref{eq:ugw_nested_end_to_end}.

Finally, for fixed $\rho$ and bounded
$\overline R,\overline\Gamma,\overline G$,
\begin{align}
    \overline M
    &=O(1),
    &
    \overline L
    &=O(\varepsilon^{-1}),
    &
    \overline N_\Xi
    &=
    \widetilde{O}(\varepsilon^{-1}),
    &
    \overline N_{\mathrm S}
    &=
    \widetilde{O}\left(
        \frac{\rho+\varepsilon}{\varepsilon}
    \right).
\end{align}
Together with
\[
    T_\delta
    =
    O\left(
        \frac{\Delta^{\mathrm{out}}}{\delta^2}
    \right),
\]
this yields \cref{eq:ugw_nested_final_rate}. When $\NN=\MM=n$, each
Sinkhorn iteration, cost evaluation and moment computation costs $O(n^2)$,
which gives the stated arithmetic complexity.
\end{proof}

\paragraph{Accelerated convergence rate.}

Previously, we established the convergence rate of our solver for the entropic unbalanced Gromov--Wasserstein problem. As discussed in \cref{sub:acceleraded_versions}, when $h_\Gamma$ is smooth, the inner optimization can instead be accelerated using inexact first-order methods. This is particularly effective when the underlying $\Reg\OT$ solver converges geometrically, so that achieving accuracy $\delta$ requires only $O(\log(1/\delta))$ iterations. This is precisely the case here for generalized Sinkhorn, and we make the resulting accelerated complexity explicit below. Note also that, if the marginal penalties in $\Reg$ are chosen to be $\chi^2$ rather than $\KL$, one may similarly use the accelerated-gradient method of \cite{pmlr-v336-genans26a}, which also enjoys geometric convergence.

\begin{corollary}[Accelerated complexity for entropic UGW]
\label{cor:ugw_nesterov_acceleration}
Assume $\rho_1=\rho_2=\rho$ and the uniform bounds of
\cref{thm:ugw_end_to_end}, and let
\[
    \overline L
    =
    1+\frac{R_A^2\overline M}{\varepsilon}.
\]
At each outer iteration $t$, apply the accelerated inexact-gradient method
of \citet{devolder2013strongly} to
$\psi_t:=-h_{\Gamma_t}$ on the ball
$\{\norm{\Xi}\le\overline R\}$, using at every oracle query a plan returned
by generalized Sinkhorn with UOT accuracy
\begin{equation}
\label{eq:ugw_acc_uot_accuracy}
    \overline\delta_{\Reg\OT}^{\mathrm{acc}}
    :=
    \frac{\delta^2}
    {104\,\overline L^{3/2}}.
\end{equation}
Then each inner problem reaches $
    \operatorname{Cert}_{\Gamma_t}(\Xi,\widehat\pi)
    \le\frac{\delta^2}{4}$
after
$ \widetilde O(\sqrt{\overline L})$ 
UOT calls.

Consequently, the outer algorithm terminates within
\[
    T_\delta
    =
    O\left(
        \frac{\Delta^{\mathrm{out}}}{\delta^2}
    \right)
\]
iterations with a $(\delta,\delta^2/4)$-critical plan. If $\rho$ is fixed
and $\overline R,\overline\Gamma,\overline G=O(1)$ as
$\varepsilon,\delta\to0$, then
$\overline L=O(\varepsilon^{-1})$ and
\begin{equation}
\label{eq:ugw_nested_acc_rate}
    N_{\mathrm{Sink,acc}}^{\mathrm{tot}}
    =
    \widetilde O\left(
        \frac{
            \Delta^{\mathrm{out}}(\rho+\varepsilon)
        }{
            \delta^2\varepsilon^{3/2}
        }
    \right).
\end{equation}
\end{corollary}

\begin{proof}
A $\delta_{\Reg\OT}$-accurate UOT solve defines a
$(2\delta_{\Reg\OT},2\overline L,1)$ inexact oracle for
$\psi_t$. Hence the accelerated bound of
\citet{devolder2013strongly} gives
\[
    \psi_t(y_k)-\psi_t^\star
    \le
    C_0
    \exp\left(
        -\frac{k}{2\sqrt{2\overline L}}
    \right)
    +
    2(1+\sqrt{2\overline L})\,
    \delta_{\Reg\OT}.
\]
Using \cref{eq:unified_smooth_certificate_bound} at $y_k$,
\begin{align*}
    \operatorname{Cert}_{\Gamma_t}(y_k,\widehat\pi)
    &\le
    2\overline L
    C_0
    \exp\left(
        -\frac{k}{2\sqrt{2\overline L}}
    \right)\\
    &\quad+
    \bigl[
        4\overline L(1+\sqrt{2\overline L})
        +2\overline L+1
    \bigr]
    \delta_{\Reg\OT}.
\end{align*}
Since $\overline L\ge1$,
\[
    4\overline L(1+\sqrt{2\overline L})
    +2\overline L+1
    \le
    13\overline L^{3/2}.
\]
With \cref{eq:ugw_acc_uot_accuracy}, the second term is therefore at most
$\delta^2/8$. Taking
$k=\widetilde O(\sqrt{\overline L})$ makes the first term at most
$\delta^2/8$, proving the required certificate.

Finally, \cref{lem:ugw_sinkhorn_oracle} gives
\[
    N_{\mathrm S}(\zeta)
    =
    \widetilde O\left(
        \frac{\rho+\varepsilon}{\varepsilon}
    \right),
\]
because its dependence on $\zeta$ is logarithmic. Thus the smaller UOT
accuracy in \cref{eq:ugw_acc_uot_accuracy} affects only logarithmic factors.
Multiplying the
$\widetilde O(\sqrt{\overline L})$ UOT calls per outer iteration by
$T_\delta$ and using
$\overline L=O(\varepsilon^{-1})$ yields
\cref{eq:ugw_nested_acc_rate}.
\end{proof}

\end{document}